\documentclass[oneside,reqno]{amsart}
\usepackage{amssymb}
\usepackage{eucal}
\usepackage{lscape}
\usepackage{cite}
\usepackage{caption}
\usepackage{float}
\newtheorem{theorem}{Theorem}

\newtheorem{lemma}[theorem]{Lemma}
\theoremstyle{remark}
\newtheorem*{remark}{Remark}
\newcommand{\myStrut}[1]{\parbox{0.11 pt}{\rule{0 ex}{#1 ex}}}

\newcommand{\te}[2]{\parbox[b]{#1 cm}{ \centering #2}}
\newcommand{\be}{\begin{equation}}
\newcommand{\ee}{\end{equation}}
\newcommand{\bc}{\begin{center}}
\newcommand{\ec}{\end{center}}
\newcommand{\bmt}{\begin{pmatrix}}
\newcommand{\emt}{\end{pmatrix}}
\newcommand{\bsmt}{\left(\begin{smallmatrix}}
\newcommand{\esmt}{\end{smallmatrix}\right)}

\newcommand{\vpu}[1]{^{\vphantom{#1}}}
\newcommand{\graw}{\mathcal{G}}
\newcommand{\gl}{\bar{\mathcal{G}}}
\newcommand{\gEl}{\mathcal{G}}
\newcommand{\glt}{\bar{\mathcal{G}}_{\mathbb{E}}}
\newcommand{\gc}{\bar{\mathcal{G}}_c}
\newcommand{\gEc}{\mathcal{G}_c}
\newcommand{\gZ}{\bar{\mathcal{G}}_0}
\newcommand{\gEZ}{\mathcal{G}_0}
\newcommand{\gZb}{\bar{\mathcal{G}}_1}
\newcommand{\gEZb}{\mathcal{G}_1}
\newcommand{\hgr}{K}
\newcommand{\eC}{\mathbb{E}_c}
\newcommand{\el}{\bar{\mathbb{E}}}
\newcommand{\elt}{\mathbb{E}}
\newcommand{\eZ}{\mathbb{E}_0}
\newcommand{\eZb}{\mathbb{E}_1}
\newcommand{\stbF}[1]{S_{#1}}
\newcommand{\stb}[1]{S^{\bar{\mathbb{E}}}_{#1}}
\newcommand{\tstb}[1]{\tilde{S}_{#1}}
\newcommand{\bstb}[1]{\bar{S}_{#1}}
\newcommand{\nbstb}[1]{\bar{N}_{#1}}
\newcommand{\cbstb}[1]{\bar{C}_{#1}}

\DeclareMathOperator{\re}{Re}
\newcommand{\qv}{\mathbf{q}}
\newcommand{\rv}{\mathbf{r}}

\DeclareMathOperator{\Tr}{Tr}
\DeclareMathOperator{\Det}{Det}
\DeclareMathOperator{\SL}{SL}
\DeclareMathOperator{\ESL}{ESL}
\DeclareMathOperator{\GL}{GL}
\DeclareMathOperator{\C}{CL}
\DeclareMathOperator{\EC}{ECL}

\newcommand{\ECb}{\EC_{\el}}
\newcommand{\spn}[2]{\multicolumn{1}{|r|}{\te{#1}{$\phantom{x}$} #2.}}

\newcommand{\simple}{simple }
\newcommand{\Simple}{Simple }
\newcommand{\lfl}{\mathcal{L}_{\mathbb{F}}}
\newcommand{\lgr}{\mathcal{L}_{G}}

\begin{document}
\begin{titlepage}
\begin{center}
\bfseries  Galois Automorphisms of a Symmetric Measurement
\end{center}
\vspace{0.3 cm}
\begin{center} D.M.~Appleby
\\
\emph{Perimeter Institute for Theoretical Physics,  Waterloo, Ontario N2L 2Y5, Canada} \\  \vspace{0.05 in} \emph{and}  \vspace{0.05 in}  \\
\emph{Stellenbosch Institute for Advanced Study (STIAS), Wallenberg Research Centre at Stellenbosch University, Marais Street, Stellenbosch 7600, South Africa}
 \end{center}

\vspace{0.2 cm}

\begin{center}
Hulya Yadsan-Appleby \\
 \emph{Dept. of  Physics and Astronomy,  University College London, Gower St., London WC1E 6BT, UK}
 \end{center}

\vspace{0.2 cm }

\begin{center}
Gerhard Zauner 
\\
\emph{ Zieglergasse 45/13-14, 1070 Vienna, Austria}
\end{center}
\vspace{0.5 cm}

\begin{center} \textbf{Abstract}

\vspace{0.1cm}

\parbox{12 cm }{
Symmetric Informationally Complete Positive Operator Valued Measures (usually referred to as SIC-POVMs or simply as SICS) have been constructed in every  dimension $\le 67$.    However, a proof that they exist in every finite dimension has yet to be constructed.  In this paper we examine the Galois group of SICs covariant with respect to the Weyl-Heisenberg group (or WH SICs as we  refer to them).  The great majority (though not all) of the known examples are of this type.    Scott and Grassl have noted that every known exact WH SIC is expressible in radicals (except for dimension $3$ which is exceptional in this and several other respects), which means that the corresponding Galois group is solvable.  They have also calculated the Galois group for most known exact examples.  The purpose of this paper is to take the analysis of Scott and Grassl further.  We first prove a number of theorems regarding the structure of the Galois group and the relation between it and the extended Clifford group.  We then examine the Galois group for the known exact fiducials and on the basis of this we propose a list of nine conjectures concerning its structure.   These conjectures represent a considerable strengthening of the theorems we have actually been able to prove.  Finally we generalize the concept of an anti-unitary to the concept of a $g$-unitary, and show that every WH SIC fiducial is an eigenvector of a family of $g$-unitaries (apart from dimension 3).

}
\vspace{0.35 cm}
\parbox{12 cm }{}
\end{center}
\end{titlepage}

{\allowdisplaybreaks
\section{Introduction}
\label{sec:intro}
Symmetric informationally complete positive operator valued measures, or SIC-POVMs as they are often called, or SICs as we shall call them in this paper, have attracted much interest in recent years~\cite{Appleby2EtAl,Appleby09a,Zauner,RenesEtAl,Grassl04,Appleby1,Grassl05,Klappenecker,Flammia,Appleby07a,Appleby07b,Kibler,Hughston,Grassl08,Khatirinejad,Godsil,Bengtsson08,Fickus,Appleby09b,ScottGrassl,Zhu10a,Zhu10b,Fillipov,Appleby11,Bengtsson1,Appleby3EtAl,Ballester,Scott06,Petz,Zhu11,Englert,Rehacek,Renes05,Durt06,Durt08,Fuchs03,Fuchs04,Kim,Bodmann07,Oreshkov,Bengtsson11b,Howard,Bodmann08,Hermann,Balan,Fuchs09,Fuchs10,Fuchs11,Medendorp,KalevEtAl,Tabia1}.  SICs are important practically, with applications to quantum tomography~\cite{Ballester,Scott06,Petz,Zhu11}, quantum cryptography~\cite{Englert,Rehacek,Renes05,Durt06,Durt08}, quantum communication~\cite{Fuchs03,Fuchs04,Kim,Bodmann07,Oreshkov}, Kochen-Specker arguments~\cite{Bengtsson11b}, high precision radar~\cite{Howard,Bodmann08,Hermann} and speech recognition~\cite{Balan} (the last two applications being classical). They are  important from a  mathematical point of view, as giving insight into the geometrical structure of quantum state space~\cite{BengtssonBook}.  They  have a connection with the theory of Lie Algebras~\cite{Appleby11} and, at least in dimension $3$, a connection with the theory of elliptic curves~\cite{Hughston,Bengtsson1}.  They are important from a foundational point of view, playing a central role in the qbist approach to the interpretation of quantum mechanics~\cite{Fuchs09,Fuchs10,Fuchs11}.   SICs have been realized experimentally~\cite{Durt06,Medendorp}, and further experiments have been proposed~\cite{KalevEtAl,Tabia1}.  
They have been calculated numerically for every dimension $\le 67$ and exact solutions have been constructed for dimensions $2$--$16$, $19$, $24$, $28$,  $35$ and $48$ (see ref.~\cite{ScottGrassl} for a comprehensive listing of solutions known in 2010 and refs.~\cite{Appleby2EtAl,Appleby3EtAl} for the exact solutions in dimensions $16$ and $28$).  There are therefore grounds for conjecturing that SICs actually exist in every finite dimension.  However, in spite of strenuous efforts by many investigators over a period of more than 10 years the question is still undecided.

In this paper we will exclusively be concerned with SICs which are covariant with respect to the Weyl-Heisenberg group~\cite{Zauner,RenesEtAl,Grassl04,Appleby1}---WH SICs as we will call them.  The vast majority of known SICs are of this type; moreover it has been shown~\cite{Zhu10a} that in prime dimensions WH SICs are the only possible group covariant SICs.

The purpose of this paper is to examine the Galois symmetries~\cite{Jacobson,Lang,Roman,Edwards,Stewart,Milne} of WH SICs, and the relation between these and the Extended Clifford Group.    Galois theory as it applies to a \emph{finite} field plays an important role in the theory of mutually unbiassed bases~\cite{Wootters87}. So to avert a possible misunderstanding let us stress that what concerns us here is, not that, but Galois theory as it applies to an \emph{infinite}  field.  That the Galois symmetries might be interesting is suggested by the fact that the known exact SIC fiducial projectors are all expressible in radicals (apart from dimension $3$ which is exceptional  in several respects, not the least of which is that it admits WH SICs for which the projectors involve  numbers which are not only not expressible in radicals, but are even transcendental).  It follows that the Galois group of the field generated by the matrix elements of the fiducial projectors must be solvable~\cite{Jacobson,Lang,Roman,Edwards,Stewart,Milne} (dimension $3$ excepted).   To see why that is surprising consider the defining equations for a WH SIC fiducial vector.  These are a system of multivariate quartic polynomials in the real and imaginary parts of the components.  The standard method for solving such a system of equations is to construct a Gr\"{o}bner basis~\cite{Grassl04,Grassl05,Grassl08,ScottGrassl,Cox}, which reduces the problem to successively solving a series of univariate polynomial equations.  However it usually turns out (and does turn out in the case of the SIC problem) that the equations in a single variable are  of much higher degree than the multivariate equations with which one started.  Since the generic univariate polynomial equation of degree $>4$ is not solvable  in radicals, one would not \emph{a priori} expect the solutions of a system of multivariate quartic equations to be solvable in radicals.   The fact that the known exact SICs nonetheless are expressible in radicals (dimension 3 excepted) is therefore  a very interesting \emph{datum} concerning them.  Another striking fact about the known exact solutions is that, although the solutions themselves are very complicated, often running to several pages of print-out, the field generators out of which the solutions are constructed are comparatively simple.  If one looks at the expressions for the field generators in Tables~\ref{tble:artGenerators} and~\ref{tble:bGenerators} in Appendix~\ref{sec:tbles} it will be seen that none of them involve integers $>10^4$, and most of them only involve integers much smaller than that.  Also, aside from the generators $=\cos \frac{\pi}{d}$ or $\sin \frac{\pi}{d}$ (the $t$ generators in Table~\ref{tble:artGenerators}; here, as everywhere else in this paper,  $d$  denotes the dimension), they are all constructed by taking square roots or cube roots.  This is unexpected, since the  polynomials  in the Gr\"{o}bner basis tend to have (in the words of Scott and Grassl~\cite{ScottGrassl}) ``large degrees and huge coefficients'' ($\sim 10^{200}$ for dimension $11$).  Finally, the Galois groups tabulated by Scott and Grassl all have a normal series~\cite{Jacobson,Lang,Roman,Edwards,Stewart,Milne}  of the form
\be
\langle e \rangle \vartriangleleft H \vartriangleleft G
\label{eq:SGNormalSeries}
\ee
where $e$ is the identity, $H$ and $G/H$ are Abelian, and the index $[G:H] = 2$.  It could be said that the Galois group is not only solvable:  it is a solvable group of a particularly simple kind.   We should remark that for dimension $14$ the above statement does not hold for the group as it is  given by Scott and Grassl.  However, there is an error in ref.~\cite{ScottGrassl} at this point (as originally noticed by Jon Yard~\cite{Yard}); when that error is corrected it is found that the statement does hold.  We should also remark that the field extension for which Scott and Grassl construct the Galois group does not always coincide with  the field extension considered in this paper.  However the Galois groups we calculate all have a normal series of the kind just described.

The purpose of this paper is to take the analysis of Scott and Grassl further.  After a brief review of relevant background material in Sections~\ref{sec:cliff} and~\ref{sec:Gal} we begin, in Section~\ref{sec:GenRes}, by proving a number of general results concerning the Galois group of a WH SIC.  In Section~\ref{sec:GalCorrEEbar} we describe the subfields corresponding to the subgroups introduced in Section~\ref{sec:GenRes}.  In Section~\ref{sec:StructureThm} we prove a   structure theorem for the Galois group of an arbitrary WH SIC.     In Section~\ref{sec:conjectures} we examine the Galois groups of the known exact fiducials for $d\ge 4$ and on the basis of this propose nine conjectures which strengthen the statements proved in Sections~\ref{sec:GenRes}--\ref{sec:StructureThm}.  In Section~\ref{sec:fldGen} we  make some further observations concerning the known exact fiducials for $d\ge 4$.  In Section~\ref{sec:Gunit} we introduce the concept of a $g$-unitary.  This is a generalization of the concept of an anti-unitary in which the role of complex conjugation is played by an arbitrary Galois automorphism which commutes with complex conjugation.  We show that every WH  fiducial projector is a joint eigenprojector of a group of $g$-unitaries (except for dimension $3$).  Finally, in Section~\ref{sec:d2and3}, we discuss the  fiducials in dimension $d=2$ and $3$ which are,  in several respects, exceptional. 

We conclude this introductory section by drawing the reader's attention to two particularly striking points to emerge from our analysis. 
Let $\Pi$ be a fiducial projector.  Let $\elt$  be the smallest normal extension of $\mathbb{Q}$ containing the standard basis matrix elements of $\Pi$ and $\tau = -e^{\frac{i \pi  }{d}}$.  $\elt$ only depends on the extended Clifford group orbit to which $\Pi$ belongs.  
It turns out that if $d>3$ then $\elt$ is an  Abelian extension of the real quadratic field
\be
\mathbb{Q}\left(\sqrt{(d-3)(d+1)}\right)
\ee
for all $27$ extended Clifford group orbits on which an exact fiducial is known.
The fact that the $\elt$ is an Abelian extension of a  quadratic field is already suggested by the result of Scott and Grassl mentioned above.  Our result goes further than that since it suggests that there is a simple formula for a generator of the quadratic field.

The second point to which we wish to draw the reader's attention is another conjecture regarding the structure of the Galois group.  The $27$ orbits  for $d>3$  on which an exact fiducial has been calculated comprise $9$ doublets (pairs of orbits connected by a Galois automorphism) and $17$ singlets (orbits closed under the action of the Galois group).   Let $\gEl$ be the Galois group of $\elt$ over $\mathbb{Q}$.  We find that in all 27 cases $\gEl$ has the normal series
\be
\langle e\rangle  \vartriangleleft  \gEZ \vartriangleleft \gEl
\ee
with $\gEZ$, $\gEl/\gEZ$ both Abelian and
\begin{align}
\gEZ & \cong \mathbb{Z}_2 \oplus \hgr
\\
\gEl/\gEZ & \cong 
\begin{cases}
\mathbb{Z}_2 \qquad & \text{orbit is a singlet}
\\
\mathbb{Z}_2\oplus\mathbb{Z}_2 \qquad & \text{orbit is one of a doublet}
\end{cases}
\end{align}
The group $\hgr$ is defined in Section~\ref{sec:conjectures}.  For now suffice it to say that $\hgr$ is given in terms of the subgroup of the extended Clifford group having $\Pi$ as an eigenprojector.  Of course, a knowledge of  the quotient groups in a normal series does not amount to a knowledge of the group itself.  However, if the conjecture this result suggests held generally it would be fair to say that we could tell a  great deal about the structure of the Galois group just from a knowledge of the extended Clifford symmetries, without first having to  calculate an exact fiducial.

\section{Extended Clifford Group and WH SICs}
\label{sec:cliff}
The purpose of this section is to briefly review some  relevant facts concerning the extended Clifford group and WH SICs.  For proofs and further details see refs.~\cite{Appleby1,ScottGrassl}.

Let $|0\rangle, \dots |d-1\rangle$ be the standard basis in dimension $d$.  Define the operators $X$ and $Z$ by
\begin{align}
X |r\rangle & = |r+1\rangle
\\
Z |r\rangle & = \tau^{2r} |r\rangle
\end{align}
where addition of ket-labels is mod $d$ and $\tau= -e^{\frac{\pi i}{d}}$.  The Weyl-Heisenberg displacement operators are then defined by
\be
D_{\mathbf{p}} = \tau^{p_1 p_2} X^{p_1} Z^{p_2}
\ee 
where $\mathbf{p} = \left(\begin{smallmatrix} p_1 \\ p_2 \end{smallmatrix}\right)$. We have
\begin{align}
D_{\mathbf{p}}D_{\mathbf{q}} & = \tau^{\langle \mathbf{p},\mathbf{q}\rangle} D_{\mathbf{p}+\mathbf{q}} & D^{\dagger}_{\mathbf{p}} & = D\vpu{\dagger}_{-\mathbf{p}}
\end{align}
for all $\mathbf{p}$, $\mathbf{q}$, where $\langle \cdot, \cdot \rangle$ is the symplectic form
\be
\langle \mathbf{p},\mathbf{q}\rangle=p_2q_1 -p_1q_2
\ee
The $D_{\mathbf{p}}$ are a basis for operator space.  An arbitrary operator $A$ has the expansion
\begin{align}
A &= \sum_{\mathbf{p}\in\mathbb{Z}^2_d} A_{\mathbf{p}} D_{\mathbf{p}} & A_{\mathbf{p}}& =\frac{1}{d} \Tr(A D^{\dagger}_{\mathbf{p}})
\label{eq:OpExpn}
\end{align}

Let
\be
\bar{d} = \begin{cases} d \qquad & \text{$d$ odd} \\ 2 d \qquad & \text{$d$ even} \end{cases}
\ee
We then define the group $\ESL(2,\mathbb{Z}_{\bar{d}})$ to be the set of all matrices 
\be
F = \bmt \alpha & \beta \\ \gamma & \delta \emt
\ee
such that $\alpha$, $\beta$, $\gamma$, $\delta \in \mathbb{Z}_{\bar{d}}$ and $\Det F= \pm 1$ (mod $\bar{d}$).  We call matrices whose determinant $=1$ (respectively $-1$)  symplectic matrices (respectively anti-symplectic matrices).   The set of symplectic matrices form the symplectic group, denoted  $\SL(2,\mathbb{Z}_{\bar{d}})$.  We refer to $\ESL(2,\mathbb{Z}_{\bar{d}})$ as the extended symplectic group.

We associate to each symplectic matrix (respectively anti-symmetric matrix) $F$ a unitary (respectively anti-unitary) $U_F$ as follows:
\begin{enumerate}
\item If
\be
F = \bmt \alpha & \beta \\ \gamma & \delta \emt \in \SL(2,\mathbb{Z}_{\bar{d}})
\ee
is a prime matrix~\cite{Appleby1} (\emph{i.e.}\ if $\beta$ is relatively prime to $\bar{d}$) we define
\be
U_F = \frac{1}{\sqrt{d}} \sum_{r,s} \tau^{\beta^{-1}(\delta r^2 -2rs+\alpha s^2)} |r\rangle\langle s|
\label{eq:UFprimeFFormula}
\ee
where $\beta^{-1}$ is the unique element of $\mathbb{Z}_{\bar{d}}$ such that $\beta \beta^{-1} = 1$ (mod $\bar{d}$) (\emph{not} the rational number $1/\beta$). 
\item For each non-prime matrix $F\in\SL(2,\mathbb{Z}_{\bar{d}})$ we make a fixed but arbitrary choice of a pair of prime matrices $F_1$, $F_2$ such that
\be
F = F_1 F_2
\ee
(it is shown in ref.~\cite{Appleby1} that such a pair always exists).  We then define
\be
U_F = U_{F_1} U_{F_2}
\label{eq:UFnonprimeFFormula}
\ee
\item For 
\be
J = \bmt 1 & 0 \\ 0 & -1\emt \in \ESL(2,\mathbb{Z}_{\bar{d}})
\ee
we define $U_J$ to be the anti-unitary which acts by complex conjugation in the standard basis:
\be
U_J |\psi\rangle = \sum_r \langle r | \psi \rangle^{*} | r\rangle
\ee
for all $|\psi\rangle$.
\item For every other anti-symplectic matrix $F \in \ESL(2,\mathbb{Z}_{\bar{d}})$ we define
\be
U_F = U_{FJ} U_J
\ee
\end{enumerate}
It is shown in ref.~\cite{Appleby1} that with these definitions
\be
U\vpu{\dagger}_F D\vpu{\dagger}_{\mathbf{p}} U^{\dagger}_F = D\vpu{\dagger}_{F\mathbf{p}}
\ee
for all $F$, $\mathbf{p}$, 
and that the map $F\to U_F$ is a projective representation of $\ESL(2,\mathbb{Z}_{\bar{d}})$, so that
\begin{align}
U_{F_1 F_2} &\dot{=} U_{F_1} U_{F_2} & &\text{and} &U^{\dagger}_{F} &\dot{=} U\vpu{\dagger}_{F^{-1}}
\end{align}
for all $F_1,F_2, F$ (where $\dot{=}$ means ``equals up to a phase'').

We  define
\begin{align}
\C(d) & = \{e^{i\xi}D_{\mathbf{p}} U_F \colon \xi\in\mathbb{R},\;\mathbf{p}\in \mathbb{Z}_{d}\times \mathbb{Z}_{d}, \; F \in \SL(2,\mathbb{Z}_{\bar{d}})\}
\label{eq:CLdef}
\\
\intertext{and}
\EC(d) & = \{e^{i\xi}D_{\mathbf{p}} U_F \colon \xi \in \mathbb{R}, \; \mathbf{p}\in \mathbb{Z}_{d}\times \mathbb{Z}_{d}, \; F \in \ESL(2,\mathbb{Z}_{\bar{d}})\}
\label{eq:ECLdef}
\end{align}
We refer to $\C(d)$ as the Clifford group and to $\EC(d)$ as the extended Clifford group.

In the sequel we will need the following fact, proved in ref.~\cite{Appleby1}:
\begin{theorem} 
\label{eq:kerCliffThm}
If $d$ is odd  $D_{\mathbf{p}} U_F \dot{=} I$ if and only if $\mathbf{p} = 0$ (mod $d$) and $F = 0$ (mod $d$).
If $d$ is even $D_{\mathbf{p}} U_F \dot{=} I$ if and only if
\begin{align}
\mathbf{p} & = \bmt \frac{sd}{2} \\ \frac{td}{2}\emt  & \text{ (mod $d$)}
\\
 F & = \bmt 1+r d & s d\\ t d & 1+r d\emt & \text{ (mod $2d$)}
\end{align}
for suitable integers $r,s,t = 0$ or $1$.
\end{theorem}

A SIC in dimension $d$ is specified by a set of $d^2$ rank-$1$ projectors $\Pi_1, \dots , \Pi_{d^2}$ with the property
\be
\Tr(\Pi_r \Pi_s) = \frac{d\delta_{rs}+1}{d+1}
\ee
It  follows from this property that the $\Pi_r$ are linearly independent and that
\be
\sum_{r} \Pi_r = dI
\ee
implying that the operators $\frac{1}{d}\Pi_1, \dots , \frac{1}{d} \Pi_{d^2}$ are an informationally complete POVM.  This POVM is the SIC.  

In a WH SIC the projectors are obtained by acting on a single projector $\Pi$ (called the fiducial projector) with displacement  operators, according to the prescription
\begin{align}
\Pi\vpu{\dagger}_{\mathbf{p}} &= D\vpu{\dagger}_{\mathbf{p}} \Pi D^{\dagger}_{\mathbf{p}} 
\end{align}
One has
\be
\Pi = |\psi\rangle \langle \psi | 
\ee
for some normalized vector $|\psi\rangle$.  In many situations it is convenient to work with  $|\psi\rangle$ rather than $\Pi$.  However, if one is specifically interested in the Galois symmetries there are some disadvantages to that.  The vector $|\psi\rangle$ is defined in terms of the matrix elements of $\Pi$ by
\be
|\psi\rangle = \frac{e^{i\theta}}{\sqrt{\langle 0 | \Pi |  0 \rangle}} \sum_{r} \langle r | \Pi | 0 \rangle |  r\rangle 
\ee
where $e^{i\theta}$ is an arbitrary phase.  It can (and often does) happen that the factor $\sqrt{\langle 0 | \Pi | 0\rangle}$ is not in the field generated by the components of $\Pi$.  So the field generated by the components of $|\psi\rangle$ will typically be larger than the field generated by the matrix elements of $\Pi$.  Moreover the arbitrariness of the  phase means that the field generated by  the components of $|\psi\rangle$ is not unique.  For these reasons we will work exclusively in terms of the fiducial projector. 

The unitaries and anti-unitaries in $\EC(d)$ take WH fiducials to WH fiducials.  We refer to
\be
\{ U\Pi U^{\dagger} \colon U \in \EC(d)\}
\ee
as the $\EC(d)$ orbit of $\Pi$ or, when there is no risk of confusion, simply as the orbit of $\Pi$.

Let
\be
\stbF{\Pi} = \{U \in \EC(d)\colon U \Pi U^{\dagger} = \Pi\}
\label{eq:stbFDef}
\ee
be the stability group of $\Pi$.
 It is a striking, though so far unexplained, fact~\cite{Appleby1,ScottGrassl} that for every known WH fiducial (exact or numerical) $\stbF{\Pi}$ contains a unitary
 \be
D_{\mathbf{p}} U_F
 \ee
with $F\neq I$ and $\Tr F = -1$ (mod $d$)  (if $d\neq 3$ it is enough to require that $\Tr F =-1$ (mod $d$) as it is then automatic that $F\neq I$). 
$D_{\mathbf{p}} U_F$ is order 3 up to a phase.  We refer to $D_{\mathbf{p}} U_F$ as a canonical order 3 unitary, and to $F$ as a canonical trace $-1$ symplectic.  The  majority~\cite{ScottGrassl} of orbits for which a fiducial has been calculated (exact or numerical) contain a fiducial stabilized by the canonical order 3 unitary $U_{F_z}$, where
\be
F_z = \bmt 0 & d-1 \\ d+1 & d-1\emt
\ee
We will refer to such orbits as type $z$.  When $d=9k+3$, for some positive integer $k$, orbits have also been calculated~\cite{ScottGrassl} which contain a fiducial stabilized by the canonical order 3 unitary $U_{F_a}$, where
\be
F_a =\bmt 1 & d+3 \\ d+3k & d-2 \emt
\ee
We refer to such orbits as type $a$.  Observe that, although the unitaries $U_{F_z}$, $U_{F_a}$ are always order $3$ up to a phase, the symplectics $F_z$ and $F_a$ are order 6 when $d$ is even. Observe, also, that $F_z$ and $F_a$ are not conjugate~\cite{Flammia,ScottGrassl}, so it cannot happen that an orbit is both type $z$ and type $a$.

Another significant fact about the orbits for which a WH fiducial has been calculated (exact or numerical) is that there is always a fiducial on the orbit for which $\stbF{\Pi}$ consists entirely of unitaries/anti-unitaries of the form $e^{i\xi}U_F$.  We will say that the stability group of such a fiducial is displacement-free.  For such fiducials we define
\be
\tstb{\Pi} = \{ F \in \ESL(2,\mathbb{Z}_{\bar{d}}) \colon U_F \in \stbF{\Pi}\}
\label{eq:tstbDef}
\ee

   If 
\begin{enumerate}
\item $\stbF{\Pi}$ contains a canonical order $3$ unitary
\item $\stbF{\Pi}$ is displacement-free
\end{enumerate}
 we say that $\Pi$ is \simple.  \Simple  fiducials will play an important role in the sequel.

\section{Galois Theory}
\label{sec:Gal}
A reader who is unfamiliar with Galois theory may consult (for example) one or more of refs.~\cite{Jacobson,Lang,Roman,Edwards,Stewart,Milne}.  The brief review which follows is not intended as a substitute for these texts.  It may, however, be useful to the reader who has studied Galois theory in the past, but would like to be reminded of the basic facts.  It may also be useful to a reader who has no previous knowledge of Galois theory, as indicating the minimum he or she will need to know in order to understand this paper.

In the following we only consider fields which contain $\mathbb{Q}$ as a subfield and are themselves subfields of $\mathbb{C}$.  Let $\mathbb{F}$ be such a field, and let $u$ be any complex number.  We  define $\mathbb{F}(u)$ to be the smallest field containing $\mathbb{F}$ and $u$.  We say that $u$  is the generator of the extension field $\mathbb{F}(u)$.  Suppose that $u$ is algebraic over $\mathbb{F}$ (\emph{i.e.} a root of some polynomial with coefficients in $\mathbb{F}$).  Let 
\be
p(x) = x^n + c_{n-1} x^{n-1} + \dots + c_0
\ee 
be the minimal polynomial of $u$ over $\mathbb{F}$:  \emph{i.e.}\ the polynomial of lowest degree with all coefficients $\in \mathbb{F}$ and  leading coefficient  $=1$ having $u$ as a root.  Then
\be
\mathbb{F} (u)= \{ a_0 + a_1 u + \dots  + a_{n-1} u^{n-1} \colon a_0, \dots a_{n-1} \in \mathbb{F}\}
\ee 
One sees from this that $\mathbb{F}(u)$ can be regarded as a vector space over $\mathbb{F}$.  Its dimension $n$  is called the degree of $\mathbb{F}(u)$ over $\mathbb{F}$, denoted $[\mathbb{F}(u) : \mathbb{F}]$. 

 We can generalize this to the case of several generators  by defining $\mathbb{F}(u_1,\dots, u_m)$ to be the smallest field containing the field $\mathbb{F}$ and the numbers $u_1,\dots, u_n$.   One can build up the field by appending the generators successively:
\be
\mathbb{F}(u_1,\dots, u_m) = ( \dots (( \mathbf{F}(u_1))(u_2) \dots )(u_m)
\ee
The order in which one takes the generators when performing this construction is irrelevant.  The degree of the extension (its dimension regarded as a vector space over $\mathbb{F}$)  is then given by
\begin{align}
[\mathbb{F}(u_1,\dots,u_m) : \mathbb{F}] &= [\mathbb{F}(u_1,\dots,u_m) : \mathbb{F}(u_1,\dots , u_{m-1}] 
\nonumber
\\
&\hspace{0.5 in} \times  [\mathbb{F}(u_1,\dots,u_{m-1}) : \mathbb{F}(u_1,\dots , u_{m-2})] 
\nonumber
\\
&\hspace{1 in} \times \dots 
\nonumber
\\
& \hspace{1.5 in} \times  [\mathbb{F}(u_1) : \mathbb{F}]
\end{align}
A field extension obtained in this way, by appending a finite set of algebraic numbers, is said to be finitely generated and algebraic.  In the following all field extensions will be assumed without comment to be of this kind. Let us note that algebraic extensions have the property that, not only the generators, but every element is algebraic over $\mathbb{F}$.

Let $\mathbb{F} \subseteq \mathbb{F}' \subseteq \mathbb{F}''$ be a tower of field extensions.  Then
\be
[\mathbb{F}'' : \mathbb{F}] = [\mathbb{F}'' : \mathbb{F}'] [\mathbb{F}': \mathbb{F}] 
\ee

Let $\mathbb{F}(u_1,\dots, u_m)$ be an extension of $\mathbb{F}$.  Let $p_j(x)$ be the minimal polynomial of $u_j$ over $\mathbb{F}$.   The extension is said to be normal if $p_j(x)$ fully factorizes over $\mathbb{F}(u_1,\dots, u_m)$ for all $j$, in which case we write $\mathbb{F} \vartriangleleft  \mathbb{F}(u_1,\dots, u_m)$.  Thus $\mathbb{Q}(2^{\frac{1}{4}})$ is not a normal extension of $\mathbb{Q}$ (because $x^4-2$ does not fully factorize over $\mathbb{Q}(2^{\frac{1}{4}})$).   On the other hand $\mathbb{Q}(2^{\frac{1}{4}},i)$ is a normal extension of $\mathbb{Q}$.  It should be observed that the relation of being a normal extension is not transitive:  if $\mathbb{F} \vartriangleleft \mathbb{F}'$ and $\mathbb{F}' \vartriangleleft \mathbb{F}''$   it is not always the case  that $\mathbb{F} \vartriangleleft \mathbb{F}''$.  On the other hand if $\mathbb{F} \vartriangleleft \mathbb{F}''$ it is always the case that $\mathbb{F}' \vartriangleleft \mathbb{F}''$ (though not always the case that $\mathbb{F} \vartriangleleft \mathbb{F}'$) for every intermediate field $\mathbb{F}'$ (\emph{i.e.}\ for every field $\mathbb{F}'$ such that $\mathbb{F} \subseteq \mathbb{F}' \subseteq \mathbb{F}''$).  

Suppose $\mathbb{F} \vartriangleleft \mathbb{F}'$ and let $f(x)$ be any irreducible polynomial with coefficients in $\mathbb{F}$.  Then if $\mathbb{F}'$ contains one of the roots of $f(x)$ it contains them all.  One consequence of this which will be important to us is that a normal extension of the rationals is closed under complex conjugation (since  $z^{*}$ is  a root of the minimal polynomial of $z$ over the rationals for every algebraic number $z$).

Let $\mathbb{F}'$ be an extension of $\mathbb{F}$.  A Galois automorphism of $\mathbb{F}'$ over $\mathbb{F}$ is a bijective mapping $g\colon \mathbb{F}' \to \mathbb{F}'$ with the properties
\begin{enumerate}
\item $g(z_1+z_2) = g(z_1)+ g(z_2)$ and $g(z_1z_2) = g(z_1)g(z_2)$ for all $z_1$, $z_2\in\mathbb{F}'$.
\item $g(z) = z$ for all $z \in \mathbb{F}$.
\end{enumerate}
The group of all such automorphisms is called the Galois group of $\mathbb{F}'$ over $\mathbb{F}$ and is denoted $\graw_{\mathbb{F}}(\mathbb{F}')$.  An arbitrary element of $\mathbb{F}' = \mathbb{F}(u_1,\dots,u_m)$ can be written as a linear combination of monomials of the form $u_1^{k_1} \dots u_m^{k_m}$ with coefficients in $\mathbb{F}$.  It follows that a Galois automorphism $g$ is fully specified by the numbers $g(u_1), \dots, g(u_m)$.  Also  $g(u_j)$ is a root of the minimal polynomial of $u_j$ over $\mathbb{F}$  for all $j$.  Using these two facts it is straightforward to determine the Galois group in any given case (conceptually straightforward, that is---the actual calculations can be tedious, and for the ones in this paper we relied heavily on the computer algebra package \emph{Magma}).   If $\mathbb{F}'$ is a normal extension of $\mathbb{F}$ then the order of $\graw_{\mathbb{F}}(\mathbb{F}')$ is  $[\mathbb{F}' : \mathbb{F}]$.  

Let $\lgr$ be the set of  subgroups of $\graw_{\mathbb{F}}(\mathbb{F}')$ and let $\lfl$ of be the set of intermediate fields (\emph{i.e.} the set of fields $\mathbb{K}$ such that $\mathbb{F}\subseteq \mathbb{K} \subseteq \mathbb{F}'$).  $\lgr$ and $\lfl$ are lattices when partially ordered by set inclusion.  For each $H\in\lgr$ define $\mathbb{K}_H\in\lfl$ by
\begin{align}
\mathbb{K}_H& = \{z \in \mathbb{F}' \colon h(z) = z \quad \forall h \in H\}
\\
\intertext{$\mathbb{K}_H$ is called the fixed field of $H$.  For each $\mathbb{K}\in\lfl$ define $H_{\mathbb{K}}\in\lgr$ by}
H_{\mathbb{K}} & = \graw_{\mathbb{K}}(\mathbb{F}')
\end{align}
The maps $H \to \mathbb{K}_H$ and $\mathbb{K}\to H_{\mathbb{K}}$ are order reversing:
\begin{align}
H_1 &\subseteq H_2 & &\implies & \mathbb{K}_{H_1} &\supseteq \mathbb{K}_{H_2}
\\
\mathbb{K}_1 & \subseteq \mathbb{K}_2 && \implies & H_{\mathbb{K}_1} &\supseteq H_{\mathbb{K}_2}
\end{align}

Now suppose that $\mathbb{F}'$ is a normal extension of $\mathbb{F}$.  Then the maps  $H \to \mathbb{K}_H$ and $\mathbb{K}\to H_{\mathbb{K}}$ are mutually inverse bijections:
\begin{align}
\mathbb{K}_{H_{\mathbb{K}}} & = \mathbb{K}  & H_{\mathbb{K}_H} & = H
\end{align}
for all $\mathbb{K}\in\lfl$ and $H\in \lgr$. This pairing between subgroups and intermediate fields for a normal extension is called the Galois correspondence.  The Galois correspondence takes normal subgroups to normal extensions of the base field, and vice versa:
\begin{align}
H & \vartriangleleft \graw_{\mathbb{F}}(\mathbb{F}') &&   \Longleftrightarrow & \mathbb{F} \vartriangleleft \mathbb{K}_H 
\end{align}
(note that $\mathbb{K}_H \vartriangleleft \mathbb{F}'$ whether $H$ is normal or not). If $H\vartriangleleft \graw_{\mathbb{F}} (\mathbb{F}')$ each $g\in\graw_{\mathbb{F}}(\mathbb{F}')$ restricted to $\mathbb{K}_H$ is a Galois automorphism of $\mathbb{K}_H$ over $\mathbb{F}$.  This gives us a natural isomorphism
\be
\graw_{\mathbb{F}} (\mathbb{K}_H) \cong \graw_{\mathbb{F}}(\mathbb{F}')/H
\ee

A number is said to be expressible in radicals if it can be constructed recursively from the rationals by taking sums, products, and rational powers.  If the generators of $\mathbb{F} = \mathbb{Q}(u_1,\dots,u_m)$ are expressible in radicals then so is every element of $\mathbb{F}$.  The necessary and sufficient condition for this to be true is that  $\graw_{\mathbb{Q}} (\mathbb{F})$    be a soluble group---\emph{i.e.}\ that  $\graw_{\mathbb{Q}} (\mathbb{F})$  have a normal series
\be
H_0 = \langle e\rangle \vartriangleleft H_1 \vartriangleleft \dots \vartriangleleft H_n = \graw_{\mathbb{Q}}(\mathbb{F})
\ee
such that the quotient groups $H_1/H_0, H_2/H_1 , \dots H_n/H_{n-1}$ are all Abelian.

\section{Action of the Galois group on a fiducial projector}
\label{sec:GenRes}
In this Section we discuss the action of the Galois group on a fiducial projector.
Let $\Pi$ be a WH SIC fiducial, and let $\elt$ be as defined in the Introduction.  $\elt$ is not guaranteed to include $\sqrt{d}$   (though it turns out that in practice it often does---see Table~\ref{tble:fields} in Appendix~\ref{sec:tbles}). Consequently it is not guaranteed to include the standard basis matrix elements of the unitaries $U_F$.  This can be inconvenient.  We will therefore work with the  field
\be  
\el = \elt(\sqrt{d})
\ee
and only apply our results to the field $\elt$ at the end of our calculations.  $\el$ is a normal extension of $\mathbb{Q}$ (this is immediate if $\sqrt{d}\in\elt$ as then $\el = \elt$; otherwise it is a consequence of the fact that   the minimal polynomial of $\sqrt{d}$ over $\elt$  factorizes completely over $\el$).

With a slight abuse of notation we will say that a linear operator $\Gamma$ is in a field $\mathbb{F}$ and write $\Gamma \in \mathbb{F}$ if its standard basis matrix elements all $\in \mathbb{F}$.  We  say that an anti-linear operator $\Gamma$ is in $\mathbb{F}$ and write $\Gamma \in \mathbb{F}$ if the linear operator $\Gamma U_J \in \mathbb{F}$.  
The arbitrary phase in Eq.~(\ref{eq:ECLdef}) means that the elements of $\EC(d)$ do not all $\in \el$.  However it is easily seen that $D_{\mathbf{p}}U_F\in \el$ for all $\mathbf{p}$, $F$.  So  every element of $\EC(d)$ is equal to an element of 
\begin{align}
\ECb(d) & = \{ U \in \EC(d) \colon U \in \el\}
\end{align} 
up to a phase.
The product of any two elements of $\ECb(d)$ is in $\ECb(d)$.   Also the fact that $\el$ is a normal extension of the rationals, and therefore closed under complex conjugation, means that $U^{\dagger}\in\ECb(d)$ whenever $U\in\ECb(d)$.  So $\ECb(d)$ is a group.

 It is  straightforward to verify that if $\Pi'$ is another fiducial on the same orbit as $\Pi$ then $\Pi'\in\el$.  So $\el$ only depends on the orbit, and not on the particular fiducial used to define it.
Let $\gl=\graw_{\mathbb{Q}}(\el)$ be the Galois group of $\el$ over $\mathbb{Q}$.   For each $g\in \gl$ and linear operator $\Gamma   \in \el$ define
\begin{align}
 g(\Gamma) &=\sum_{r,s}g\bigl(\langle r | \Gamma | s \rangle\bigr) |r\rangle \langle s|
\\
\intertext{If $\Gamma$ is an anti-linear operator $\in \el$  define}
g(\Gamma) &= g(\Gamma U_J)U_J
\end{align}
Let $g_c$ be complex conjugation (guaranteed to be in $\gl$ because $\el$ is a normal extension of $\mathbb{Q}$) and let $\gc\subseteq \gl$ be the centralizer of $g_c$.
If $g\in\gc$ we  have
\begin{align}
g(\Gamma_1\Gamma_2) &= g(\Gamma_1)g(\Gamma_2)  &  g(\Gamma^{\dagger}) & = g(\Gamma)^{\dagger}
\label{eq:gProdAndHConj}
\end{align}
for all $\Gamma_1,\Gamma_2, \Gamma\in \el$, linear or anti-linear (note that in general neither of these statements is true if $g\notin \gc$; however if $\Gamma_1$ is linear then  $g(\Gamma_1\Gamma_2) = g(\Gamma_1)g(\Gamma_2)$ even when $g \notin \gc$). 

For each $g\in\gl$ let $k_g$ be the unique integer in the range $0 \le k_g < \bar{d}$ which is relatively prime to $\bar{d}$ and is such that
\be
g(\tau) = \tau^{k_g}
\ee
Let
\be
H_g = \bmt 1 & 0 \\ 0 & k_g \emt \in \GL(2,\mathbb{Z}_{\bar{d}})
\ee
We then have
\begin{theorem} 
\label{thm1}
For all  $g\in\gl$, $\mathbf{p} \in \mathbb{Z}^2_{d}$,  $F\in \ESL(2,\mathbb{Z}_{\bar{d}})$
\begin{align}
g(D_{\mathbf{p}}) =&\  D_{H_g\mathbf{p}}
&
g(U_F) \ \dot{=} & \  U_{H_g F H^{-1}_g}
\end{align}
where $\dot{=}$ means ``equals up to a phase''.
\end{theorem}
\begin{remark}
This is a generalization of the formulae
\begin{align}
g_c(D_{\mathbf{p}}) =& D_{J\mathbf{p}}
&
g_c(U_{F}) \dot{=}& U_{JFJ}
\end{align}
\end{remark}
\begin{proof}
The first statement is immediate.  To prove the second suppose, to begin with, that $F=\left( \begin{smallmatrix}\alpha & \beta \\ \gamma & \delta \end{smallmatrix}\right)$
is a prime matrix $\in \SL(2,\mathbb{Z}_{\bar{d}})$.  Then 
\be
U_F = \frac{1}{\sqrt{d}} \sum_{r,s} \tau^{\beta^{-1}(\delta r^2 -2rs+\alpha s^2)} |r\rangle\langle s|
\ee
 Consequently
\be
g\bigl(U_F\bigr)  = \pm \frac{1}{\sqrt{d}}\sum_{r,s} \tau^{k_g \beta^{-1}(\delta r^2 -2rs+\alpha s^2)} |r\rangle\langle s|
\  \dot{=} \ U_{H_g F H^{-1}_g}
\ee
If, on the other hand, $F$ is a non-prime matrix $\in \SL(2,\mathbb{Z}_{\bar{d}})$ choose prime symplectic matrices $F_1$, $F_2$ such that $U_F=U_{F_1}U_{F_2}$.  Then
\begin{align}
g(U_F) = U_{H_g F_1 H_{g}^{-1}}U_{H_g F_2 H_{g}^{-1}} \
\dot{=}
\
U_{H_g F H_{g}^{-1}}
\end{align}
Finally, if  $F$ is anti-symplectic 
\begin{align}
g(U_F)  =  \ g(U_{FJ}) U_J \ 
\dot{=}  \ U_{H_{g}FJ H^{-1}_{g}} U_{\vphantom{H^{-1}_g}J}
=  \ U_{H_{g}F H^{-1}_{g}}
\end{align}
where in the last step we used the fact that $H^{-1}_g$ commutes with $J$.
\end{proof}
\begin{theorem}
If $\Pi'$ is a fiducial projector $\in \el$ and $g \in \gc$ then $g(\Pi')$ is also a fiducial projector.
\end{theorem}
\begin{proof} It follows from Eq.~(\ref{eq:gProdAndHConj}) that $g(\Pi')$ is a rank-1 projector and

\be
\Bigl| \Tr\bigl(g(\Pi') D_{\mathbf{p}}\bigr)\Bigr|^2 
=
g\left( \Bigr| \Tr\bigl(\Pi' D_{H_g^{-1} \mathbf{p}}\bigr)\Bigr|^2\right)
=
\begin{cases}
1 & \text{if $\mathbf{p} = \boldsymbol{0}$ (mod $d$)}
\\
\frac{1}{d+1} & \text{otherwise}
\end{cases}
\ee
\end{proof}

Now let $\gZ$ be the set of $g\in \gc$ with the property that $g(\Pi)$ is on the same orbit of the extended Clifford group as $\Pi$.  It is easily seen that $\gZ$ depends only the orbit and not on the particular fiducial $\Pi$.  For each $g\in\gZ$ choose a fixed (anti-)unitary $U_g\in \ECb(d)$ such that
\be
g(\Pi) = U\vpu{\dagger}_g \Pi U_g^{\dagger}
\ee
There is some arbitrariness in this choice.  Let
\be
\stb{\Pi} = \stbF{\Pi} \cap \ECb(d)
\ee
where $\stbF{\Pi}$ is defined in Eq.~(\ref{eq:stbFDef}).  Then we can replace $U_g$ with any other element of the coset $U_g \stb{\Pi}$.
\begin{theorem}
\label{thm3}
$\gZ$ is a subgroup of $\gc$.  We have
\begin{align}
U_{g_1g_2} &\in g_1(U_{g_2}) U_{g_1} \stb{\Pi}
&
U_{g^{-1}} &\in g^{-1}(U_g^{\dagger}) \stb{\Pi}
\label{eq:thm2B}
\end{align}
for all $g_1,g_2,g\in \gZ$.
\end{theorem}
\begin{proof}
For all $g_1,g_2\in\gZ$
\begin{align}
g_1g_2(\Pi) = g_1\bigl( U\vpu{\dagger}_{g_2} \Pi U^{\dagger}_{g_2}\bigr) = g_1(U_{g_2})U_{g_1}\Pi U^{\dagger}_{g_1} \bigl(g_1(U_{g_2})\bigr)^{\dagger}
\end{align}
So  $g_1g_2\in\gZ$ and $U_{g_1g_2} \in g_1(U_{g_2}) U_{g_1} \stb{\Pi}$.
Also for  all $g\in\gZ$ 
\begin{align}
\Pi & = g^{-1} \left( U_g \Pi U^{\dagger}_g \right) = g^{-1}(U_g) g^{-1}(\Pi) \bigl(g^{-1}(U_g)\bigr)^{\dagger}
\\
\intertext{implying}
g^{-1}(\Pi) & = g^{-1}(U_g^{\dagger}) \Pi \bigl(g^{-1}(U_g^{\dagger})\bigr)^{\dagger}
\end{align}
So $g^{-1}\in \gZ$ and $U_{g^{-1}} \in g^{-1}(U_g^{\dagger}) \stb{\Pi}$.
\end{proof}

We conclude this section by showing that in the case of a \simple fiducial (as defined in Section~\ref{sec:cliff}) $U_g$ can be chosen to be of a particularly simple form:

\begin{theorem}
\label{thm4}
Let $\Pi$ be a \simple fiducial.  
\begin{enumerate}
\item If $d\neq 0$ (mod $3$)  then, for all $g\in \gZ$, we can choose $U_g$ to be of the form
\be
U_g = U_{F_g}
\ee
for some $F_g \in \ESL(2,\mathbb{Z}_{\bar{d}})$.
\item If $d=0$ (mod $3$) then, for all $g\in\gZ$, we can choose $U_g$ to be of the form
\be
U_g = D_{\qv_g} U_{F_g}
\ee
for some $F_g \in \ESL(2,\mathbb{Z}_{\bar{d}})$ and unique $\qv_g \in \mathbb{Z}_d \times \mathbb{Z}_d$   such that
\be
\qv_g = \boldsymbol{0} \quad (\text{mod $d/3$})
\ee
\end{enumerate}
\end{theorem}
\begin{proof}Choose $\qv'_g\in \mathbb{Z}_d\times \mathbb{Z}_d$, $F'_g\in \ESL(2,\mathbb{Z}_{\bar{d}})$ such that
\be
g(\Pi) = D_{\qv'_g} U_{F'_g} \Pi U^{\dagger}_{F'_g} D^{\dagger}_{\qv'_g}
\ee
Then
\be
g(\stb{\Pi}) = \stb{g(\Pi)}= D_{\qv'_g} U_{F'_g} \stb{\Pi} U^{\dagger}_{F'_g} D^{\dagger}_{\qv'_g}
\label{eq:gOnstb}
\ee
Now let $L\in \tstb{\Pi}$ be a canonical trace $-1$ symplectic (see Eq.~(\ref{eq:tstbDef}) for the definition of $\tstb{\Pi}$).  It follows from Eq.~(\ref{eq:gOnstb}) that there exists $M\in\tstb{\Pi}$ such that 
\be
U_{H_{g}LH^{-1}_g}\dot{=}D\vpu{\dagger}_{\qv'_g} U\vpu{\dagger}_{F'_g} U\vpu{\dagger}_{M} U^{\dagger}_{F'_g} D^{\dagger}_{\qv'_g}
\ee
After rearranging this becomes
\be
D_{H_g L H_g^{-1} \qv'_g - \qv'_g} U_{H_gLH^{-1}_g F'_g M^{-1} F'^{-1}_g} \dot{=} I
\label{eq:thm4C}
\ee
In view of Theorem~\ref{eq:kerCliffThm}
\be
\tilde{L}\qv'_g -\qv'_g  
=
\begin{cases}
\text{$\boldsymbol{0}$ (mod $d$)} \qquad & \text{if $d$ is odd}
\\
\text{$\boldsymbol{0}$ (mod $\frac{d}{2}$)} \qquad & \text{if $d$ is even}
\end{cases}
\ee
where $\tilde{L} = H\vpu{-1}_g L H^{-1}_g$.  
Since $\Tr\tilde{L}= \Tr(L)= -1$ we can write
\be
\tilde{L} = \bmt \alpha & \beta \\ \gamma & -\alpha -1 \emt
\ee
for some $\alpha, \beta, \gamma$ such that
\be
\alpha^2 + \alpha + \beta \gamma +1 = 0 
\ee
It is easily verified that
\be
\bmt -(\alpha+2) & -\beta \\ -\gamma & (\alpha-1) \emt (\tilde{L} - I) = 3 I 
\ee
implying
\be
3\qv'_g  
=
\begin{cases}
\text{$\boldsymbol{0}$ (mod $d$)} \qquad & \text{if $d$ is odd}
\\
\text{$\boldsymbol{0}$ (mod $\frac{d}{2}$)} \qquad & \text{if $d$ is even}
\end{cases}
\label{eq:thm4A}
\ee

\vspace{0.1 in}

\noindent \textbf{Case 1.} Suppose $d$ is not divisible by 3.  Then Eq.~(\ref{eq:thm4A}) implies
\be
\qv'_g  =
\begin{cases}
\text{$\boldsymbol{0}$ (mod $d$)} \qquad & \text{if $d$ is odd}
\\
\text{$\boldsymbol{0}$ (mod $\frac{d}{2}$)} \qquad & \text{if $d$ is even}
\end{cases} 
\ee
If $d$ is odd the claim is now immediate.    If, on the other hand, $d$ is even we have
\be
\qv'_g  = \bmt \frac{sd}{2} \\ \frac{td}{2} \emt
\ee
with $s,t = 0$ or $1$.  In view of Theorem~\ref{eq:kerCliffThm}
\be
D_{\qv'_g} U_{F'_g} \dot{=} U_{F_g\vphantom{F'_g}}
\ee
where
\be
F_g = \bmt 1 & sd\\ td & 1\emt F'_g
\label{eq:thm4B}
\ee
So one can choose $U_g$ to be of the stated form in this case also. 

\vspace{0.1 in}

\noindent \textbf{Case 2.} Suppose $d$ is divisible by 3. Then Eq.~(\ref{eq:thm4A}) implies
\be
\qv'_g  =
\begin{cases}
\text{$\boldsymbol{0}$ (mod $\frac{d}{3}$)} \qquad & \text{if $d$ is odd}
\\
\text{$\boldsymbol{0}$ (mod $\frac{d}{6}$)} \qquad & \text{if $d$ is even}
\end{cases} 
\ee
If $d$ is odd we  can set $\qv_g = \qv'_g$, $F_g = F'_g$.     If, on the other hand, $d$ is even we have
\be
\qv'_g =  \bmt \frac{(2j+3s)d}{6} \\ \frac{(2k+3t)d}{6} \emt \text{  (mod $d$)}
\ee
where $j,k=0,1$ or $2$ and $s,t=0$ or $1$.  In view of Theorem~\ref{eq:kerCliffThm}
\be
D_{\qv'_g} U_{F'_g} \dot{=} D_{\qv_g} U_{F_g}
\ee
where
\begin{align}
\qv_g & = \bmt \frac{jd}{3} \\ \frac{kd}{3} \emt 
&
F_g &= \bmt 1 & sd\\ td & 1\emt F'_g
\end{align}
So one can choose $U_g$ to be of the stated form in this case also. 

It remains to show that $\qv_g$ is unique.    Let $D_{\qv'_g}U_{F'_g}$  be another possible choice for $U_g$, with $\qv'_g$ satisfying the requirements stated in the theorem.  Then
\be
D_{F^{-1}_g(\qv'_g - \qv_g)} U_{F^{-1}_g F'_g} \in \stb{\Pi}
\ee
Since $\stb{\Pi}$ is displacement-free this means, in view of   Theorem~\ref{eq:kerCliffThm},
\be
\qv'_g - \qv_g=
\begin{cases}
\text{$\boldsymbol{0}$ (mod $d$)} \qquad & \text{if $d$ is odd}
\\
\text{$\boldsymbol{0}$ (mod $\frac{d}{2}$)} \qquad & \text{if $d$ is even}
\end{cases}
\ee
which is easily seen to imply $\qv'_g = \qv_g$.
\end{proof}
\section{Subfields of $\el$}
\label{sec:GalCorrEEbar}
Applying the Galois correspondence to the increasing series of subgroups
\be
\langle e \rangle \subseteq \gZ \subseteq \gc \subseteq \gl
\ee
we obtain the decreasing series of subfields
\be
\el \supseteq \eZ \supseteq \eC \supseteq \mathbb{Q}
\ee
where $\eZ$, $\eC$ are the fixed fields of $\gZ$, $\gc$ respectively.   Applying the inverse map we find
\begin{align}
\gZ & =\graw_{\eZ}(\el) & \gc & = \graw_{\eC}(\el)
\end{align}

In the Introduction we also introduced the field $\elt$.  $\elt$ is not guaranteed to contain $\sqrt{d}$ (although in practice it often does---see Table~\ref{tble:fields} in Appendix~\ref{sec:tbles}).   Consequently it is not guaranteed to contain all the unitaries and anti-unitaries $U_F$.   However it does contain every fiducial on the same orbit as $\Pi$ and so, like $\el$, it only depends on the orbit and not on the particular fiducial used to define it.  To see this observe that it follows from Eqs.~(\ref{eq:UFprimeFFormula}) and~(\ref{eq:UFnonprimeFFormula}) that every symplectic unitary $U_F$ can be written 
\be
U_F = \frac{1}{d^{\frac{r}{2}}} M_F
\ee
where the matrix $M_F\in\elt$ and where $r = 1$ if $F$ is  a prime matrix, and $=0$ otherwise.  So if $\Pi'$ is on the same orbit of $\EC(d)$ as $\Pi$ we can write
\be
\Pi'
=
\begin{cases}
\frac{1}{d^r} D\vpu{\dagger}_{\mathbf{p}} M\vpu{\dagger}_{F} \Pi M^{\dagger}_F D^{\dagger}_{\mathbf{p}} \qquad & \det F = 1
\\
\frac{1}{d^r} D\vpu{\dagger}_{\mathbf{p}} M\vpu{\dagger}_{FJ} g_c(\Pi) M^{\dagger}_{JF}D^{\dagger}_{\mathbf{p}} \qquad & \det F = -1
\end{cases}
\ee
for some $\mathbf{p},F$, where $\Pi^{\mathrm{T}}$ is the transpose of $\Pi$. Since $\elt$ contains the displacement operators and is  a normal extension of $\mathbb{Q}$  we conclude that $\Pi' \in \elt$.

Let $\glt = \graw_{\elt}(\el)$.  Since each $g \in \glt$ fixes $\Pi$ it must also $\in \gZ$.  So
\be
\langle e \rangle \subseteq\glt  \subseteq \gZ \subseteq \gc \subseteq \gl
\ee
It turns out that this is a normal series with Abelian quotient groups for all the orbits for which an exact fiducial has been calculated with $d>3$, and also for the three orbits for $d\le 3$ analyzed in Section~\ref{sec:d2and3}.  However, we have not been able to prove that that must always be the case.  In view of  the Galois correspondence we also have
\be
\el \supseteq \elt \supseteq \eZ \supseteq \eC \supseteq \mathbb{Q}
\ee 

Now let $\gEl= \graw_{\mathbb{Q}}(\elt)$.  The fact that $\el$   is a normal extension of $\elt $ means
\be
\gEl = \{ g|_{\elt} \colon g \in \gl\} \cong \gl/\glt
\ee
where  $g|_{\elt}$ is  the restriction of $g$ to $\elt$. Applying  the Galois  correspondence to the series 
\be
\elt \supseteq \eZ \supseteq \eC \supseteq \mathbb{Q}
\ee
we obtain the series
\be
\langle e \rangle \subseteq \gEZ \subseteq \gEc \subseteq \gEl
\ee
where
\begin{align}
\gEZ & = \graw_{\eZ}(\elt) = \{ g|_{\elt} \colon g \in \gZ\} \cong \gZ/\glt
\\
\gEc & = \graw_{\eC}(\elt) =  \{ g|_{\elt} \colon g \in \gc\} \cong \gc/\glt
\end{align}
\section{Structure theorem for the subgroups $\gEZ$ and $\gZ$}
\label{sec:StructureThm}
In the Introduction we stated two conjectures regarding the structure of the Galois group.  Unfortunately, we have not been able to prove these conjectures.  However, we have been able to prove a weaker statement, which holds for any orbit on which there exists a \simple fiducial.  That is the subject to which we now turn.  
In this section it will always be assumed that $\Pi$ is a simple fiducial.

We begin by considering  the action of $\gZ$ on the overlaps
\be
\chi_{\mathbf{p}} = \Tr(\Pi D_{\mathbf{p}})
\ee
Define
\be
\bstb{\Pi} = \{ (\Det F) F \colon F\in \tstb{\Pi}\}
\ee
We  have
\begin{lemma} 
\label{lem:sbarpiStab} Let $G\in \ESL(2,\mathbb{Z}_{\bar{d}})$.  Then the following statements are equivalent
\begin{enumerate}
\item $G \in \bstb{\Pi}$
\item $\chi_{G\mathbf{p}} = \chi_{\mathbf{p}}$ for all $\mathbf{p} \in \mathbb{Z}^2_{\bar{d}}$
\end{enumerate}
\end{lemma}
\begin{proof}
Suppose $G\in\bstb{\Pi}$.  
If $\Det G = 1$
\be
\chi_{G\mathbf{p}} = \Tr\bigl(\Pi U\vpu{\dagger}_{G} D_{\mathbf{p}} U^{\dagger}_G\bigr) = \chi_{\mathbf{p}}
\ee
while if $\Det G = -1$
\begin{align}
\chi_{G\mathbf{p}} & =
\Tr\bigl(\Pi U\vpu{\dagger}_{-GJ} D\vpu{\dagger}_{-J \mathbf{p}}U^{\dagger}_{-GJ}\bigr)\bigr)
\nonumber
\\
&= g_c\Bigl(\Tr\bigl(  g_c(\Pi) U\vpu{\dagger}_{-JG} D_{-\mathbf{p}}U^{\dagger}_{-JG}\bigr)\Bigr)
\nonumber
\\
& = g_c\Bigl(\Tr\bigl(  U^{\dagger}_{-G}\Pi U\vpu{\dagger}_{-G} D_{-\mathbf{p}}\bigr)\Bigr)
\nonumber
\\
&= \chi_{\mathbf{p}}
\end{align}
Suppose, on the other hand, $G\in \ESL(2,\mathbb{Z}_{\bar{d}})$ is such that $\chi_{G\mathbf{p}} = \chi_{\mathbf{p}}$ for all $\mathbf{p}$.  Let $s = \Det G$.  Then it follows from Eq.~(\ref{eq:OpExpn}) that
\begin{align}
\Pi &= \frac{1}{d} \sum_{\mathbf{p}\in \mathbb{Z}_d^2} \chi_{\mathbf{p}} D_{\mathbf{p}}
\nonumber
\\
& =  \frac{d}{{\bar{d}}^2} \sum_{\mathbf{p}\in \mathbb{Z}_{\bar{d}}^2} \chi_{G\mathbf{p}} D_{\mathbf{p}}
\nonumber
\\
& =  \frac{d}{{\bar{d}}^2} \sum_{\mathbf{p}\in \mathbb{Z}_{\bar{d}}^2} \chi_{s\mathbf{p}}U\vpu{\dagger}_{sG^{-1}} D_{\mathbf{p}} U^{\dagger}_{sG^{-1}}
\nonumber
\\
& = U\vpu{\dagger}_{sG^{-1}} \Pi U^{\dagger}_{sG^{-1}}
\end{align}
So  $G \in \bstb{\Pi}$.
\end{proof}

For each $g\in\gZ$ let $F_g$ and, when $d=0$ (mod 3), $\qv_g$ be as in the statement of Theorem~\ref{thm4}.  
Define
\begin{align}
G\vpu{-1}_g
&=
(\Det F_g)F_g^{-1} H\vpu{-1}_g 
\\
\intertext{If $d=0$ (mod $3$) also define $\rv_g \in \mathbb{Z}_3\times \mathbb{Z}_3$ by}
\mathbf{r}_g
& =
-\frac{3k_g}{d} H_g^{-1} \qv_g  \quad \text{ (mod $3$)}
\end{align}
Note that  $G_g\in\GL(2,\mathbb{Z}_{\bar{d}})$ but  is not necessarily an element of $\ESL(2,\mathbb{Z}_{\bar{d}})$.
Let $\nbstb{\Pi}$ be the  normalizer of $\bstb{\Pi}$ considered as a subgroup of $\GL(2,\mathbb{Z}_{\bar{d}})$.
We have
\begin{lemma}
$G_g\in \nbstb{\Pi}$ for all $g\in\gl$.
\end{lemma}
\begin{proof}
We prove the theorem on the assumption that $d=0$ (mod $3$).  The proof when $d\neq 0$ (mod $3$) is essentially the same, though slightly easier as we do not have the factor $D_{\qv_g}$ in the expression for $U_g$.
It is easily seen that $\nbstb{\Pi}$ is the normalizer of $\tstb{\Pi}$ as well as $\bstb{\Pi}$.
So it is enough to prove that $G_g$ is in the normalizer of $\tstb{\Pi}$.
Let $L$ be any element of $\tstb{\Pi}$.  By the argument leading to Eq.~(\ref{eq:thm4C}) there exists $M\in\tstb{\Pi}$ such that
\be
D_{H_gLH_g^{-1}\qv_g -\qv_g}U_{H_gLH_g^{-1}F_{g}M^{-1}F_g^{-1}}\dot{=} I
\ee
If $d$ is odd it follows from Theorem~\ref{eq:kerCliffThm} that
\begin{align}
H\vpu{-1}_g L H_g^{-1} F\vpu{-1}_g M^{-1} F^{-1}_g & = I
\\
\intertext{Hence}
G\vpu{-1}_g L G_g^{-1} = M
\end{align}
Since $L\in\tstb{\Pi}$ is  arbitrary this implies $G_g\in \nbstb{\Pi}$.  If, on the other hand, $d$ is even we have by the same theorem that
\begin{align}
H_gLH_g^{-1}\qv_g -\qv_g & = \bmt \frac{sd}{2} \\ \frac{td}{2}\emt & &\text{(mod $d$)}
\\
H_gLH_g^{-1}F_{g}M^{-1}F_g^{-1} & = \bmt 1+ rd & sd \\ td & 1+rd\emt &&\text{(mod $2d$)}
\end{align}
for some $r,s,t=0$ or $1$.  We also have
\be
H_gLH_g^{-1}\qv_g -\qv_g = \begin{cases} \boldsymbol{0} \text{  (mod $d$)} \qquad & d\neq 0 \text{ (mod 3)}
\\
\boldsymbol{0} \text{  (mod $\frac{d}{3}$)} \qquad & d=0 \text{ (mod 3)}
\end{cases}
\ee
It follows that $s=t=0$ and, consequently,
\be
G\vpu{-1}_g L G^{-1}_g = F^{-1}_g P F\vpu{-1}_g M= P M
\ee
where
\be
P = \bmt 1+r d & 0 \\ 0 & 1+ rd\emt
\ee
and where we used the fact that $P$ commutes with $F_g$.  The fact that $U_P \dot{=} I$ means $PM\in\tstb{\Pi}$.  So $G_g\in \nbstb{\Pi}$ in this case also.
\end{proof}

\begin{theorem} For all $g\in \gZ$ and $\mathbf{p}\in\mathbb{Z}_{\bar{d}}\times \mathbb{Z}_{\bar{d}}$
\be
g(\chi_{\mathbf{p}}) = \begin{cases}
\chi_{G\mathbf{p}} \qquad & d\neq 0 \text{  (mod 3)}
\\
\sigma^{\langle \rv_g ,\mathbf{p}\rangle} \chi_{G \mathbf{p}} \qquad & d= 0 \text{ (mod 3)}
\end{cases}
\ee
where $\sigma = e^{\frac{2\pi i}{3}}$ and $G$ is any element of the coset $G_g \bstb{\Pi}$.
\end{theorem}
\begin{proof}
If $d=0$ (mod $3$)
\begin{align}
g(\chi_{\mathbf{p}}) & = \Tr\left(D_{\qv_g}U\vpu{\dagger}_{F_g} \Pi U^{\dagger}_{F_g} D^{\dagger}_{\qv_g} D_{H_g\mathbf{p}}\right)
\\
& =
\begin{cases}
\omega^{-\langle \qv_g,H_g \mathbf{p}\rangle}\Tr\bigl(  \Pi D_{F_g^{-1}H_g \mathbf{p}}\bigr) 
\qquad & \Det F_g = 1
\\
\omega^{-\langle \qv_g,H_g \mathbf{p}\rangle}\Tr\bigl(  g_c(\Pi) D_{JF_g^{-1}H_g \mathbf{p}}\bigr) \qquad & \Det F_g = -1
\end{cases}
\nonumber
\\
& = 
\begin{cases}
\sigma^{\langle \rv_g,\mathbf{p}\rangle}\Tr\bigl(  \Pi D_{G_g \mathbf{p}}\bigr) 
\qquad & \Det F_g = 1
\\
\sigma^{\langle \rv_g,\mathbf{p}\rangle}g_c\Bigl(\Tr\bigl(  \Pi D_{F_g^{-1}H_g \mathbf{p}}\bigr)\Bigr) \qquad & \Det F_g = -1
\end{cases}
\nonumber
\\
& = 
\begin{cases}
\sigma^{\langle \rv_g,\mathbf{p}\rangle}\Tr\bigl(  \Pi D_{G_g \mathbf{p}}\bigr) 
\qquad & \Det F_g = 1
\\
\sigma^{\langle \rv_g,\mathbf{p}\rangle}\Tr\bigl(  \Pi D_{-F_g^{-1}H_g \mathbf{p}}\bigr)\qquad & \Det F_g = -1
\end{cases}
\nonumber
\\
& = 
\sigma^{\langle \rv_g ,\mathbf{p}\rangle} \chi_{G_g\mathbf{p}}
\label{eq:gOverlapAction1}
\end{align}
If $d\neq 0$ ( mod $3$) we find, by essentially the same argument, that 
\be
g(\chi_{\mathbf{p}}) = \chi_{G_g \mathbf{p}}
\label{eq:gOverlapAction2}
\ee
Finally, for all $H\in \bstb{\Pi}$
\be
\chi_{G_g H \mathbf{p}} = \chi_{G_g H G^{-1}_g G_g\mathbf{p}} = \chi_{G_g \mathbf{p}}
\ee
where we used Lemma~\ref{lem:sbarpiStab} and the fact that $G_g H G^{-1}_g\in \bstb{\Pi}$ (because $G_g \in \nbstb{\Pi}$).  So we can replace $G_g$ with arbitrary $G \in G_g\bstb{\Pi}$ in Eqs.~(\ref{eq:gOverlapAction1}) and~(\ref{eq:gOverlapAction2}).
\end{proof}

\begin{lemma}
\label{lm:homgZ}
The map   $f\colon \gZ \to \nbstb{\Pi}/\bstb{\Pi}$ defined by 
\be
f(g) = G_g \bstb{\Pi}
\ee
is a  homomorphism.
\end{lemma}
\begin{proof} We will prove the result on the assumption that $d=0$ (mod $3$).  The proof for  $d\neq 0$ (mod $3$)  is essentially the same, though somewhat  easier because we do not have the factor $D_{\qv_g}$ in the expression for $U_g$.

It follows from Theorems~\ref{thm1} and~\ref{thm3} that
\begin{align}
D_{\qv_{g_1g_2}}U_{F_{g_1g_2}} &\dot{=} D_{H_{g_1}\qv_{g_2}}U_{H_{g_1}F_{g_2}H^{-1}_{g_1}}D_{\qv_{g_1}} U_{F_{g_1}}U_L
\nonumber
\\
&\dot{=} D_{H_{g_1}\qv_{g_2} + H_{g_1}F_{g_2} H_{g_1}^{-1}\qv_{g_1}}U_{H_{g_1}F_{g_2}H_{g_1}^{-1} F_{g_1} L}
\end{align}
for some $L \in\tstb{\Pi}$.  
If $d$ is odd it follows from Theorem~\ref{eq:kerCliffThm} that
\be
F_{g_1g_2} = H_{g_1}F_{g_2}H_{g_1}^{-1} F_{g_1} L
\ee
If, on the other hand, $d$ is even it follows from the same theorem that
\begin{align}
H_{g_1}\qv_{g_2} + H_{g_1}F_{g_2} H_{g_1}^{-1}\qv_{g_1} -\qv_{g_1g_2} &= \bmt\frac{sd}{2} \\ \frac{td}{2} \emt & &\text{ (mod $d$)}
\\
H_{g_1}F_{g_2}H_{g_1}^{-1} F_{g_1} LF^{-1}_{g_1g_2}& = \bmt 1+rd & sd\\ td & 1+rd\emt & &\text{ (mod $2d$)}
\end{align}
for some $r,s,t=0,1$.
We also have
\be
H_{g_1}\qv_{g_2} + H_{g_1}F_{g_2} H_{g_1}^{-1}\qv_{g_1} -\qv_{g_1g_2} = 
\begin{cases}
\boldsymbol{0} \text{ (mod $d$)} \qquad & d\neq 0 \text{ (mod 3)}
\\
\boldsymbol{0} \text{ (mod $\frac{d}{3}$)} \qquad & d=0 \text{ (mod 3)}
\end{cases}
\ee
It follows  that $s=t=0$ and, consequently,
\begin{align}
F_{g_1g_2} &= H_{g_1}F_{g_2}H_{g_1}^{-1} F_{g_1}L'
\\
\intertext{where}
L' & =L\bmt 1+rd & 0 \\ 0 & 1+rd \emt \in \tstb{\Pi}
\end{align}
So
\be
H\vpu{-1}_{g_1} F\vpu{-1}_{g_2} H_{g_1}^{-1} F\vpu{-1}_{g_1} = F\vpu{-1}_{g_1g_2}M
\ee
for some $M\in\tstb{\Pi}$ in every case, irrespective of whether $d$ is odd or even.  Consequently
\begin{align}
G_{g_1}G_{g_2} &= \Det(F_{g_1}) \Det(F_{g_2}) F^{-1}_{g_1} H\vpu{-1}_{g_1} F^{-1}_{g_2} H\vpu{-1}_{g_2}
\nonumber
\\
&=\Det(F\vpu{-1}_{g_1g_2}) \Det(M)  M^{-1} F^{-1}_{g_1g_2}H\vpu{-1}_{g_1g_2}
\nonumber
\\
& \in  \bstb{\Pi}G_{g_1g_2} 
\end{align}
Taking account of the fact that $G_g \bstb{\Pi} = \bstb{\Pi}G_g$ for all $g$ we conclude
\be
(G_{g_1}\bstb{\Pi})( G_{g_2} \bstb{\Pi}) = (G_{g_1g_2}\bstb{\Pi})
\ee
\end{proof}

It follows from this result that $\gZ/\gZb$ is isomorphic to a subgroup of $\nbstb{\Pi}/\bstb{\Pi}$, where
\be
\gZb = \{ g \in \gZ \colon G_g \in \bstb{\Pi}\}
\label{eq:gZbDef}
\ee 
is the kernel of the homomorphism $f$.  We now examine the structure of $\gZb$.  
\begin{lemma}
Let $\mathcal{O}$ be the orbit of $\Pi$ and let
\be
\mathcal{P} = \{ g \in \gZ \colon g(\tau) = \tau^{-1} \text{ and } g(\Pi') = U\vpu{\dagger}_{-J}\Pi'U^{\dagger}_{-J} \ \forall \, \Pi' \in \mathcal{O} \}
\ee
Then 
\be
\gZb = \glt \cup \mathcal{P}
\ee
\end{lemma}
\begin{remark}
$\mathcal{O}$ is the set of all fiducials on the orbit of $\Pi$, not just the simple ones.  So this result means that $\gZb$ depends only on the orbit, and not on the particular simple fiducial $\Pi$ featuring in the definition, Eq.~(\ref{eq:gZbDef}).

The result also shows that $\glt \subseteq \gZb$.
\end{remark}
\begin{proof}
Let $g\in\gZb$.   Then $G_g \in \bstb{\Pi}$, implying that $\Det G_g=\pm 1$.  In view of the relation
$G_g = \Det (F_g) F^{-1}_g H_g$ this means $k_g = \Det (H_g) = \pm 1$. 

 Suppose $k_g = 1$.  Then $H_g = I$ and  $F_g \in \tstb{\Pi}$, implying 
\begin{align}
g(\tau) & = \tau & g(\Pi) & = \Pi
\end{align}
So $g \in\glt$. 

 Suppose, on the other hand, $k_g = -1$.  Then $H_g = J$ and  $F_g \in -J \tstb{\Pi}$, implying
\begin{align}
g(\tau) & = \tau^{-1}  & g(\Pi) & =U\vpu{\dagger}_{-J} \Pi U^{\dagger}_{-J}
\end{align}
Let $\Pi'$ be an arbitrary element of $\mathcal{O}$. Then $\Pi' = D\vpu{\dagger}_{\mathbf{p}} U\vpu{\dagger}_F \Pi U^{\dagger}_F D^{\dagger}_{\mathbf{p}}$ for some $F$, $\mathbf{p}$.  Consequently
\be
g(\Pi') = D\vpu{\dagger}_{-J \mathbf{p}} U\vpu{\dagger}_{JFJ}U\vpu{\dagger}_{-J} \Pi U^{\dagger}_{-J}U^{\dagger}_{JFJ} D^{\dagger}_{-J\mathbf{p}} = U\vpu{\dagger}_{-J} \Pi' U^{\dagger}_{-J}
\ee
So $g \in \mathcal{P}$.  

It follows that $\gZb \subseteq \glt \cup \mathcal{P}$.  To prove the reverse inclusion let $g \in \glt\cup\mathcal{P}$.  If $g\in\glt$ it can be assumed that $F_g = I$.  Also $H_g=I$.  So $G_g = I $ implying $g \in \gZb$.  If, on the other hand, $g\in \mathcal{P}$ it can be assumed that $F_g = -J$.  Also $H_g = J$.  So we again have $G_g = I $ implying $g\in\gZb$.  Consequently $  \glt \cup \mathcal{P}\subseteq \gZb$.
\end{proof}
This result is easily seen to imply:
\begin{description}
\item[Case 1]
If $\sqrt{d}\in\elt$ and $\mathcal{P}$ is empty then 
\be
\glt= \gZb = \langle e\rangle
\ee
\item[Case 2]
If $\sqrt{d}\in\elt$ and $\mathcal{P}$ is non-empty then
\begin{align}
\glt&=\langle e\rangle & \gZb & = \langle \bar{g}_1\rangle
\end{align}
 where $\bar{g}_1$ is the unique automorphism such that
\begin{align}
\bar{g}_1(\tau) & = \tau^{-1} & \bar{g}_1(\Pi') & = U\vpu{\dagger}_{-J} \Pi' U_{-J}^{\dagger} 
\end{align}
for all $\Pi' \in \mathcal{O}$.
\item[Case 3]
If $\sqrt{d}\notin\elt$ and $\mathcal{P}$ is empty then 
\begin{align}
\glt & = \gZb = \langle \bar{g}_2\rangle
\end{align}
where $\bar{g}_2$ is the unique automorphism such that
\begin{align}
\bar{g}_2(\tau) & = \tau & \bar{g}_2(\Pi') & =  \Pi'   & \bar{g}_2(\sqrt{d}) = -\sqrt{d}
\end{align}
for all $\Pi' \in \mathcal{O}$.
\item[Case 4]
If $\sqrt{d}\notin\elt$ and $\mathcal{P}$ is non-empty then 
\begin{align}
\glt & = \langle \bar{g}_2\rangle & \gZb & = \langle \bar{g}_1,\bar{g}_2\rangle
\end{align}
where $\bar{g}_2$ is as in case $3$ and $\bar{g}_1$ the unique automorphism such that
\begin{align}
\bar{g}_1(\tau) & = \tau^{-1} & \bar{g}_1(\Pi') & =  U\vpu{\dagger}_{-J} \Pi' U^{\dagger}_{-J}  & \bar{g}_1(\sqrt{d}) = \sqrt{d}
\end{align}
for all $\Pi' \in \mathcal{O}$.
\end{description}
It is also easy to see that, when defined, $\bar{g}_1$, $\bar{g}_2$ are order $2$ and $\bar{g}_1\bar{g}_2 = \bar{g}_2\bar{g}_1$.  So
\begin{align}
\glt &\cong \begin{cases} \mathbb{Z}_1 \quad & \text{cases 1,2}  \\ \mathbb{Z}_2 \quad &  \text{cases 3,4} \end{cases}
&
\gZb & \cong \begin{cases}\mathbb{Z}_1 \quad & \text{case 1} \\ \mathbb{Z}_2 \quad & \text{cases 2,3} \\ \mathbb{Z}_2\oplus \mathbb{Z}_2 \quad & \text{case 4} \end{cases}
\end{align}
As we will see in the next section for each of the $27$ orbits with $d>3$ for which an exact fiducial is known  $\mathcal{P}$ is non-empty.  So for these orbits only cases $2$ and $4$ apply.

Turning to the group $\gEZ$, 
 observe that the fact that $\glt \subseteq \gZb$  and the isomorphism 
\be
\gEZ \cong \gZ/\glt
\ee
mean that $f$ induces a homomorphism
\be
f_{\elt} \colon \gEZ \to \nbstb{\Pi}/\bstb{\Pi}
\ee
with the property $f_{\elt}(g|_{\elt}) = f(g)$.  It is straightforward to show that $\ker f_{\elt} = \gEZb$, where
\be
\gEZb = 
\begin{cases}
\langle e \rangle \qquad & \text{$\bar{g}_1$ is not defined}
\\
\langle \bar{g}_1|_{\elt} \rangle \qquad & \text{$\bar{g}_1$ is defined}
\end{cases}
\ee
We have thus proved the following  structure theorem for the groups $\gZ$, $\gEZ$:
\begin{theorem}
\label{thm:Structure}
There exists a subgroup $\hgr \subseteq \nbstb{\Pi}/\bstb{\Pi}$ such that
\begin{enumerate}
\item In  the normal series
\be
\langle e \rangle \vartriangleleft \gZb \vartriangleleft \gZ
\ee
\begin{enumerate}
\item $\gZb$ is  isomorphic to $\mathbb{Z}_1$, $\mathbb{Z}_2$ or $\mathbb{Z}_2 \oplus \mathbb{Z}_2$
\item $\gZ/\gZb \cong \hgr$
\end{enumerate}
\item In  the normal series
\be
\langle e \rangle \vartriangleleft \gEZb \vartriangleleft \gEZ
\ee
\begin{enumerate}
\item $\gEZb$ is isomorphic to $\mathbb{Z}_1$ or $\mathbb{Z}_2$
\item $\gEZ/\gEZb \cong \hgr$
\end{enumerate}
\end{enumerate}
\end{theorem}

Finally, let $\eZb$ be the fixed field of $\gZb$.  Then applying the Galois correspondence to the series
\be
\langle e \rangle \vartriangleleft \glt   \vartriangleleft  \gZb  \vartriangleleft  \gZ \subseteq \gc \subseteq \gl
\label{eq:gpSeries}
\ee
gives us the series
\be
\el  \vartriangleright  \elt \vartriangleright  \eZb \vartriangleright \eZ \supseteq \eC \supseteq   \mathbb{Q}
\label{eq:fldSeries}
\ee

\section{Conjectures}
\label{sec:conjectures}
In this section we propose some conjectures suggested by a study of the known exact fiducials for $d\ge 4$ (we discuss the case $d<4$ in Section~\ref{sec:d2and3}; for the detailed analysis on which on the statements in this section are based see Appendix~\ref{sec:tbles}).  The orbits to which these  fiducials belong comprise (using the classification of Scott and Grassl~\cite{ScottGrassl}) the $5$ doublets $9a,b$, $11a,b$, $13a,b$, $14a,b$, $16a,b$ and the $17$ singlets $4a$, $5a$, $6a$, $7a$, $7b$, $8a$, $8b$, $10a$, $11c$, $12a$, $12b$, $15d$, $19e$, $24c$, $28c$, $35j$, $48j$, where by a doublet we mean a pair of orbits related by a Galois automorphism (so $\eZ \neq \eC$), and by a singlet we mean an orbit closed under the action of the Galois group (so $\eZ = \eC$).  That is $27$ orbits in total.  Exact fiducials for orbits $16a,b$ and $28c$ can be found in refs.~\cite{Appleby2EtAl,Appleby3EtAl}, exact fiducials for all the others can  be found in  ref.~\cite{ScottGrassl}.   In many cases these fiducials are not simple.  The exact simple fiducials on which the calculations in this paper are based are available online at ref.~\cite{fiducials}.

In  Section~\ref{sec:StructureThm} we showed that the 3 leftmost inclusions  in Eqs.~(\ref{eq:gpSeries}) and~(\ref{eq:fldSeries}) are necessarily normal.  
We also proved a structure theorem for the groups  $\gEZ$ and $\gZ$.  The known exact fiducials for $d> 3$ suggest that it may be possible to prove some much stronger statements.   
Specifically, we find that, for all $27$ known cases with $d> 3$,
\begin{enumerate}
\item \label{it:norm} Every inclusion in Eqs.~(\ref{eq:gpSeries}) and~(\ref{eq:fldSeries}) is  normal:
\begin{align}
\langle e \rangle \vartriangleleft \glt \vartriangleleft \gZb&\vartriangleleft  \gZ \vartriangleleft \gc \vartriangleleft \gl
\\
\el \vartriangleright \elt \vartriangleright \eZb &\vartriangleright \eZ \vartriangleright \eC \vartriangleright \mathbb{Q}
\\
\intertext{and, consequently,}
\langle e \rangle \vartriangleleft \gEZb&\vartriangleleft  \gEZ \vartriangleleft \gEc \vartriangleleft \gEl
\end{align}
\item \label{it:gcAbelian}$\gEc$, $\gc$  are Abelian. 
\item \label{it:QaProperty}  For the extension $\eC \vartriangleright \mathbb{Q}$
\begin{align}
\eC &= \mathbb{Q}(a)
\\
\gEl/\gEc \cong \gl / \gc &\cong \mathbb{Z}_2
\end{align}
where $a$ is the square-root of the square-free part of $(d-3)(d+1)$:  \emph{i.e.}
the quantity
\be
a = \frac{\sqrt{(d-3)(d+1)}}{n}
\ee
with $n$  the largest positive integer such that $n^2$ is a factor of $(d-3)(d+1)$.
\item For the extension $\eZ \vartriangleright \eC$
\begin{align}
\eZ & =
 \begin{cases} 
 \eC \qquad &\text{orbit is a singlet}
\\
\eC(\sqrt{p}) \qquad & \text{orbit is one of a doublet}
\end{cases}
\\
\gEc/\gEZ \cong \gc/\gZ &\cong
 \begin{cases} 
 \mathbb{Z}_1 \qquad &\text{orbit is a singlet}
\\
\mathbb{Z}_2  \qquad & \text{orbit is one of a doublet}
\end{cases}
\end{align}
where $p$ is a prime divisor of $a^2$ (\emph{i.e.} a prime divisor of $(d-3)(d+1)$ having odd multiplicity).
\item \label{it:gZgZbQuotient} For the extension $\eZb \vartriangleright \eZ$
\be
\gEZ/ \gEZb \cong \gZ/\gZb\cong \cbstb{\Pi}/\bstb{\Pi}
\ee
where $\Pi$ is any simple fiducial on the orbit and $\cbstb{\Pi}$ is the centralizer of $\bstb{\Pi}$ considered as a subgroup of $\GL(2,\mathbb{Z}_{\bar{d}})$.  This is a much stronger statement than Theorem~\ref{thm:Structure}, which only says that $\gEZ/\gEZb$, $\gZ/\gZb$ are isomorphic to a subgroup of $\nbstb{\Pi}/\bstb{\Pi}$, without saying which particular subgroup.
\item For the extension $\elt \vartriangleright \eZb$
\begin{align}
\elt & = \eZb(i) \text{ or } \eZb(i \sqrt{d})
\\
\gEZb \cong \gZb/\glt &\cong \mathbb{Z}_2
\end{align}
\item For the extension $\el \vartriangleright \elt$ 
\begin{align}
\el & =
\begin{cases}
 \elt \qquad & \sqrt{d} \in \elt
 \\
 \elt(\sqrt{d}) \qquad & \sqrt{d} \notin \elt
 \end{cases}
\\
\glt &\cong
\begin{cases}
 \mathbb{Z}_1\qquad & \sqrt{d} \in \elt
 \\
 \mathbb{Z}_2 \qquad & \sqrt{d} \notin \elt
 \end{cases}
\end{align}
\item \label{it:glgZnormSeries} The series
\be
\langle e \rangle \vartriangleleft \gZ \vartriangleleft \gl
\ee
is normal, with $\gZ$, $\gl/\gZ$ both Abelian and
\begin{align}
\gZ & \cong
\begin{cases}
\mathbb{Z}_2 \oplus \cbstb{\Pi}/\bstb{\Pi} \qquad & \sqrt{d} \in \elt
\\
\mathbb{Z}_2 \oplus \mathbb{Z}_2 \oplus \cbstb{\Pi}/\bstb{\Pi} \qquad & \sqrt{d} \notin \elt
\end{cases}
\\
\gl/\gZ & \cong 
\begin{cases}
\mathbb{Z}_2 \qquad & \text{orbit is a singlet}
\\
\mathbb{Z}_2 \oplus \mathbb{Z}_2 \qquad & \text{orbit is one of a doublet}
\end{cases}
\end{align}
\item \label{it:gElgEZnormSeries} The series
\be
\langle e \rangle \vartriangleleft \gEZ \vartriangleleft \gEl
\ee
is normal, with $\gEZ$, $\gEl/\gEZ$ both Abelian and
\begin{align}
\gEZ & \cong
\mathbb{Z}_2 \oplus \cbstb{\Pi}/\bstb{\Pi} 
\\
\gEl/\gEZ & \cong 
\begin{cases}
\mathbb{Z}_2 \qquad & \text{orbit is a singlet}
\\
\mathbb{Z}_2 \oplus \mathbb{Z}_2 \qquad & \text{orbit is one of a doublet}
\end{cases}
\end{align}
\end{enumerate}
Note that items~\ref{it:gcAbelian} and \ref{it:QaProperty} in this list together imply the statement made in the Introduction, that $\el$ is an Abelian extension of the real quadratic field $\mathbb{Q}(a)$.  It is to be observed that the extensions $\eC \vartriangleright \mathbb{Q}$, $\eZ \vartriangleright \eC$, $\elt \vartriangleright \eZb$, $\el \vartriangleright \elt$ are all degree $2$ at most, with Galois groups $\cong \mathbb{Z}_1$ or $\mathbb{Z}_2$.  By contrast the extension $\eZb \vartriangleright \eZ$ is degree $\gg 2$ and its Galois group $\cong \cbstb{\Pi}/\bstb{\Pi}$ is correspondingly richer.  So it could be said, loosely speaking, that the most of the ``meat'' of the problem resides in the extension  $\eZb \vartriangleright \eZ$.

It is tempting to conjecture that the items on the above list hold generally, for every orbit with $d>3$.  And, indeed, for type $z$ orbits we shall so conjecture.  However, to arrive at a plausible conjecture for type $a$ orbits  items~\ref{it:gZgZbQuotient},   \ref{it:glgZnormSeries} and~\ref{it:gElgEZnormSeries}  need to be modified.  

For the  statements in the list all to be true it is necessary that $\bstb{\Pi}$ should be Abelian  (since otherwise  $\cbstb{\Pi}/\bstb{\Pi}$ would not be defined).   Assuming $\bstb{\Pi}$  is Abelian it is also necessary that $\cbstb{\Pi}/\bstb{\Pi}$ should be Abelian (since otherwise $\gc$ would be non-Abelian).   The first requirement does not present us with any obvious difficulties.  Although $\bstb{\Pi}$ is non-Abelian for the fiducials in dimensions $2$ and $3$, it is Abelian (in fact cyclic) for every orbit on which a fiducial (exact or numerical) has been calculated in dimensions $d>3$.  The conjecture that it  is always Abelian for $d>3$ is therefore plausible.  However, there is a problem with the second requirement, that $\cbstb{\Pi}/\bstb{\Pi}$ is Abelian.

The following lemma shows that if the orbit is type $z$ there is no difficulty (provided that $\bstb{\Pi}$ is Abelian).
\begin{lemma}
Let $\Pi$ be a simple fiducial such that $\bstb{\Pi}$ is Abelian and contains a matrix $F$ conjugate to $F_z$.  Then $\cbstb{\Pi}$ consists of all matrices $\in \GL(2,\mathbb{Z}_{\bar{d}})$ which can be written in the form
\be
n I + m F
\label{eq:lm10Res}
\ee
for some $n$, $m\in \mathbb{Z}_{\bar{d}}$.  In particular $\cbstb{\Pi}$, and consequently $\cbstb{\Pi}/\bstb{\Pi}$, are Abelian.
\end{lemma}
\begin{remark}
This result means that we can write down the group $\cbstb{\Pi}$  without knowing which, if any,  matrices apart from $F$ are in $\bstb{\Pi}$.
\end{remark}
\begin{proof}
Let $C_{F_z}$ be the centralizer of $F_z$. The  matrix
\be
G = \bmt x & y \\ z & w \emt 
\ee
commutes with $F_z$ if and only if
\be
F_z G F^{-1}_z = G
\ee
which one easily sees to be equivalent to the condition $z = -y$ and $w=x+y$.  It follows that $C_{F_z}$ consists of all matrices $\in \GL(2,\mathbb{Z}_{\bar{d}})$ of the form
\be
n I + m F_z
\ee
for some $n,m\in\mathbb{Z}_{\bar{d}}$.
By assumption
\be
F = H F_z H^{-1}
\ee
for some $H\in\ESL(2,\mathbb{Z}_{\bar{d}})$.  So $C_F$, the centralizer of $F$, consists of all matrices of the form 
\be
n I + m F
\ee
for some $n,m\in\mathbb{Z}_{\bar{d}}$.  The fact that $C_F$ is Abelian means that $C_F = \cbstb{\Pi}$.
\end{proof}
\noindent For type $z$ orbits we shall accordingly conjecture that the items on the above list hold as stated for every dimension $>3$.

The difficulty comes when we consider type $a$ orbits.  There are two such orbits on which exact fiducials have been calculated (namely $12b$ and $48g$), and for both of those $\cbstb{\Pi}/\bstb{\Pi}$ is Abelian.  However, if one looks at other type $a$ orbits, where only a numerical fiducial is known, one finds that there are cases where $\cbstb{\Pi}/\bstb{\Pi}$ is non-Abelian. For instance, in the case of the Scott-Grassl numerical fiducial on orbit $21e$, the  matrices
\begin{align}
G_1 & =\bmt 0 & 1 \\ 2 & 6 \emt & G_2 & = \bmt 0 & 1 \\ 2 & 13 \emt
\end{align}
both $\in \cbstb{\Pi}$, but $G_1\bstb{\Pi}$ and $G_2\bstb{\Pi}$ do not commute.

The following lemma will help us to understand why $\cbstb{\Pi}/\bstb{\Pi}$ is Abelian for some type $a$ orbits, but not for others. It will also help us to formulate a modified conjecture for type $a$ orbits.
\begin{lemma} 
Let  $d=9k+3$ for some positive integer $k$ and let $F\in \SL(2,\mathbb{Z}_{\bar{d}})$ be conjugate to $F_a$.  Let $C_F$ be the centralizer of  $\langle F \rangle$ considered as a subgroup of $\GL(2,\mathbb{Z}_{\bar{d}})$.   Let $G$ be any solution to the equation $3G = F-I$.   Then $C_F$ consists of all matrices  $\in \GL(2,\mathbb{Z}_{\bar{d}})$ which can be written in the form
\be
nI + m G + \frac{\bar{d}}{3}H
\label{eq:centTypeaFmla}
\ee 
for some $n,m\in \mathbb{Z}_{\bar{d}}$ and  arbitrary matrix $H$.  
\end{lemma}
\begin{proof}
We will prove the result for the case $F=F_a$.  The generalization to arbitrary $F$ conjugate to $F_a$ is immediate.

Let 
\be
L = \bmt x & y \\ z & w \emt \in \GL(2,\mathbb{Z}_{\bar{d}})
\ee
Suppose $L$ commutes with $F_a$.  Then $L' = L - x I$ commutes with $F_a- I$.  So 
\begin{align}
\bmt 0 & y \\ z & w-x\emt \bmt 0 & 9k+6 \\ 12k + 3 & 9k \emt &= \bmt 0 & 9k+6 \\ 12k + 3 & 9k \emt \bmt 0 & y \\ z & w-x\emt  & &\text{(mod $\bar{d}$)}
\end{align}
or, equivalently,
\begin{align}
(4k+1)y & = (3k+2)z  &  &\text{(mod $\bar{d}/3$)}
\\
3k y & = (3k+2)(w-x) & & \text{(mod $\bar{d}/3$)}
\\
(4k+1)(w-x) & = 3k z  & &  \text{(mod $\bar{d}/3$)}
\end{align}
Using the fact that $-3(4k+1) = 1$ (mod $\bar{d}$)  we deduce that
\begin{align}
L & = x I + z \bmt 0 & 3k-2 \\ 1 & - (3k-2)\emt && \text{(mod $\bar{d}/3$)}
\\
& = x I - 3 z G && \text{(mod $\bar{d}/3$)}
\end{align}
where 
\be
G = \bmt 0 & 3k+2 \\ 4k+1 & 3k \emt
\label{eq:Gmt}
\ee
satisfies $3G = F_a-I$.  

Suppose, on the other hand, 
\be
L = n I + m G + \frac{\bar{d}}{3} H
\ee
where  $n,m\in\mathbb{Z}_{\bar{d}}$,  $G$ is a solution to $3G = F_a-I$ and $H$ is arbitrary.  We have
\begin{align}
(F_a-I)H - H (F_a-I) &= 3 (GH-HG)
\\
\intertext{and}
(F_a-I)G - G(F_a-I) & = 3G^2 - 3G^2 = 0
\end{align}
So $F_a-I$, and consequently $F_a$ commutes with $L$.
\end{proof}
The presence of the arbitrary matrix $H$ on the right hand side of Eq.~(\ref{eq:centTypeaFmla}) means that $C_F$ is non-Abelian.  So for orbits for which $\bstb{\Pi} = \langle F\rangle$ for some $F$ conjugate to $F_a$ the centralizer $\cbstb{\Pi}$ is non-Abelian.  Of the type $a$ orbits for which a fiducial is known (exact or numerical) this is true of $21e$, $30d$, $39g,h$,  $48e$.  But for orbits $12b$ and $48g$ the group $\langle F\rangle$ is a proper subgroup of $\bstb{\Pi}$.  As a result $\cbstb{\Pi}$ is a proper subgroup of $C_F$, which is how it comes about that $\cbstb{\Pi}$ is Abelian.

We actually find that for orbits $12b$, $48g$ the group $\cbstb{\Pi}$ consists of all matrices $\in\GL(2,\mathbb{Z}_{\bar{d}})$ which can be written in the form
\be
nI + mG
\ee
for some $n,m\in\mathbb{Z}_{\bar{d}}$, and fixed matrix $G$ satisfying $3G = F-I$, where $F\in\bstb{\Pi}$ is conjugate to $F_a$.  
Specifically:
\begin{align}
G &= 
\begin{cases}
\bmt 0 & 5 \\ 5 & 3\emt \qquad & \text{for the Scott-Grassl numerical fiducial on $12b$}
\\
\bmt 0 & 17 \\ 53 & 79 \emt \qquad & \text{for the Scott-Grassl numerical fiducial on $48g$}
\end{cases}
\end{align}

This suggests a  modification of conjectures~\ref{it:gZgZbQuotient}, \ref{it:glgZnormSeries}  which may perhaps hold for all type $a$ orbits.  If $\Pi$ is a simple fiducial on a type $a$ orbit with canonical trace $-1$ symplectic $F$ define, for  $G$ satisfying $3G = F-I$,
\be
\cbstb{\Pi}^G = \{ n I + m G \colon n,m\in\mathbb{Z}_{\bar{d}} \text{ and } \gcd(\Det(nI+mG),\bar{d}) =1 \}
\ee
For type $a$ orbits we then replace $\cbstb{\Pi}/\bstb{\Pi}$ in  conjectures~\ref{it:gZgZbQuotient}, \ref{it:glgZnormSeries} and~\ref{it:gElgEZnormSeries} with $\cbstb{\Pi}^G/\bstb{\Pi}^{\vphantom{G}}$. Note that there will, in general, be more than one candidate for the matrix $G$ (although the requirements that  $\cbstb{\Pi}^G$ is a group and that $\bstb{\Pi}^{\vphantom{G}}\subseteq \cbstb{\Pi}^G$ impose a constraint).

\section{Further remarks concerning the generators of $\el$ for $d>3$}
\label{sec:fldGen}
In Section~\ref{sec:conjectures} we presented a number of conjectures regarding the fields  and   groups for $d>3$.  The purpose of this section is to present two additional observations which, though perhaps less interesting, still seem worthy of note.

 In the first place it is noticeable that, for each of the known exact fiducials with $d >3$, the field $\el$  contains the square roots of all the prime divisors of $a^2$ and $d$.  However, with the exception of orbit $8a$ (where it contains $\sqrt{3}$, the generator denoted $b_2$ in the tables in Appendix~\ref{sec:tbles}), it does not contain the square roots of the prime divisors of $(d-3)(d+1)$ having even multiplicity. It is also noticeable that $\el$ never contains $\sqrt{p}$ for $p$  a prime which is not a divisor of $d(d-3)(d+1)$.

Our second remark concerns the fields for dimensions divisible by $3$.  With the exception of orbits $12b$ and $48g$ we find that $\el$ always has a generator of the form
\be
2\re\left((n+i m\sqrt{l})^{\frac{1}{3}}\right)
\ee
with minimal polynomial 
\be
x^3 - 3(n^2+m^2l)^{\frac{1}{3}}x - 2n
=
\begin{cases}
x^3 - d x - 2n \qquad & \text{$d$ even}
\\
x^3 - 2 d x - 2n \qquad & \text{$d$ odd}
\end{cases}
\ee
where $n$, $m$, $l$ are positive integers, $l$ is a divisor of $da^2$, and
\be
n^2+m^2l = 
\begin{cases}
\frac{d^3}{27} \qquad & \text{$d$ even}
\\
\frac{8d^3}{27} \qquad & \text{$d$ odd}
\end{cases}
\ee
 Specifically

\vspace{0.1 in}
\begin{center}

\begin{tabular}{|c|c|c|c|c|c|c|}
\hline
orbit & $n$ & $m$ & $l$ 
\\
\hline
$6a$ & 1 &1 & 7 
\\
\hline
$9a,b$ & 6 & 6 & 5
\\
\hline
$12a$ & 5 & 1 & 39
\\
\hline 
$15d$ & 10 & 30 & 1
\\
\hline
$24c$ & 8 & 8 & 7
\\
\hline
\end{tabular}

\end{center}

\vspace{0.1 in}

\noindent The fact that orbits $12b$ and $48g$ are exceptions to this rule could be related to the fact that these are type $a$ orbits.  It could also be related to the fact that for these orbits $\bstb{\Pi}$ is strictly larger than the cyclic subgroup generated by the matrix conjugate to $F_a$ (the results described in Section~\ref{sec:conjectures}  suggest that there is an inverse relation between the order of the group $\bstb{\Pi}$ and the degree of the field). 

\section{$g$-Unitaries}
\label{sec:Gunit}
By a definition a fiducial projector $\Pi$ is a joint eigenprojector of the group $\stb{\Pi}$.  In this section we show $\Pi$ is a joint eigenprojector of a larger group of $g$-unitaries.

The concept of a $g$-unitary is a generalization of the concept of an anti-unitary.  Given a Hilbert space $\mathcal{H}$  an anti-linear operator  is a map $\Gamma \colon \mathcal{H} \to \mathcal{H}$ such that 
\begin{align}
\Gamma (\alpha_1 \psi_1 + \alpha_2 \psi_2) &= \alpha_1^{*} \Gamma \psi_1 + \alpha_2^{*} \Gamma \psi_2
\\
\intertext{for all $\psi_1,\psi_2\in \mathcal{H}$ and $\alpha_1,\alpha_2 \in \mathbb{C}$ (switching temporarily from Dirac notation to standard mathematical notation).  Its adjoint  is the unique anti-linear operator $\Gamma^{\dagger}$ such that }
\langle \Gamma^{\dagger}\phi ,\psi\rangle & = \langle \phi, \Gamma \psi\rangle^{*}
\\
\intertext{for all $\phi, \psi \in\mathcal{H}$.  We can write}
\Gamma &= \tilde{\Gamma} U_J
\end{align}
where $\tilde{\Gamma}$ is a linear operator.  We then have
\be
\Gamma^{\dagger} = U^{\dagger}_J \tilde{\Gamma}^{\dagger} =U\vpu{\dagger}_J \tilde{\Gamma}^{\dagger} =  \tilde{\Gamma}^{\mathrm{T}} U\vpu{\dagger}_J
\ee
$\Gamma$ is said to be an anti-unitary if $\Gamma^{\dagger} \Gamma= I$.  $\Gamma$ is an anti-unitary  if and only if $\tilde{\Gamma}$ is a unitary.

We now define the concept of a $g$-unitary in a  way which  parallels these definitions. Let $\mathbb{F}$ be a normal extension of the rationals with Galois group $\graw^{\mathbb{F}}$.  The fact that the extension is normal means that  complex conjugation  $g_c$ is in the  group.  Let  $\graw^{\mathbb{F}}_c$ be the centralizer of $g_c$.  Let $\mathcal{H}$ be a vector space over $\mathbb{F}$.  Given $g\in\graw^{\mathbb{F}}_c$ a map $\Gamma \colon \mathcal{H}\to \mathcal{H}$ is $g$-linear if
\begin{align}
\Gamma (\alpha_1 \psi_1 + \alpha_2 \psi_2) &=g( \alpha_1)\Gamma \psi_1 + g(\alpha_2)\Gamma \psi_2
\\
\intertext{for all $\psi_1,\psi_2\in \mathcal{H}$ and $\alpha_1,\alpha_2 \in \mathbb{F}$ (so a $g_c$-linear operator is anti-linear, and an $e$-linear operator is linear).  Its adjoint  is the unique $g^{-1}$-linear operator $\Gamma^{\dagger}$ such that }
\langle \Gamma^{\dagger}\phi ,\psi\rangle & =g^{-1}\bigl( \langle \phi, \Gamma \psi\rangle\bigr)
\\
\intertext{Let $W_g$ be the $g$-linear operator whose action in the standard basis in}
\bigl(g(\psi)\bigr)_r & = g(\psi_r)
\\
\intertext{(so $W_{g_c} = U_J$).  Then}
\Gamma &= \tilde{\Gamma}W_g
\end{align}
where $\tilde{\Gamma}$ is a linear operator.  We have
\be
\Gamma^{\dagger} = W^{\dagger}_g \tilde{\Gamma}^{\dagger} = W\vpu{\dagger}_{g^{-1}} \tilde{\Gamma}^{\dagger}  = g^{-1}(\tilde{\Gamma}^{\dagger})W\vpu{\dagger}_{g^{-1}}
\ee
We say that $\Gamma$ is a $g$-unitary if $\Gamma^{\dagger}\Gamma = I$.   It is easily seen that $\Gamma$ is a $g$-unitary if and only if $\tilde{\Gamma}$ is a unitary.

Now specialize to the case of interest to us, $\mathbb{F} = \el$ and $\graw^{\mathbb{F}} = \gl$. We know from Section~\ref{sec:GenRes} that for each fiducial projector $\Pi$ and each $g\in \gZ$ there is a unitary $U_g$ such that
\be
g(\Pi) = U\vpu{\dagger}_g \Pi U^{\dagger}_g
\ee
So $\Pi$ is an eigenprojector of the $g$-unitary $V_g = U^{\dagger}_g W\vpu{\dagger}_g$:
\be
V\vpu{\dagger}_g \Pi V^{\dagger}_g = \Pi
\ee
 
 We conclude with some remarks concerning the action of the $V_g$ on the vectors in the subspace  onto which $\Pi$ projects, as this involves some subtleties which do not arise in the case of ordinary unitaries.
 
Let $|\Psi\rangle$ be an unnormalized vector corresponding to $\Pi$ (reverting to Dirac notation).  So
 \be
 \Pi = \frac{1}{\langle \Psi | \Psi\rangle} |\Psi\rangle \langle \Psi|
 \ee
 It is always possible to choose $|\Psi\rangle$ so that its standard basis components are in $\el$.  For instance one can define
 \be
 |\Psi\rangle = \sum_{r} \langle r |\Pi | s \rangle |r\rangle
 \ee
 for some fixed $s$ such that $|\Psi\rangle \neq 0$.  Assume that this has been done.  Then it is easily seen that
 \be
 V_g |\Psi\rangle = \lambda |\Psi\rangle
 \ee
 for some $\lambda$.  We have
 \be
 |\lambda|^2 = \frac{g^{-1}(\langle \Psi | \Psi\rangle)}{\langle \Psi | \Psi \rangle}
 \ee
   If $V_g$ was a unitary or anti-unitary we would have $|\lambda|^2= 1$.  But for an arbitrary $g$-unitary this will only be true if $\langle \Psi | \Psi \rangle$ is in the fixed field of $g$, which will typically not be the case.
    Of course we are free to  consider the normalized vector 
 \be
 |\psi\rangle = \frac{1}{\sqrt{\langle \Psi | \Psi \rangle}} |\Psi\rangle
 \ee
 For this one will have
 \be
 V_g |\psi\rangle = e^{i\theta} |\psi\rangle
 \ee
 for some phase $e^{i\theta}$.  However it will typically happen that the norm of $|\Psi\rangle$ is not in the field $\el$.  One is free to extend the field, so that it does contain $\sqrt{\langle \Psi | \Psi \rangle}$.  But the price one may have to pay is that the enlarged field is no longer an Abelian extension of a real quadratic field.
\section{WH SICs for  $d =2$ and $3$}
\label{sec:d2and3}
The WH SICs in dimensions $2$ and $3$ are exceptional in a number of ways.
In the first place for  $d=3$ there are infinitely  many WH SICs~\cite{Appleby1,RenesEtAl}.  The $\EC(d)$ orbits they generate are labelled by a real parameter $t \in [0,\frac{\pi}{6}]$.  An unnormalized fiducial vector on the orbit corresponding to parameter value $t$ is~\cite{Appleby1}
\be
\bmt 0 \\ 1 \\ e^{2i t}\emt
\ee
So the field $\elt$ will typically be transcendental, and even when it is algebraic it will typically not be solvable.  In this section we will therefore confine ourselves  to the two special cases $t=0$ and $\frac{\pi}{6}$  which are of particular interest~\cite{Hughston,Bengtsson1,Tabia2,LinDep}, and for which the Galois group is very  simple.  In the classification scheme of Scott and Grassl~\cite{ScottGrassl} they are orbits $3b$ (corresponding to $t=\frac{\pi}{6}$) and $3c$ (corresponding to $t=0$).   For these orbits the symmetry group is of higher order ($|\tstb{\Pi}| = 12, 48$ respectively, as opposed to $6$ for the generic orbit).  Also the SICs on these  orbits are related to the Hesse configuration known to projective geometers since the $19^{\mathrm{th}}$ century~\cite{Hughston,Bengtsson1,LinDep}.  

In the second place $(d-3)(d+1)$ is not a positive integer when $d=2$ or $3$, so we do not expect the number $a$ to have the same significance which it seems to have when $d>3$.  In fact, it turns out that $\gl$ is Abelian and $\eC = \mathbb{Q}$ for  orbits $2a$,  $3b$, $3c$.

In the third place the group $\bstb{\Pi}$ is non-Abelian for every orbit in dimensions $2$ and $3$, so  $\cbstb{\Pi}/\bstb{\Pi}$ is not defined.  In fact, it turns out that $\gZb/\gZ\cong \nbstb{\Pi}/\bstb{\Pi}$ for the three orbits $2a$, $3b$, $3c$.

The fields are as follows:

\vspace{0.1 in}

\begin{center}

\begin{tabular}{|c|c|c|c|c|c|}
\hline
orbit \myStrut{4} & $\eC$ & $\eZ$ & $\eZb$ & $\elt$ & $\el$
\\
\hline
$2a$ \myStrut{4} & $\mathbb{Q}$ & $\mathbb{Q}$ & $\eZ(\sqrt{3})$ & $\eZb(i)$ & $\elt(\sqrt{2})$
\\
\hline
$3b$, $3c$ \myStrut{4} & $\mathbb{Q}$ & $\mathbb{Q}$ & $\mathbb{Q}$ & $\eZb(i\sqrt{3})$ & $\elt(\sqrt{3})$
\\
\hline
\end{tabular}
\end{center}

\vspace{0.1 in}

For orbit $2a$   we have $\gl = \gc =\gZ =  \langle g_1,\bar{g}_1,\bar{g}_2\rangle$, $\gZb = \langle \bar{g}_1,\bar{g}_2\rangle$, $\glt = \langle \bar{g}_2\rangle$ where 

\vspace{0.1 in}

\begin{center}

\begin{tabular}{|c|c|c|c|}
\cline{2-4}
\multicolumn{1}{c|}{}  & \myStrut{4}$\sqrt{3}$ & $i$ & $\sqrt{2}$
\\
\hline
$g_1$ \myStrut{4} & $-\sqrt{3}$ & $i$ & $\sqrt{2}$
\\
\hline
$\bar{g}_1$ \myStrut{4} & $\sqrt{3}$ & $-i$ & $\sqrt{2}$
\\
\hline
$\bar{g}_2$ \myStrut{4} & $\sqrt{3}$ & $i$ & $-\sqrt{2}$
\\
\hline
\end{tabular}

\end{center}
and
\be
\gZb/\gZ \cong \nbstb{\Pi}/\bstb{\Pi} \cong \mathbb{Z}_2
\ee

For orbits $3b,c$ we have $\gl = \gc = \gZ = \gZb = \langle \bar{g}_1, \bar{g}_2\rangle$, $\glt = \langle \bar{g}_2\rangle$ where

\vspace{0.1 in}

\begin{center}

\begin{tabular}{|c|c|c|}
\cline{2-3}
\multicolumn{1}{c|}{} &\myStrut{4}  $i\sqrt{3}$ & $\sqrt{3}$
\\
\hline
$\bar{g}_1$ \myStrut{4} & $-i\sqrt{3}$ & $\sqrt{3}$
\\
\hline
$\bar{g}_2$ \myStrut{4} & $i\sqrt{3}$ & $-\sqrt{3}$
\\
\hline
\end{tabular}

\end{center}
and
\be
\gZb/\gZ \cong \nbstb{\Pi}/\bstb{\Pi} \cong \mathbb{Z}_1
\ee

\section{Conclusion}
Our  reason for undertaking this research was that we hoped that a study of the Galois group might suggest a solution to the SIC-existence problem.  That hope was not fulfilled.  Nevertheless we continue to feel that the striking simplicity of the Galois group, and the many structural features which our study has revealed, provide some clues which, combined with other insights, may eventually lead to a proof.

}

\section*{Acknowledgements}
We thank  Markus Grassl, David Gross and Jon Yard for many valuable discussions.
DMA was supported in part by the U.~S. Office of Naval Research (Grant
No.\ N00014-09-1-0247) and by the John Templeton Foundation.  DMA is also grateful to the Stellenbosch Institute for Advanced Study for their hospitality while carrying out some of the research for this paper.  Research at Perimeter Institute is supported
by the Government of Canada through Industry Canada and by the
Province of Ontario through the Ministry of Research \& Innovation.

\appendix

{\allowdisplaybreaks

\captionsetup[table]{font=scriptsize}

\section{Detailed description of the fields and groups for $d\ge 4$}
\label{sec:tbles}
We present our results in a series of tables and boxes.  We begin with an explanation of how to read this information.  

Many of the exact fiducials in refs.~\cite{ScottGrassl,Appleby2EtAl,Appleby3EtAl}  are not simple. By contrast the Scott-Grassl~\cite{ScottGrassl}  numerical fiducials are all simple. They are, besides,  eigenvectors of $U_{F_z}$ or $U_{F_a}$  (with the exception of orbits $8a,b$ where the numerical fiducials in ref.~\cite{ScottGrassl} are not in fact eigenvectors of $U_{F_z}$).    We have therefore worked with exact versions of the Scott-Grassl numerical fiducials (with the exception of orbits $8a,b$ where we have chosen fiducials which are eigenvectors of $U_{F_z}$).  Let $\Pi_0$ be the exact fiducial given in ref.~\cite{ScottGrassl,Appleby2EtAl,Appleby3EtAl} and let $\Pi$ be the exact fiducial to which the results in this appendix refer.  Then $\Pi$ is calculated using
\be
\Pi = D\vpu{\dagger}_{\mathbf{p}} U\vpu{\dagger}_F \Pi_0 U^{\dagger}_F D^{\dagger}_{\mathbf{p}}
\ee
where $\mathbf{p}$ and $F$ are tabulated in Table~\ref{tble:EtoNTransform}.  Exact fiducials corresponding to the projectors $\Pi$ are available online~\cite{fiducials}.

We choose the field generators to be
\begin{enumerate}
\item The  number $a$, defined to be  the square root  of the square-free part  of $(d-3)(d+1)$.
\item The  number $t$, defined to be either $\sin \frac{\pi}{d}$ or $\cos\frac{\pi}{d}$.
\item A set of numbers $r_1, \dots, r_{j} $ defined to be the prime factors of $d$ and $a^2$.
\item A set of numbers $b_1, \dots, b_{k} $ constructed recursively from $\mathbb{Q}(a,t,r_1,\dots, r_{k})$  by taking sums, products, square roots and  cube roots.
\item The  number $i=\sqrt{-1}$.
\end{enumerate}
With the exception of $i$  their values, and minimal polynomials when these are not obvious, are tabulated in Tables~\ref{tble:artGenerators}--\ref{tble:minPolyGt2}  for each dimension.  The various subfields are defined in terms of the field generators in Table~\ref{tble:fields}.  Expressions for $\tau$ in terms of the generators are given in Table~\ref{tble:tauVals}.  

The remaining information is tabulated orbit by orbit, in a series of boxes.  For each orbit box 1 specifies  the Galois group generators  in terms of the auxiliary quantities defined in box 2.   In terms of these generators $\gl$ and its subgroups are

\vspace{0.1 in}
 
\SMALL
\begin{center}
\begin{tabular}{|c|c|c|c|c|c|}
\hline
 \myStrut{4} & $\gl$ & $\gc$ & $\gZ$ & $\gZb$ & $\glt$
\\
\hline
 \myStrut{4} y, s& $\langle  g_a,g_1,\dots,g_n,\bar{g}_1\rangle$  & $\langle  g_1,\dots,g_n,\bar{g}_1\rangle$  & $\langle g_1,\dots,g_n, \bar{g}_1\rangle $ & $\langle \bar{g}_1\rangle$ & $\langle e\rangle$
\\
\hline
 \myStrut{4} y, d &$\langle g_a, g_s, g_1,\dots,g_n,\bar{g}_1\rangle$  & $\langle g_s, g_1,\dots,g_n,\bar{g}_1\rangle$  & $\langle g_1,\dots,g_n, \bar{g}_1\rangle $ & $\langle \bar{g}_1\rangle$ & $\langle e\rangle$
\\
\hline
 \myStrut{4} n, s & $\langle  g_a , g_1,\dots,g_n,\bar{g}_1,\bar{g}_2\rangle$& $\langle  g_1,\dots,g_n,\bar{g}_1,\bar{g}_2\rangle$  & $\langle g_1,\dots,g_n, \bar{g}_1,\bar{g}_2\rangle $ & $\langle \bar{g}_1,\bar{g}_2\rangle$ & $\langle \bar{g}_2\rangle$
\\
\hline
 \myStrut{4} n, d& $\langle g_a, g_s, g_1,\dots,g_n,\bar{g}_1,\bar{g}_2\rangle$  &$\langle g_s,  g_1,\dots,g_n,\bar{g}_1,\bar{g}_2\rangle$  & $\langle  g_1,\dots,g_n, \bar{g}_1,\bar{g}_2\rangle $ & $\langle \bar{g}_1,\bar{g}_2\rangle$ & $\langle \bar{g}_2\rangle$
\\
\hline
\end{tabular}
\end{center}
\normalsize

\vspace{0.1 in}

\noindent where in the left hand column  the notation ``y'' (respectively, ``n'')  signifies that $\sqrt{d}\in\elt$ (respectively, $\sqrt{d}\notin \elt$) and the notation ``s'' (respectively, ``d'') signifies that the orbit is a singlet (respectively, doublet).
 In the case of doublets $g_s$ switches between the $a$ orbit and the $b$ orbit while the generators $g_1,\dots,g_n,\bar{g}_1$ and (when defined) $\bar{g}_2$  are orbit-preserving. Aside from $g_a$ the generators are mutually commuting. For every orbit
 \begin{align}
 g_a^2 &= \bar{g}_1^2 = e  & g\vpu{-1}_a \bar{g}\vpu{-1}_1 g^{-1}_a& = g_c
\intertext{For those orbits for which $\bar{g}_2$ is defined we also have}
 \bar{g}_2^2 & = e  &   g\vpu{-1}_a \bar{g}\vpu{-1}_2 g^{-1}_a & = \bar{g}_2
 \end{align}
Remaining, orbit-specific relations are given in box 3 (aside from $g_s^2 = g_1$ which, in the case of those doublets for which it holds, is given in box 1).

For the $F$ and $G$ matrices corresponding to the generators $g_1, \dots, g_n$ we write
\begin{align}
F_{g_j} & = F_j  & G_{g_j} & = G_j
\end{align}
and tabulate them in box $4$.  In the case of a doublet we distinguish the two orbits with an additional label, writing $F_{aj}, G_{aj}$ for orbit $a$ and $F_{bj}$, $G_{bj}$ for orbit $b$. 
For the $F$ and $G$ matrices corresponding to the generator $\bar{g}_1$ we have, in every case,
\begin{align}
F_{\bar{g}_1} & =-J  & G_{\bar{g}_1} & = I
\\
\intertext{For those orbits for which $\bar{g}_2$ is defined we  also have}
F_{\bar{g}_2} &  = 
I 
& 
G_{\bar{g}_2} & = I
\end{align}
When $d=0$ (mod 3) we  write
\begin{align}
\qv_{g_j} & = \qv_j &  \qv_{\bar{g}_1} & = \bar{\qv}_1 & \rv_{g_j} & = \rv_j & \rv_{\bar{g}_1} & = \bar{\rv}_1
\end{align}
 and tabulate the vectors in box 7. As with the $F$ and $G$ matrices we add an additional label $a$ or $b$ in the case of a doublet.  When $\bar{g}_2$ is defined we have
\be
\qv_{\bar{g}_2} = \rv_{\bar{g}_2} = \boldsymbol{0}
\ee
in every case.
For doublets the orbit-switching automorphism $g_s$ acts according to
\begin{align}
g_s (\Pi_a) &=
\begin{cases}
U\vpu{\dagger}_{F_{as}} \Pi_bU^{\dagger}_{F_{as}} \qquad & d\neq 0 \text{ (mod $3$)}
\\
 D\vpu{\dagger}_{\qv_{as}}U\vpu{\dagger}_{F_{as}} \Pi_bU^{\dagger}_{F_{as}} D^{\dagger}_{\qv_{as}} \qquad & d = 0 \text{ (mod $3$)}
 \end{cases}
\\
g_s (\Pi_b) &= 
\begin{cases}
U\vpu{\dagger}_{F_{bs}} \Pi_aU^{\dagger}_{F_{bs}} & d\neq 0 \text{ (mod $3$)}
\\
D\vpu{\dagger}_{\qv_{bs}} U\vpu{\dagger}_{F_{bs}} \Pi_aU^{\dagger}_{F_{bs}}D^{\dagger}_{\qv_{bs}}  \qquad & d = 0 \text{  (mod $3$)}
\end{cases}
\\
g_s\Bigl(\Tr\left( D_{\mathbf{p}}\Pi_a\right) \Bigr)& =
\begin{cases}
\Tr\left( D_{G_{as}\mathbf{p}}\Pi_b\right) \qquad & d\neq 0 \text{ (mod $3$)}
\\
 \sigma^{\langle \rv_{as},\mathbf{p}\rangle} \Tr\left( D_{G_{as}\mathbf{p}}\Pi_b\right) \qquad & d = 0 \text{ (mod $3$)}
\end{cases}
\\
g_s\Bigl(\Tr\left( D_{\mathbf{p}}\Pi_b\right) \Bigr)& = 
\begin{cases}
 \Tr\left( D_{G_{bs}\mathbf{p}}\Pi_a\right)  \qquad & d \neq 0 \text{ (mod $3$)}
\\
\sigma^{\langle \rv_{bs},\mathbf{p}\rangle} \Tr\left( D_{G_{bs}\mathbf{p}}\Pi_a\right)  \qquad & d = 0 \text{ (mod $3$)}
\end{cases}
\end{align}
where $F_{as}$, $F_{bs}$, $G_{as}$, $G_{bs}$ are tabulated in box 4 and $\qv_{as}$, $\qv_{bs}$, $\rv_{as}$, $\rv_{bs}$ are tabulated in box 7 (when defined).

In every case the groups $\tstb{\Pi}$, $\bstb{\Pi}$ are cyclic:
\begin{align}
\tstb{\Pi} & = \langle F_0 \rangle  &  \bstb{\Pi} & = \langle G_0 \rangle 
\end{align}
The matrices $F_0$, $G_0$ are specified in box 4.   Note that in a doublet $F_0$, $G_0$ are the same for both orbits. For singlets
\begin{align}
\cbstb{\Pi} &= \langle G_0,G_1,\dots , G_n\rangle
\\
\intertext{while for doublets}
\cbstb{\Pi_a}& =\cbstb{\Pi_b} =  \langle G_0,G_{a1},\dots , G_{an}\rangle= \langle G_0,G_{b1},\dots , G_{bn}\rangle
\end{align}
The orders of the symmetry groups and centralizers are given in box 5.   Relations for the Abelian group $\cbstb{\Pi}$ ($\cbstb{\Pi_a}=\cbstb{\Pi_b}$ in the case of a doublet) are given in box 6. 

\vspace{0.2 in}

\begin{table}[H]
\SMALL
\begin{center}
\begin{tabular}{|c c c|c c c| c c c |}
\hline
\parbox{1 pt}{\rule{0 ex}{5 ex}}
orbit & $\mathbf{p}$ & $F$ \hspace{1 ex} & 
\hspace{1 ex}
orbit &  $\mathbf{p}$ & $F$& 
\hspace{1 ex}
orbit &  $\mathbf{p}$ & $F$
\\
\hline
$4a$ & \myStrut{5} $\bsmt 3 \\ 2\esmt$ & \myStrut{5} $\bsmt 1 & 1 \\ 3 & 4 \esmt$
& 
$5a$ & \myStrut{5} $\bsmt 1 \\ 4\esmt$ & \myStrut{5} $\bsmt 4 & 1 \\ 2 & 2 \esmt$
& 
$6a$ & \myStrut{5} $\bsmt 0 \\ 1\esmt$ & \myStrut{5} $\bsmt 9 & 8 \\ 4 & 1 \esmt$
\\
\hline
\end{tabular}
\end{center}
\caption{}
\label{tble:EtoNTransform}
\end{table}

\newpage

\begin{landscape}

\vspace{0.2 in}

\begin{table}[hbt]
\SMALL
\begin{center}
\begin{tabular}{|c | c | c| c | c|c | c | c|}
\hline
\myStrut{4} $d$ & $a$ & $r_1$ & $r_2$ & $r_3$ & $r_4$ & $t$ & minimal polynomial of $t$ over $\mathbb{Q}(a,r_1,\dots, r_j)$
\\
\hline
\myStrut{4}$4$ & $\sqrt{5}$ & $\sqrt{2}$ & & & & &
\\
\hline
\myStrut{4}$5$ & $\sqrt{3}$ & $\sqrt{5}$ & & & & $\sin\frac{\pi}{5}$& $8x^2-(5-r_1)$
\\
\hline
\myStrut{4}$6$ & $\sqrt{21}$ & $\sqrt{2}$ & $\sqrt{3}$ & & & & 
\\
\hline
\myStrut{4}$7$ & $\sqrt{2}$ & $\sqrt{7}$ &  & & & $\cos\frac{\pi}{7}$& $8x^3-4x^2-4x+1$
\\
\hline
\myStrut{4}$8$ & $\sqrt{5}$ & $\sqrt{2}$ &  & & & $\cos\frac{\pi}{8}$& $4x^2-(r_1+2)$
\\
\hline
\myStrut{4}$9$ & $\sqrt{15}$ & $\sqrt{3}$ &  & & & $\cos\frac{\pi}{9}$& $8x^3-6x-1$
\\
\hline
\myStrut{4}$10$ & $\sqrt{77}$ & $\sqrt{2}$ & $\sqrt{5}$ & $\sqrt{7}$ & & $\cos\frac{\pi}{10}$& $8x^2-(r_2+5)$
\\
\hline
\end{tabular}
\end{center}
\caption{$a$, $r$ and $t$ generators\label{tble:artGenerators}}
\end{table}

\newpage

\begin{table}[hbt]
\tiny
\begin{center}
\begin{tabular}{|c | c | c| c |}
\hline
 \myStrut{4}$d$ & $b_1$ & $b_2$ & $b_3$
\\
\hline
\myStrut{6} $4$ & $i\sqrt{a+1}$ & & 
\\
\hline
\myStrut{6} $5$ & $i\sqrt{5a+(5+3r_1)t}$ & & 
\\
\hline
\myStrut{6} $6$ & $i\sqrt{2a+6}$ &  $2\re\Bigl((1+i\sqrt{7})^{\frac{1}{3}} \Bigr)$ & 
\\
\hline
\myStrut{6} $7$ & $i\sqrt{2a+1}$ &  & 
\\
\hline
\myStrut{6} $8$ & $i\sqrt{a-1}$ & $\sqrt{3}$ & $\sqrt{6+r_1b_2}$  
\\
\hline
\myStrut{6} $9$ & $i\sqrt{2a+4r_1}$ &  $2\re\Bigl((6+6i\sqrt{5})^{\frac{1}{3}} \Bigr)$ &  
\\
\hline
\myStrut{6} $10$ & $i\sqrt{ 14(5 + 5a - 5r_2 - ar_2) -8(21 + 2a - 14r_2 -ar_2)r_3t}$ &  
$\frac{2}{3}\re\Bigl((405+15i\sqrt{231})^{\frac{1}{3}} \Bigr)$ &  
\\
\hline
\myStrut{6} $11$ &   $\sqrt{11+2r_2}$ & $i\sqrt{22a-6r_2+3ar_2-22+(4-6a+2r_2)b_1}$ & 
\\
\hline
\myStrut{6} $12$ & $i\sqrt{a-1}$ &  $i\sqrt{a-3}$ & $2\re\Bigr( (5+i\sqrt{39})^{\frac{1}{3}}\Bigr) $
\\
\hline
\myStrut{6} $13$ &   $\sqrt{250-40a+50r_1-8a r_1}$ & $i \sqrt{205+40a+95r_1+16ar_1+(30+5a+12r_1+2ar_1)b_1}$ &
\\
\hline
\myStrut{6} $14$ &   $i\sqrt{-250+20a+50r_1-4ar_1}$ & $\re\biggl(\Bigl(908-132a+12i\sqrt{4962-234a}\Bigr)^{\frac{1}{3}}\biggr)$ &
\\
\hline
\myStrut{6} $15$ &   $i\sqrt{-30+28a-6r_1+4ar_1+(-60+24a-12r_1+8ar_1) t}$ & $2\re\Bigl((10+30i)^{\frac{1}{3}}\Bigr)$ &
\\
\hline
\end{tabular}

\myStrut{5}

\begin{tabular}{|c | c | c| c | c|}
\hline
\myStrut{4}$d$ & $b_1$ & $b_2$ & $b_3$ & $b_4$
\\
\hline
\myStrut{6} $16$ &   $i\sqrt{a-11}$ & $\sqrt{ar_1+195}$ & $\sqrt{195-3 a r_1-(a-4r_1)b_2}$ & 
$\sqrt{29744-2704r_1-\bigl(936-936r_1-16ar_1+(65+5a-3r_1-3ar_1)b_2\bigr)b_3}$
\\
\hline
\myStrut{6} $19$ &   $i\sqrt{2a+1}$ &&  & 
\\
\hline
\myStrut{6} $24$ &   $i\sqrt{a-3}$ & $\sqrt{(4+r_1)((3+r_2)}$& $2\re\Bigl((8+8i\sqrt{7})^{\frac{1}{3}}\Bigr)$  & 
\\
\hline
\myStrut{6} $28$ &   $i\sqrt{a+1}$ & $i\sqrt{a+5}$& $2\re\Bigl((189+21i\sqrt{87})^{\frac{1}{3}}\Bigr)$ & 
\\
\hline
\myStrut{6} $35$ &   $i\sqrt{2a+1}$ & $2\re\Bigl((280+210i\sqrt{6})^{\frac{1}{3}}\Bigr)$&  & 
\\
\hline
\myStrut{6} $48$ &   $i\sqrt{a-1}$ & $\sqrt{6+r_1r_2}$& $\sqrt{6-2r_1r_2+(r_2-r_1)b_2}$  & 
\\
\hline
\end{tabular}
\end{center}
\caption{$b$ generators \label{tble:bGenerators}}
\end{table}

\end{landscape}

\begin{table}[htpb]
\SMALL
\begin{center}
\begin{tabular}{|c | c | c| c |}
\hline
\myStrut{7} $d$ & generator & in terms of radicals & minimal polynomial 
\\
\hline
\myStrut{7} $6$ & $b_2$ & $2\re\Bigl(\bigl(1+i\sqrt{7}\bigr)^{\frac{1}{3}} \Bigr)$ & $x^3-6x-2$
\\
\hline
\myStrut{7} $9$ & $b_2$ & $2\re\Bigl( (6+6i\sqrt{5})^{\frac{1}{3}} \Bigr)$ & $x^3-18x-12$
\\
\hline
\myStrut{7} $10$ & $b_2$ & $\frac{2}{3}\re\biggl( \Bigl(405+15i\sqrt{231}\Bigr)^{\frac{1}{3}} \biggr)$ & $x^3-20x-30$
\\
\hline
\myStrut{7} $12$ & $b_3$ & $2\re\Bigl( (5+i\sqrt{39})^{\frac{1}{3}} \Bigr)$ & $x^3-12x-10$
\\
\hline
\myStrut{7} $14$ & $b_2$ & $\re\biggl( \Bigl(908-132a+12i\sqrt{4962-234a}\Bigr)^{\frac{1}{3}} \biggr)$ & $x^3+(3a-111)x+(33a-227)$
\\
\hline
\myStrut{7} $15$ & $b_2$ & $2\re\Bigl( (10+30i)^{\frac{1}{3}} \Bigr)$ & $x^3-30x-20$
\\
\hline
\myStrut{7} $24$ & $b_3$ & $2\re\Bigl( (8+8i\sqrt{7})^{\frac{1}{3}} \Bigr)$ & $x^3-24x-16$
\\
\hline
\myStrut{7} $28$ & $b_3$ & $2\re\Bigl( (189+21i\sqrt{87})^{\frac{1}{3}} \Bigr)$ & $x^3-126x-378$
\\
\hline
\myStrut{7} $35$ & $b_2$ & $2\re\Bigl( (280+210i\sqrt{6})^{\frac{1}{3}} \Bigr)$ & $x^3-210x-560$
\\
\hline
\end{tabular}
\end{center}
\caption{Minimal polynomials of $b$ generators where these are degree $3$ \label{tble:minPolyGt2}}
\end{table}

\begin{table}[H]
\SMALL
\begin{center}
\begin{tabular}{|c | c | c| c | c|c |}
\hline
\myStrut{4} orbit & $\eC$ & $\eZ$ & $\eZb$ & $\elt$ & $\el$ 
\\
\hline
\myStrut{4}$4a$ & $\mathbb{Q}(a)$ & $\eC$ & $\eZ(r_1,b_1)$ & $\eZb(i)$ & $\elt$ 
\\
\hline
\myStrut{4}$5a$ & $\mathbb{Q}(a)$ & $\eC$ & $\eZ(r_1,t,b_1)$ & $\eZb(i)$ & $\elt$  
\\
\hline
\myStrut{4}$6a$ & $\mathbb{Q}(a)$ & $\eC$ & $\eZ(r_2,b_1,b_2)$ & $\eZb(i)$ & $\elt(r_1)$  
\\
\hline
\myStrut{4}$7a$ & $\mathbb{Q}(a)$ & $\eC$ & $\eZ(r_1,t,b_1)$ & $\eZb(i)$ & $\elt$  
\\
\hline
\myStrut{4}$7b$ & $\mathbb{Q}(a)$ & $\eC$ & $\eZ(t,b_1)$ & $\eZb(i r_1)$ & $\elt(r_1)$ 
\\
\hline
\myStrut{4}$8a$ & $\mathbb{Q}(a)$ & $\eC$ & $\eZ(r_1,t,b_1,b_2,b_3)$ & $\eZb(i)$ & $\elt$  
\\
\hline
\myStrut{4}$8b$ & $\mathbb{Q}(a)$ & $\eC$ & $\eZ(r_1,t,b_1)$ & $\eZb(i)$ & $\elt$  
\\
\hline
\end{tabular}
\end{center}
\caption{Fields \label{tble:fields}}
\end{table}

\begin{table}[htpb]
\SMALL
\begin{center}
\begin{tabular}{|c|c|}
\hline
$d$ & \myStrut{4} $\tau$
\\
\hline
$4$ &  \myStrut{4} $-\frac{1}{2} r_1-\frac{1}{2} ir_1$
\\
\hline
$5$ &  \myStrut{4} $-\frac{1}{4}(r_1+1)-it$
\\
\hline
$6$ &  \myStrut{4} $-\frac{1}{2}r_2-\frac{1}{2}i$
\\
\hline
$7$ &  \myStrut{4} $-t-\frac{1}{7}ir_1(4t^2+t-3)$
\\
\hline
$8$ &  \myStrut{4} $-t-it(r_1-1)$
\\
\hline
$9$ &  \myStrut{4} $-t-\frac{1}{3}ir_1(4t^2-t-2)$
\\
\hline
$10$ &  \myStrut{4} $-t-\frac{1}{4}i(r_2-1)$
\\
\hline
$11$ &  \myStrut{4} $-t-\frac{1}{22}ir_1r_2(16t^4-16t^3-12t^2+11t+2)$
\\
\hline
$12$ &  \myStrut{4} $-\frac{1}{52}r_1(ar_2+13)-\frac{1}{52}ir_1(ar_2-13)$
\\
\hline
$13$ &  \myStrut{4} $-\frac{1}{2}\bigl(-16t^4+2(7-r_2)t^2+r_2-2\bigr)-it$
\\
\hline
$14$ &  \myStrut{4} $-t-i(4t^2-r_4t-1)$
\\
\hline
$15$ &  \myStrut{4} $-t-\frac{1}{6}ia(2t-r_1+1)$
\\
\hline
$16$ &  \myStrut{4} $-t+it(-4r_2t^2+3r_2+1)$
\\
\hline
$19$ &  \myStrut{4} $-t-\frac{1}{19}ir_1(128t^7-224t^5+16t^4+104t^3-20t^2-7t+5)$
\\
\hline
$24$ &  \myStrut{4} $-t-it(r_1r_2+r_1-r_2-2)$
\\
\hline
$28$ &  \myStrut{4} $-t-\frac{1}{2}i\bigl(-4r_1(3-r_2)t^2-2(8-3r_2)t+r_1(r_2-1)\bigr)$
\\
\hline
$35$ & \myStrut{4} $-t-\frac{1}{7}ir_2\bigl(16(r_1-1)t^5-16t^4-20(r_1-1)t^3+16t^2+(5r_1-4)t-2\bigr)$
\\
\hline
$48$ & \myStrut{4} $-t-it\bigl(4r_1(r_2+1)t^2-3r_1r_2-3r_1-r_2-2\bigr)$
\\
\hline
\end{tabular}
\end{center}
\caption{$\tau$ in terms of the  generators $a,r_1,\dots, r_j,i$ \label{tble:tauVals}}
\end{table}

\SMALL

\textbf{Orbit $4a$}

\nopagebreak

\vspace{0.15 in}

\nopagebreak

\begin{tabular}{l|c|c|c|c|c|c |c|c|c|c|c|c|c|c|c|}
\cline{2-6} \cline{8-10}
1. & \myStrut{4} & $a$ & $r_1$ &$b_1$ & $i$ 
& \spn{0.1}{2} & \multicolumn{3}{|c |}{$b'_1 = \sqrt{a-1}$} 
\\
\cline{2-6} \cline{8-10}
& \myStrut{4} $g_a$ & $-a$ & $r_1$ & $b'_1$ & $i$ 
\\
\cline{2-6} \cline{8-16}
& \myStrut{4} $g_1$ & $a$ & $-r_1$ & $b_1$ & $i$
& \spn{0.1}{3} & \multicolumn{9}{|c |}{$g_1^2 = g_2^2  = e \quad g\vpu{-1}_a g\vpu{-1}_1 g_a^{-1} = g_1\quad g\vpu{-1}_a g\vpu{-1}_2 g_a^{-1} = g\vpu{-1}_2$}
\\
\cline{2-6} \cline{8-16}
& \myStrut{4} $g_2$ & $a$ & $r_1$ & $-b_1$ & $i$ 
\\
\cline{2-6} \cline{8-13}
& \myStrut{4} $\bar{g}_1$ & $a$ & $r_1$ & $b_1$ & $-i$ 
& \spn{0.1}{4}& $G_0$ & $G_1$ & $G_2$ &$F_0$ & $F_1$ & $F_2$
\\
\cline{2-6} \cline{8-13}
& \multicolumn{5}{|c |}{\myStrut{4} $g_c = g_2\bar{g}_1$} 
&& $\bsmt 6 & 7 \\ 1 & 5 \esmt$ & $\bsmt 6 & 1 \\ 7 & 7 \esmt$
& $\bsmt 7 & 0 \\ 0 & 7 \esmt$ &  $\bsmt 2 & 1 \\ 7 & 3 \esmt$ & $\bsmt 3 & 3 \\ 1 & 6 \esmt$
& $\bsmt 7 & 0 \\ 0 & 7 \esmt$
\\
\cline{2-6} \cline{8-13}
\multicolumn{6}{c}{} 
\end{tabular}

\begin{tabular}{l|c|c|c|c|c|c|c|c|c|c|c|}
\cline{2-4} \cline{6-6}
5. &\multicolumn{2}{c|}{$|\tstb{\Pi}|$,  $|\bstb{\Pi}|$} & \myStrut{4}  $|\cbstb{\Pi}|$
& \multicolumn{1}{r}{\te{1.15}{$\vphantom{x}$} 6.}  & \multicolumn{1}{|c|}{$G_0^{12} = I \qquad G_1^2 = G_0^2 \qquad G_2^2 = I$}
\\
\cline{2-4} \cline{6-6}
&  \multicolumn{2}{c|}{$12$} & \myStrut{4}  $48$ 
\\
\cline{2-4}
\end{tabular}

\vspace{0.2 in}

\textbf{Orbit $5a$} 

\vspace{0.15 in}

\begin{tabular}{l|c|c|c|c|c|c|cl|c|c|c|c|c r|c|c|c|}
\cline{2-7}\cline{9-14}
 1. &\myStrut{4} & $a$ & $r_1$ & $t$ &  $b_1$ & $i$ 
& \spn{0.1}{2} & \multicolumn{6}{|c |}{$t'= \sin\frac{3\pi}{5}$} 
\\
\cline{2-7}
&\myStrut{4} $g_a$ & $-a$ & $r_1$ & $t$ & $b'_1$ & $i$ 
&&\multicolumn{6}{|c |}{\myStrut{4}$b'_1 = \sqrt{5a-(5+3r_1)t}  \quad b''_1 = i\sqrt{5a+(5-3r_1)t'}$} 
\\
\cline{2-7} \cline{9-14}
&\myStrut{4} $g_1$ & $a$ & $-r_1$ & $t'$ & $b''_1$ & $i$ 
\\
\cline{2-7} \cline{9-12}
&\myStrut{4} $\bar{g}_1$ & $a$ & $r_1$ & $t$ & $b_1$ & $-i$ 
&\spn{0.1}{3}&\multicolumn{4}{|c |}{\myStrut{4}$g_1^8 = e \qquad g\vpu{-1}_a g\vpu{-1}_1 g_a^{-1} = g^5_1$}
\\
\cline{2-7} \cline{9-12}
&\multicolumn{6}{|c |}{\myStrut{4} $g\vpu{4}_c =g_1^4\bar{g}\vpu{4}_1 $} 
\\
\cline{2-7} 
\multicolumn{6}{c}{} 
\end{tabular}

\begin{tabular}{r|c|c|c|c|c|c|c|c|c|c|}
\cline{2-5} \cline{7-9} \cline{11-11}
4. & \myStrut{4} $G_0$ & $G_1$ &$F_0$ & $F_1$ 
&\spn{0.1}{5} & \multicolumn{2}{c|}{$|\tstb{\Pi}|$, $|\bstb{\Pi}|$} & \myStrut{4} $|\cbstb{\Pi}|$ 
& \spn{0.1}{6} & $G_0^{3} = I \qquad G_1^8 = I$
\\
\cline{2-5}  \cline{7-9} \cline{11-11}
 & \myStrut{4} $\bsmt 0 & 4 \\ 1 & 4 \esmt$ & $\bsmt 1 & 3 \\ 2 & 4 \esmt$ & $\bsmt 0 & 4 \\ 1 & 4 \esmt$ & $\bsmt 3 & 4 \\ 3 & 1 \esmt$ 
 && \multicolumn{2}{c|}{$3$} & $24$
\\
\cline{2-5} \cline{7-9}
\end{tabular}

\vspace{0.2 in}

\textbf{Orbit $6a$}

\nopagebreak

\vspace{0.15 in}

\nopagebreak

\begin{tabular}{l|c|c|c|c|c|c|c|cr|c|c|c|c|c|c|c|}
\cline{2-8} \cline{10-15}
1. & \myStrut{5.2} & $a$ & $r_1$ & $r_2$ & $b_1$ & $b_2$ & $i$ 
& \spn{0.1}{2} & 
\multicolumn{6}{|c |}{$b'_1= \sqrt{2a-6}$}
\\
\cline{2-8} \cline{10-15}
&\myStrut{5} $g_a$ & $-a$ & $r_1$ & $r_2$ & $b'_1$ & $b_2$ &  $i$ 
&& 
\multicolumn{6}{|c|}{$b'_2= \frac{2}{3}\re\Bigl(e^{\frac{2\pi i}{3}} \left( 27+9iar_2\right)^{\frac{2}{3}} \Bigr)$}
\\
\cline{2-8} \cline{10-15}
&\myStrut{4.7} $g_1$ & $a$ & $r_1$ & $r_2$ & $-b_1$ & $b'_2$ &  $i$ 
\\
\cline{2-8} \cline{10-15}
&\myStrut{4.7} $g_2$ & $a$ & $r_1$ & $-r_2$ & $b_1$ & $b_2$ &  $i$ 
&\spn{0.1}{3} & \multicolumn{6}{|c|}{$ g_1^6 = g^2_2  = e$}
\\
\cline{2-8}
&\myStrut{4.7} $\bar{g}_1$ & $a$ & $r_1$ & $r_2$ & $b_1$ & $b_2$ &  $-i$ 
&& \multicolumn{6}{|c|}{$ g\vpu{-1}_a g\vpu{-1}_1 g^{-1}_a  =g_1^5 \quad g\vpu{-1}_a g\vpu{-1}_2 g^{-1}_a  =g_1^3g\vpu{3}_2$}
\\
\cline{2-8} \cline{10-15}
&\myStrut{4.7} $\bar{g}_2$ & $a$ & $-r_1$ & $r_2$ & $b_1$ & $b_2$ &  $i$ 
\\
\cline{2-8} 
&\multicolumn{7}{|c |}{\myStrut{4.7} $g\vpu{4}_c = g^3_1 \bar{g}\vpu{3}_1 $} 
\\
\cline{2-8} 
\multicolumn{8}{c}{} 
\end{tabular}

\begin{tabular}{r|c|c|c|c|c|c|c|c|c|c|c|}
\cline{2-7} \cline{9-11}
4.& \myStrut{4} $G_0$ & $G_1$ & $G_2$ & $F_0$ & $F_1$ & $F_2$ 
&\spn{0.1}{5}&
\multicolumn{2}{c|}{ $|\tstb{\Pi}|$, $|\bstb{\Pi}|$} & $|\cbstb{\Pi}|$ 
\\
\cline{2-7} \cline{9-11}
&\myStrut{5} $\bsmt 0 & 5 \\ 7 & 5 \esmt$ & $\bsmt 4 & 11 \\ 1 & 3 \esmt$ & $\bsmt 7 & 11 \\ 1 & 6 \esmt$ &$\bsmt 0 & 5 \\ 7 & 5 \esmt$ & $\bsmt 3 & 1 \\ 11 & 4 \esmt$ & $\bsmt 6 & 5 \\ 11 & 7 \esmt$ 
&& \multicolumn{2}{c|}{ 6}  & 72
\\
\cline{2-7} \cline{9-11}
\multicolumn{8}{c}{} 
\end{tabular}

\begin{tabular}{r|c|c|c|c|c|c|c|c|c|c|c|}
\cline{2-5} \cline{7-12}
6.& \multicolumn{4}{|c|}{$G_0^{6} = I \qquad G_1^6 = I \qquad G_2^2 = G^4_0$}
& \spn{0.1}{7} &\myStrut{4} $\rv_{1}$  & $\rv_{2}$  & $\bar{\rv}_{1}$  & $\qv_{1}$ & $\qv_{2}$ & $\bar{\qv}_{1}$
\\
\cline{2-5} \cline{7-12}
\multicolumn{6}{c|}{} & \myStrut{5} $\bsmt 2 \\1 \esmt$ & $\bsmt 0 \\ 0 \esmt$ & $\bsmt 1 \\ 2\esmt$ &$\bsmt 2 \\ 4\esmt$ & $\bsmt 0 \\ 0 \esmt$ & $\bsmt 2 \\ 2\esmt$ 
\\
 \cline{7-12}
\end{tabular}

\vspace{0.2 in}

\textbf{Orbit $7a$}

\nopagebreak

\vspace{0.15 in}

\nopagebreak

\begin{tabular}{l|c|c|c|c|c|c|cl|c|c|c|c|c|c|c|c|c|}
\cline{2-7} \cline{9-12}
 1. & \myStrut{4} & $a$ & $r_1$ &$t$ & $b_1$  & $i$ 
& \spn{0.1}{2} & 
\multicolumn{4}{c|}{$t'=\cos\frac{3\pi}{7} \qquad b'_1= \sqrt{2a-1}$}
\\
\cline{2-7} \cline{9-12}
&\myStrut{4} $g_a$ & $-a$ & $r_1$ & $t$ & $b'_1$ &  $i$ 
\\
\cline{2-7} \cline{9-17}
&\myStrut{4} $g_1$ & $a$ & $r_1$ & $t'$ & $-b_1$ &  $i$  
&\spn{0.1}{3} & 
\multicolumn{9}{c|}{$g_1^6 = g^2_2  = e \quad g\vpu{-1}_a g\vpu{-1}_1 g^{-1}_a  =g\vpu{3}_1
\quad
g\vpu{-1}_a g\vpu{-1}_2 g^{-1}_a  =g_2$}
\\
\cline{2-7} \cline{9-17}
&\myStrut{4} $g_2$ & $a$ & $-r_1$ & $t$ & $b_1$ &  $-i$  
\\
\cline{2-7} \cline{9-14}
&\myStrut{4} $\bar{g}_1$ & $a$ & $r_1$ & $t$ & $b_1$ &  $-i$ 
&\spn{0.1}{4} &  
$G_0$ & $G_1$ & $G_2$ &  $F_0$ & $F_1$ & $F_2$
\\
\cline{2-7} \cline{9-14}
&\multicolumn{6}{|c |}{\myStrut{4} $g\vpu{4}_c =  g^3_1\bar{g}\vpu{3}_1$}
&\multicolumn{1}{r|}{}& 
$\bsmt 0 & 6 \\ 1 & 6\esmt$ & $ \bsmt 5 & 0 \\ 0 & 5\esmt$ & $\bsmt 1 & 3 \\ 4 & 4\esmt$ &$\bsmt 0 & 6 \\ 1 & 6 \esmt$ & $ \bsmt 3 & 0 \\ 0 & 5 \esmt $  &  $\bsmt 4 & 4 \\ 3 & 1\esmt$
\\
\cline{2-7} \cline{9-14}
\multicolumn{8}{}{}
\end{tabular}

\begin{tabular}{l|c|c|c|c|c|c|cl|c|c|c|c|c|c|c|c|c|}
\cline{2-4} \cline{9-12}
5.&  \multicolumn{2}{c|}{$|\tstb{\Pi}|$, $|\bstb{\Pi}|$} & \myStrut{4} $|\cbstb{\Pi}|$ 
& \multicolumn{4}{r}{\te{1.7}{$\phantom{x}$} 6.} &
\multicolumn{4}{|c|}{$G_0^{3} = I \qquad G_1^6 = I \qquad G_2^2 = G^2_0$}
\\
\cline{2-4} \cline{9-12}
&\multicolumn{2}{c|}{$3$} &  \myStrut{4}  $36$
\\
\cline{2-4}
\end{tabular}

\vspace{0.2 in}

\textbf{Orbit $7b$}

\nopagebreak

\vspace{0.15 in}

\nopagebreak

\begin{tabular}{l|c|c|c|c|c|c|c|c|c|c|c|}
\cline{2-7} \cline{9-12}
 1. &\myStrut{4} & $a$ & $r_1$ & $t$ & $b_1$  & $i$
&\spn{0.1}{2} & 
\multicolumn{4}{c|}{$t'=\cos\frac{3\pi}{7} \quad b'_1= \sqrt{2a-1}$}
\\
\cline{2-7} \cline{9-12}
&\myStrut{4} $g_a$ & $-a$ & $r_1$ & $t$ & $b'_1$ &  $i$  
\\
\cline{2-7} \cline{9-11}
&\myStrut{4} $g_1$ & $a$ & $r_1$ & $t'$ & $-b_1$ &  $i$  
&\spn{0.1}{3} & 
\multicolumn{3}{c|}{$g_1^6  = e \quad g\vpu{-1}_a g\vpu{-1}_1 g^{-1}_a  =g\vpu{3}_1$}
\\
\cline{2-7} \cline{9-11}
&\myStrut{4} $\bar{g}_1$ & $a$ & $r_1$ & $t$ & $b_1$ &  $-i$ 
\\
\cline{2-7} \cline{9-12}
&\myStrut{4} $\bar{g}_2$ & $a$ & $-r_1$ & $t$ & $b_1$ &  $-i$ 
& \spn{0.1}{4} &   $G_0$ & $G_1$  & $F_0$ & $F_1$
\\
\cline{2-7} \cline{9-12}
&\multicolumn{6}{|c |}{\myStrut{4} $g\vpu{4}_c =  g^3_1\bar{g}\vpu{3}_1$}
&&
$\bsmt 4 & 4 \\ 3 & 1\esmt$ & $ \bsmt 5 & 0 \\ 0 & 5\esmt$ & $\bsmt 3 & 3 \\ 4 & 6 \esmt$ & $\bsmt 3 & 0 \\ 0 & 5\esmt$
\\
\cline{2-7} \cline{9-12}
\end{tabular}

\vspace{0.15 in}

\begin{tabular}{l|c|c|c|c|c|c|}
\cline{2-4} \cline{7-7}
5. &\multicolumn{2}{c|}{$|\tstb{\Pi}|$, $|\bstb{\Pi}|$} & \myStrut{4} $|\cbstb{\Pi}|$
& \multicolumn{2}{r|}{\te{1.7}{$\phantom{x}$} 6.} & $G_0^{6} = I \qquad G_1^6 = I$ 
\\
\cline{2-4} \cline{7-7}
& \multicolumn{2}{c|}{$6$} & \myStrut{4} $36$
\\
\cline{2-4}
\end{tabular}

\vspace{0.2 in}

\textbf{Orbit $8a$} 

\nopagebreak

\vspace{0.15 in}

\nopagebreak

\begin{tabular}{r|c|c|c|c|c|c|c|c|c|c|c|c|c|c|c|}
\cline{2-9} \cline{11-16}
1. & \myStrut{4} & $a$ & $r_1$ &$t$ & $b_1$ & $b_2$ & $b_3$   & $i$ 
&\spn{0.1}{2} & 
\multicolumn{6}{c|}{$t'=\cos\frac{3\pi}{8}$}
\\
\cline{2-9} 
&\myStrut{4} $g_a$ & $-a$ & $r_1$ & $t$ & $b'_1$ & $b_2$ & $b_3$ & $i$ 
&& 
\multicolumn{6}{c|}{$b'_1= \sqrt{a+1} \quad b'_3= \sqrt{6-r_1b_2}$}
\\
\cline{2-9} \cline{11-16}
&\myStrut{4} $g_1$ & $a$ & $-r_1$ & $t'$ & $b_1$ & $-b_2$ & $-b_3$ & $i$ 
\\
\cline{2-9} \cline{11-14}
&\myStrut{4} $g_2$ & $a$ & $r_1$ & $t$ & $b_1$ & $-b_2$ & $b'_3$ & $i$
& \spn{0.1}{3} & \multicolumn{4}{c|}{$g_1^4 = g_2^4 = g_3^2  = e$}
\\
\cline{2-9}
&\myStrut{4} $g_3$ & $a$ & $r_1$ & $t$ & $-b_1$ & $b_2$ & $b_3$ & $i$
&& \multicolumn{4}{c|}{$g\vpu{-1}_a g\vpu{-1}_1 g^{-1}_a  =g\vpu{-1}_1$ }
\\
\cline{2-9}
&\myStrut{4} $\bar{g}_1$ &  $a$ & $r_1$ & $t$ & $b_1$ & $b_2$ & $b_3$ & $-i$
&& \multicolumn{4}{c|}{$
g\vpu{-1}_a g\vpu{-1}_2 g^{-1}_a  =g^3_2 $}
\\
\cline{2-9}
&\multicolumn{8}{|c |}{\myStrut{4} $g\vpu{4}_c =g_3\bar{g}_1$} 
&& \multicolumn{4}{c|}{$g\vpu{-1}_a g\vpu{-1}_3 g^{-1}_a  =g\vpu{-1}_3$}
\\
\cline{2-9} \cline{11-14}
\end{tabular}

\vspace{0.1 in}

\begin{tabular}{r|c|c|c|c|c|c|c|c|}
\cline{2-9}
4. & \myStrut{4} $G_0$ & $G_1$ & $G_2$ & $G_3$ &$F_0$ & $F_1$  & $F_2$ & $F_3$
\\
\cline{2-9}
& \myStrut{5} $\bsmt 0 & 7 \\ 9& 7\esmt$ & $\bsmt 2 & 3 \\ 13 & 5\esmt$ & $\bsmt 13 & 9 \\ 7 & 6\esmt$  & $\bsmt 15 & 0\\ 0 & 15\esmt$
& $ \bsmt 0 & 7 \\ 9 & 7 \esmt$ & $\bsmt 9 & 1 \\ 3 & 2\esmt$  & $ \bsmt 6 & 7 \\ 9 & 13\esmt$& $\bsmt 15 & 0\\ 0 & 15\esmt  $
\\
\cline{2-9}
\end{tabular}

\vspace{0.1 in}

\begin{tabular}{l|c|c|c|c|c|}
\cline{2-4}\cline{6-6}
5. & \multicolumn{2}{c|}{$|\tstb{\Pi}|$, $|\bstb{\Pi}|$} &\myStrut{4} $|\cbstb{\Pi}|$
&\spn{0.1}{6} & $G_0^6 = I \quad G_1^4  = G_0  \quad G_2^4  = G_0^5 \quad G_3^2  = I$ 
\\
\cline{2-4}\cline{6-6}
&\multicolumn{2}{c|}{$6$} & \myStrut{4} $192$
\\
\cline{2-4}
\end{tabular}

\vspace{0.2 in}

 \textbf{Orbit $8b$} 

\nopagebreak

\vspace{0.15 in}

\nopagebreak

\begin{tabular}{l|c|c|c|c|c|c|c|c|c|c|c|c|c|c|c|c|}
\cline{2-7} \cline{9-11}
 1.&\myStrut{4} & $a$ & $r_1$ &$t$ & $b_1$ & $i$ 
& \spn{0.1}{2} & 
\multicolumn{3}{c|}{$t'=\cos\frac{3\pi}{8} \qquad b'_1= \sqrt{a+1} $}
\\
\cline{2-7} \cline{9-11}
&\myStrut{4} $g_a$ & $-a$ & $r_1$ & $t$ & $b'_1$ &  $i$  
\\
\cline{2-7} \cline{9-17}
&\myStrut{4} $g_1$ & $a$ & $-r_1$ & $t'$ & $b_1$ &  $i$  
& \spn{0.1}{3} & 
\multicolumn{9}{c|}{$g_1^4 = g_2^2   = e \quad g\vpu{-1}_a g\vpu{-1}_1 g^{-1}_a  =g\vpu{-1}_1
\quad
g\vpu{-1}_a g\vpu{-1}_2 g^{-1}_a  =g\vpu{-1}_2$}
\\
\cline{2-7}\cline{9-17}
&\myStrut{4} $g_2$ & $a$ & $r_1$ & $t$ & $-b_1$ &  $i$
\\
\cline{2-7}
&\myStrut{4} $\bar{g}_1$ &  $a$ & $r_1$ & $t$ & $b_1$ &  $-i$
\\
\cline{2-7}
&\multicolumn{6}{|c |}{\myStrut{4} $g\vpu{4}_c =g_2\bar{g}_1$}
\\
\cline{2-7}
\end{tabular}

\vspace{0.1 in}

\begin{tabular}{r|c|c|c|c|c|c|c|c|c|c|}
\cline{2-7}\cline{9-11}
4. & \myStrut{4} $G_0$ & $G_1$ & $G_2$ & $F_0$ & $F_1$  & $F_2$ 
& \spn{0.1}{5} &  \multicolumn{2}{c|}{$|\tstb{\Pi}|$, $|\bstb{\Pi}|$} & $|\cbstb{\Pi}|$
\\
\cline{2-7}\cline{9-11}
& \myStrut{5} $\bsmt 10 & 5 \\ 11& 15\esmt$ & $\bsmt 2 & 3\\ 13 & 5\esmt$ & $\bsmt 6 & 11 \\ 5 & 1 \esmt$ & $\bsmt 9 & 1 \\ 3 & 2\esmt $ 
& $\bsmt 15 & 0 \\ 0 & 15\esmt$
&$\bsmt 15 & 0 \\ 0 & 15\esmt $
&&  \multicolumn{2}{c|}{$24$}& $192$
\\
\cline{2-7}\cline{9-11}
\end{tabular}

\vspace{0.1 in}

\begin{tabular}{r|c|}
\cline{2-2}
6. & \myStrut{5} $G_0^{24} = I \quad G_1^4  = G_0^{4} \quad G_2^2  = G_0 $ 
\\
\cline{2-2}
\end{tabular}

\vspace{0.2 in}

\textbf{Orbits $9a,b$} 

\nopagebreak

\vspace{0.15 in}

\nopagebreak

\begin{tabular}{r|c|c|c|c|c|c|c|c|c|c|c|c|c|c|}
\cline{2-8}\cline{10-13}
1.& \myStrut{4} & $a$ & $r_1$ & $t$ & $b_1$ & $b_2$ & $i$
&\spn{0.1}{2} &
\multicolumn{4}{c|}{$t'=\cos\frac{5\pi}{9} \quad b'_1 = \sqrt{2a-4r_1} $}
\\
\cline{2-8}
&\myStrut{4} $g_a$ & $-a$ & $r_1$ & $t$ & $b'_1$ &  $b_2$  & $i$
&& 
\multicolumn{4}{c|}{$ b'_2 = 2\re\left( e^{\frac{2\pi i}{3}} \left( 6+2iar_1\right)^{\frac{1}{3}}\right) $}
\\
\cline{2-8}\cline{10-13}
&\myStrut{4} $g_s$ &  $a$ & $-r_1$ & $t'$ & $ib'_1$ &  $b_2$  & $i$ 
\\
\cline{2-8} \cline{10-15}
&\myStrut{4} $g_2$ & $a$ & $r_1$ & $t$ & $b_1$ &  $b'_2$  & $i$
&\spn{0.1}{3} &
\multicolumn{6}{c|}{$g_s^{12} =  g_2^3   = e \quad g\vpu{-1}_a g\vpu{-1}_s g^{-1}_a  =g_s
\quad
g\vpu{-1}_a g\vpu{-1}_2 g^{-1}_a  =g^2_2 $}
\\
\cline{2-8}\cline{10-15}
&\myStrut{4} $\bar{g}_1$ & $a$ & $r_1$ & $t$ & $b_1$ &  $b_2$  & $-i$
\\
\cline{2-8}
&\multicolumn{7}{|c |}{\myStrut{4} $g\vpu{2}_1 = g_s^2 \qquad g\vpu{4}_c = g_s^6\bar{g}\vpu{3}_1$} 
\\
\cline{2-8} 
\end{tabular}

\vspace{0.1 in}

\begin{tabular}{r|c|c|c|c|c|c|c|}
\cline{2-7}
4. & \myStrut{4} $G_0$ & $G_{as}$ & $G_{bs}$& $G_{a1}=G_{b1}$ & $G_{a2}$  & $G_{b2}$ 
\\
\cline{2-7}
& \myStrut{5} $\bsmt 0 & 8 \\ 1 & 8\esmt $ & $\bsmt 4 & 8\\ 3 & 5\esmt$ & $\bsmt 5 & 1 \\ 6 & 4\esmt$&
$\bsmt 5 & 0\\ 0 & 5\esmt $ & $\bsmt 2 & 4\\ 5 & 6\esmt$   & $\bsmt 6 & 5 \\ 4 & 2\esmt $
\\
\cline{2-7}
 & \myStrut{4} $F_0$ & $F_{as}$ & $F_{bs}$ & $F_{a1}=F_{b1}$ & $F_{a2}$  & $F_{b2}$ 
\\
\cline{2-7}
& \myStrut{5} $\bsmt 0 & 8 \\ 1 & 8 \esmt$ & $\bsmt 1 & 2 \\ 6 & 4\esmt$& $\bsmt 8 & 7 \\ 3 & 5\esmt$ & $\bsmt 2& 0 \\ 0 & 5\esmt $ & $\bsmt 6 & 5 \\ 4 & 2\esmt $ & $\bsmt 2 & 4 \\ 5 & 6\esmt$
\\
\cline{2-7}
\end{tabular}

\vspace{0.1 in}

\begin{tabular}{r|c|c|c|c|c|}
\cline{2-4}\cline{6-6}
5. & \multicolumn{2}{c|}{$|\tstb{\Pi_a}|,|\tstb{\Pi_b}|,|\bstb{\Pi_a}|,|\bstb{\Pi_b}|$} &\myStrut{4} $|\cbstb{\Pi_a}|,|\cbstb{\Pi_b}|$ 
& \spn{0.1}{6} &  \multicolumn{1}{c|}{$G_0^{3} = I $ }
\\
\cline{2-4}
&\multicolumn{2}{c|}{ $3$} & \myStrut{4} $54$
&& \multicolumn{1}{c|}{$G_{a1}^6  = G_{b1}^6  =I \quad G_{a2}^3 = G_{b2}^3 = I$ }
\\
\cline{2-4}\cline{6-6}
\end{tabular}

\vspace{0.1 in}

\begin{tabular}{r|c|c|c|c|c|c|c|c|}
\cline{2-6}
7. & \myStrut{4}  $\rv_{as}$ & $\rv_{bs}$ &  $\rv_{a1}=\rv_{b1}$ & $\rv_{a2}=\rv_{b2}$& $\bar{\rv}_{a1}=\bar{\rv}_{b1}$
\\
\cline{2-6}
& \myStrut{5}  $\bsmt 1 \\2\esmt $ & $\bsmt 0 \\0\esmt $ & $\bsmt 2 \\1\esmt $ & $\bsmt 0 \\0\esmt $ & $\bsmt 1 \\2\esmt $
\\
\cline{2-6}
 & \myStrut{4}  $\qv_{as}$  & $\qv_{bs}$ &  $\qv_{a1}=\qv_{b1}$& $\qv_{a2}=\qv_{b2}$ & $\bar{\qv}_{a1}=\bar{\qv}_{b1}$
\\
\cline{2-6}
& \myStrut{5} $\bsmt 3 \\3\esmt $ & $\bsmt 0 \\0 \esmt$ &
$\bsmt 3 \\ 6\esmt $ & $\bsmt 0 \\ 0\esmt$ & $\bsmt 3 \\ 3 \esmt$ 
\\
\cline{2-6}
\end{tabular}

\vspace{0.2 in }

\textbf{Orbit $10a$} 

\nopagebreak

\vspace{0.15 in}

\nopagebreak

\begin{tabular}{l|c|c|c|c|c|c|c|c|c|}
\cline{2-10}
 1. &\myStrut{4} & $a$ & $r_1$ & $r_2$ & $r_3$ & $t$ & $b_1$ & $b_2$ & $i$  
\\
\cline{2-10}
&\myStrut{4} $g_a$ & $-a$ & $r_1$ & $r_2$ & $r_3$ & $t$ & $b'_1$ &  $b_2$  & $i$ 
\\
\cline{2-10}
&\myStrut{4} $g_1$ &  $a$ & $r_1$ & $-r_2$ & $r_3$ & $t'$ & $b''_1$ &  $b'_2$  & $i$ 
\\
\cline{2-10}
&\myStrut{4} $g_2$ & $a$ & $r_1$ & $r_2$ & $-r_3$ & $-t$ & $b_1$ &  $b_2$  & $i$ 
\\
\cline{2-10}
&\myStrut{4} $\bar{g}_1$ & $a$ & $r_1$ & $r_2$ & $r_3$ & $t$ & $b_1$ &  $b_2$  & $-i$ 
\\
\cline{2-10}
&\myStrut{4} $\bar{g}_2$ & $a$ & $-r_1$ & $r_2$ & $r_3$ & $t$ & $b_1$ &  $b_2$  & $i$ 
\\
\cline{2-10} 
& \multicolumn{9}{|c |}{\myStrut{5} $ g\vpu{4}_c = g^{12}_1\bar{g}\vpu{12}_1$} 
\\
\cline{2-10}
\end{tabular}

\vspace{0.1 in}

\begin{tabular}{l|l|}
\cline{2-2}
 2. & \myStrut{5} $t'=\cos\frac{3\pi}{10} $
 \\
& \myStrut{5} $b'_1 = \sqrt{ 14(-5 + 5a + 5r_2 - ar_2) -8(-21 + 2a + 14r_2 -ar_2)r_3t}$
\\
& \myStrut{5} $b''_1  =i \sqrt{14(5+5a+5r_2+ar_2)-8(21+2a+14r_2+ar_2)r_3t'}$
\\
&\myStrut{5} $b'_2  = \frac{2}{3}\re\left( e^{\frac{2\pi i}{3}} (405+15i a\sqrt{3})^{\frac{1}{3}} \right)$
\\
\cline{2-2}
\end{tabular}

\vspace{0.1 in}

\begin{tabular}{l|c|}
\cline{2-2}
3. & \myStrut{4} $g_1^{24} =  g_2^2  =  e \quad g\vpu{-1}_a g\vpu{-1}_1 g^{-1}_a  =g_1^5
\quad
g\vpu{-1}_a g\vpu{-1}_2 g^{-1}_a  =g_1^{12} g\vpu{-1}_2$
\\
\cline{2-2}
\end{tabular}

\vspace{0.1 in}

\begin{tabular}{l|c|c|c|c|c|c|c|c|c|c|}
\cline{2-7} \cline{9-11}
4. & \myStrut{4} $G_0$ & $G_{1}$ & $G_{2}$ & $F_0$ & $F_{1}$ & $F_{2}$ 
& \spn{0.1}{5} &  \multicolumn{2}{c|}{$|\tstb{\Pi}|$, $|\bstb{\Pi}|$} & $|\cbstb{\Pi}|$
\\
\cline{2-7} \cline{9-11}
& \myStrut{5}$\bsmt 0 & 9 \\11 & 9 \esmt$ & $\bsmt 12 & 17 \\ 3 & 9 \esmt$ & $\bsmt 11 & 10 \\ 10 & 1\esmt$ 
& $\bsmt 0 & 9 \\11 & 9 \esmt$ & $\bsmt 17 & 19 \\ 17 & 12 \esmt$ & $\bsmt 9 & 10 \\ 10 & 11\esmt$
&&\multicolumn{2}{c|}{$6$} & $288$
\\
\cline{2-7} \cline{9-11}
\end{tabular}

\vspace{0.1 in}

\begin{tabular}{l|c|}
\cline{2-2}
6. & \myStrut{5} $G_0^{6} = I \qquad G_1^{24}  = I \qquad G_2^2  = I$ 
\\
\cline{2-2}
\end{tabular}

\vspace{0.2 in }

 \textbf{Orbits $11a,b$} 

\nopagebreak

\vspace{0.15 in}

\nopagebreak

\begin{tabular}{l|c |c|c|c|c|c|c|c|}
\cline{2-9}
1. &\myStrut{4} & $a$ & $r_1$ & $r_2$ &  $t$ & $b_1$ & $b_2$ & $i$
\\
\cline{2-9}
&\myStrut{4} $g_a$ &   $-a$ &  $r_1$ &  $r_2$ &  $t$  &  $b_1$  &  $b'_2$ &  $i$
\\
\cline{2-9}
&\myStrut{4} $g_s$ &   $a$ &  $-r_1$ &  $r_2$ &  $t$  &  $b_1$  &  $b_2$ &  $i$
\\
\cline{2-9}
&\myStrut{4} $g_1$ &    $a$ &  $r_1$ &  $-r_2$ &  $t'$  &  $b'_1$  &  $b''_2$ &  $i$
\\
\cline{2-9}
&\myStrut{4} $\bar{g}_1$ &  $a$ &  $r_1$ &  $r_2$ &  $t$  &  $b_1$  &  $b_2$ &  $-i$
\\
\cline{2-9}
& \multicolumn{8}{|c |}{\myStrut{5} $ g\vpu{4}_c = g^{20}_1\bar{g}\vpu{12}_1$} 
\\
\cline{2-9}
\end{tabular}

\vspace{0.1 in}

\begin{tabular}{r|c|c|c|}
\cline{2-2} \cline{4-4}
 2. & \myStrut{5} $t' = \cos\frac{3\pi}{11}$
& \spn{0.1}{3} & 
\multicolumn{1}{c|}{$g_s^2= g_1^{40}    =  e$}
\\
&\myStrut{5} $b'_1  =\sqrt{11-2r_2}$
&& 
\multicolumn{1}{c|}{$g\vpu{-1}_a g\vpu{-1}_s g^{-1}_a  = g\vpu{-1}_s g_1^{20}$}
\\
&\myStrut{5} $b'_2  = \sqrt{ 22a + 6r_2 + 3ar_2 + 22-(4 + 6a + 2r_2)b_1}$
&& 
\multicolumn{1}{c|}{$g\vpu{-1}_a g\vpu{-1}_1 g^{-1}_a  =g_1^{11}$}
\\
\cline{4-4}
&\myStrut{5} $b''_2  = i\sqrt{22 a + 6 r_2 - 3 a r_2 - 22 + (4 - 6 a - 2 r_2) b'_1}$
\\
\cline{2-2}
\end{tabular}

\vspace{0.1 in}

\begin{tabular}{l|c|c|c|c|c|c|c|c|}
\cline{2-5} \cline{7-9}
4. & \myStrut{4} $G_0$ & $G_{as}$ & $G_{bs}$& $G_{a1}=G_{b1}$ 
&\spn{0.1}{5}&
\multicolumn{2}{c|}{$|\tstb{\Pi_a}|,|\tstb{\Pi_b}|, |\bstb{\Pi_a}|,|\bstb{\Pi_b}|$} & $|\cbstb{\Pi_a}|,|\cbstb{\Pi_b}|$
\\
\cline{2-5} \cline{7-9}
& \myStrut{5}$\bsmt 0 & 10 \\1 & 10 \esmt$ & $\bsmt 7 & 2 \\ 9 & 9 \esmt$ & $\bsmt 9 & 9 \\ 2 & 7 \esmt$& $\bsmt 5 & 9 \\ 2 & 3 \esmt$ 
&&
\multicolumn{2}{c|}{$3$} & $120$
\\
\cline{2-5} \cline{7-9}
& \myStrut{4} $F_0$ &  $F_{as}$ & $F_{bs}$ &$F_{a1}=F_{b1}$ 
\\
\cline{2-5} \cline{7-9}
&\myStrut{5}$\bsmt 0 & 10 \\1 & 10 \esmt$ &  
$\bsmt 2 & 2 \\ 2 & 7\esmt$ & 
$\bsmt 4 & 9 \\ 9 & 9\esmt$  &
$\bsmt 10 & 3 \\ 9 & 5 \esmt$ 
&\spn{0.1}{6} & \multicolumn{3}{c|}{$G_0^{3} = I \quad G_{a1}^{40}  = G_{b1}^{40}  = G_0^2$ }
\\
\cline{2-5}\cline{7-9}
\end{tabular}

\vspace{0.2 in}

\textbf{Orbit $11c$} 

\nopagebreak

\vspace{0.15 in}

\nopagebreak

\begin{tabular}{l|c |c|c|c|c|c|c|c|}
\cline{2-9}
1. & \myStrut{4} & $a$ & $r_1$ & $r_2$ &  $t$ & $b_1$ & $b_2$ & $i$
\\
\cline{2-9}
&\myStrut{4} $g_a$ &  $-a$ & $r_1$ & $r_2$ & $t$  & $b_1$  & $b'_2$ & $i$
\\
\cline{2-9}
&\myStrut{4} $g_1$ &   $a$ & $r_1$ & $-r_2$ & $t'$  & $b'_1$  & $b''_2$ & $i$
\\
\cline{2-9}
&\myStrut{4} $\bar{g}_1$ & $a$ & $r_1$ & $r_2$ & $t$  & $b_1$  & $b_2$ & $-i$
\\
\cline{2-9}
&\myStrut{4} $\bar{g}_2$ &  $a$ & $-r_1$ & $r_2$ & $t$  & $b_1$  & $b_2$ & $-i$
\\
\cline{2-9}
& \multicolumn{8}{|c |}{\myStrut{5} $ g\vpu{4}_c = g^{20}_1\bar{g}\vpu{12}_1$} 
\\
\cline{2-9}
\end{tabular}

\vspace{0.1 in}

\begin{tabular}{r|c|c|c|}
\cline{2-2} \cline{4-4}
 2. & \myStrut{5} $t' = \cos\frac{3\pi}{11}$
&\spn{0.1}{3} & \multicolumn{1}{c|}{$ g_1^{40}    =  e \quad g\vpu{-1}_a g\vpu{-1}_1 g^{-1}_a  =  g_1^{11}$}
\\
\cline{4-4}
&\myStrut{5} $b'_1  =\sqrt{11-2r_2}$
\\
&\myStrut{5} $b'_2  = \sqrt{ 22a + 6r_2 + 3ar_2 + 22-(4 + 6a + 2r_2)b_1}$
\\
&\myStrut{5} $b''_2  = i\sqrt{22 a + 6 r_2 - 3 a r_2 - 22 + (4 - 6 a - 2 r_2) b'_1}$
\\
\cline{2-2}
\end{tabular}

\vspace{0.1 in}

\begin{tabular}{l|c|c|c|c|c|c|c|c|}
\cline{2-5} \cline{7-9}
4. & \myStrut{4} $G_0$ & $G_{1}$ & $F_0$& $F_1$ 
&\spn{2.1}{5}&
\multicolumn{2}{c|}{$|\tstb{\Pi}|$, $|\bstb{\Pi}|$} & $|\cbstb{\Pi}|$
\\
\cline{2-5} \cline{7-9}
& \myStrut{5} $\bsmt 0 & 10 \\1 & 10 \esmt$ 
& $\bsmt 5 & 9 \\ 2 & 3 \esmt$ & $\bsmt 0 & 10 \\1 & 10 \esmt$ &  
$\bsmt 10 & 3 \\ 9 & 5\esmt$
&&
\multicolumn{2}{c|}{ $3$} & $120$
\\
\cline{2-5} \cline{7-9}
\end{tabular}

\vspace{0.1 in}

\begin{tabular}{l|c|}
\cline{2-2}
6. & \myStrut{5} $G_0^{3} = I \quad G_{1}^{40}  = G_0^2 $ 
\\
\cline{2-2}
\end{tabular}

\vspace{0.2 in}

 \textbf{Orbit $12a$} 

\nopagebreak

\vspace{0.15 in}

\nopagebreak

\begin{tabular}{l|c |c|c|c|c|c|c|c|c|c|}
\cline{2-9} \cline{11-11}
1. &\myStrut{4} & $a$ & $r_1$ & $r_2$ &  $b_1$ & $b_2$ & $b_3$ & $i$
&\spn{0.21}{2}& 
$b'_1  =\sqrt{a+1} \quad b'_2  = \sqrt{a+3}$
\\
\cline{2-9}
&\myStrut{4} $g_a$ &  $-a$ & $r_1$ & $-r_2$ & $b'_1$  & $b'_2$  & $b_3$ & $i$
&&
$b'_3 = 2\re\Bigl( e^{\frac{2\pi i}{3}} (5+ir_2)^{\frac{1}{3}} \Bigr)$
\\
\cline{2-9} \cline{11-11}
&\myStrut{4} $g_1$ &   $a$ & $-r_1$ & $r_2$ & $b_1$  & $b_2$  & $b'_3$ & $i$
\\
\cline{2-9} \cline{11-11}
&\myStrut{4} $g_2$ &   $a$ & $r_1$ & $-r_2$ & $b_1$  & $b_2$  & $b_3$ & $i$
&\spn{0.1}{3}&
$g_1^{6} = g_2^2=g_3^2 = g_4^2    =  e$
\\
\cline{2-9}
&\myStrut{4} $g_3$ &   $a$ & $r_1$ & $r_2$ & $-b_1$  & $b_2$  & $b_3$ & $i$
&&
$g\vpu{-1}_a g\vpu{-1}_1 g^{-1}_a  = g_1^{5} \quad  $
\\
\cline{2-9}

&\myStrut{4} $g_4$ &   $a$ & $r_1$ & $r_2$ & $b_1$  & $-b_2$  & $b_3$ & $i$ 
&&
$g\vpu{-1}_a g\vpu{-1}_2 g^{-1}_a  =  g\vpu{-1}_2 g\vpu{-1}_3$
\\
\cline{2-9}
&\myStrut{4} $\bar{g}_1$ & $a$ & $r_1$ & $r_2$ & $b_1$  & $b_2$  & $b_3$ & $-i$ 
&&
$g\vpu{-1}_a g\vpu{-1}_3 g^{-1}_a  =  g\vpu{-1}_3$ 
\\
\cline{2-9}
& \multicolumn{8}{|c |}{\myStrut{5} $ g\vpu{4}_c = g_3g_4\bar{g}_1$} 
&&
$g\vpu{-1}_a g\vpu{-1}_4 g^{-1}_a  =  g\vpu{-1}_4$
\\
\cline{2-9} \cline{11-11}
\end{tabular}

\vspace{0.1 in}

\begin{tabular}{l|c|c|c|c|c|c|c|c|c|c|c|}
\cline{2-6}\cline{8-10}
4. & \myStrut{4} $G_0$ & $G_{1}$ & $G_2$& $G_3$ & $G_4$ 
&\spn{0.1}{5}&
\multicolumn{2}{c|}{$|\tstb{\Pi}|$, $|\bstb{\Pi}|$} & $|\cbstb{\Pi}|$
\\
\cline{2-6}\cline{8-10}
& \myStrut{5} $\bsmt 0 & 11 \\ 13 & 11 \esmt$ 
& $\bsmt 20 & 9 \\ 15 & 5 \esmt$ 
& $\bsmt 23 & 6 \\ 18 & 5 \esmt$
& $\bsmt 17 &0 \\ 0 & 17\esmt$
& $\bsmt 7 & 0 \\ 0 & 7\esmt$
&& 
\multicolumn{2}{c|}{$6$} & $288$
\\
\cline{2-6}\cline{8-10}
& \myStrut{4} $F_0$ & $F_{1}$ & $F_2$& $F_3$ & $F_4$ 
\\
\cline{2-6}\cline{8-12}
& \myStrut{5}
$\bsmt 0 & 11 \\13 & 11 \esmt$ &  
$\bsmt 17 & 3 \\ 9 & 20\esmt$ & 
$\bsmt 13 & 18 \\ 6 & 23 \esmt$  &
$\bsmt 17 & 0 \\ 0 & 17 \esmt$  &
$\bsmt 7 & 0 \\ 0 & 7 \esmt$
&\spn{0.1}{6}&
\multicolumn{5}{c|}{$G_0^{6} = I \quad G_{1}^{6}  = I$}
\\
\cline{2-6}
\multicolumn{7}{c}{} &
\multicolumn{5}{|c|}{\myStrut{4} $G_{2}^2 = G_{0}^3 \quad G^2_3 = I \quad G^2_4 = I$}
\\
\cline{8-12}
\end{tabular}

\begin{tabular}{l|c|}
\cline{2-2}
7. & \myStrut{4} $\rv_{1} = \rv_{2} = \rv_{3} = \rv_{4} = \bar{\rv}_{1} = \boldsymbol{0}$ 
\\
& \myStrut{4} $ \qv_{1} = \qv_{2} =\qv_{3} =\qv_{4} =\bar{\qv}_{1} =\boldsymbol{0}$
\\
\cline{2-2} 
\end{tabular}

\vspace{0.2 in}

 \textbf{Orbit $12b$} 

\nopagebreak

\vspace{0.15 in}

\nopagebreak

\begin{tabular}{l|c |c|c|c|c|c|c|c|c|}
\cline{2-8} \cline{10-10}
1. &\myStrut{4} & $a$ & $r_1$ & $r_2$ &  $b_1$ & $b_2$ &  $i$
& \spn{0.1}{2} & 
$b'_1  =\sqrt{a+1} \quad b'_2  = \sqrt{a+3}$
\\
\cline{2-8} \cline{10-10}
&\myStrut{4} $g_a$ &  $-a$ & $r_1$ & $-r_2$ & $b'_1$  & $b'_2$  &  $i$
\\
\cline{2-8} \cline{10-10}
&\myStrut{4} $g_1$ &   $a$ & $-r_1$ & $r_2$ & $b_1$  & $b_2$  &  $i$
& \spn{0.1}{3} &
$ g_1^2 = g_2^2=g_3^2 = g_4^2    =  e$
\\
\cline{2-8}
&\myStrut{4} $g_2$ &   $a$ & $r_1$ & $-r_2$ & $b_1$  & $b_2$  &  $i$
&&
$g\vpu{-1}_a g\vpu{-1}_1 g^{-1}_a  =  g\vpu{-1}_1$
\\
\cline{2-8}
&\myStrut{4} $g_3$ &   $a$ & $r_1$ & $r_2$ & $-b_1$  & $b_2$  &  $i$
&&
$g\vpu{-1}_a g\vpu{-1}_2 g^{-1}_a  =  g\vpu{-1}_2 g\vpu{-1}_3$
\\
\cline{2-8}
&\myStrut{4} $g_4$ &   $a$ & $r_1$ & $r_2$ & $b_1$  & $-b_2$  &  $i$
&&
$g\vpu{-1}_a g\vpu{-1}_3 g^{-1}_a  =  g\vpu{-1}_3 $
\\
\cline{2-8}
&\myStrut{4} $\bar{g}_1$ & $a$ & $r_1$ & $r_2$ & $b_1$  & $b_2$  &  $-i$
&&
$g\vpu{-1}_a g\vpu{-1}_4 g^{-1}_a  =  g\vpu{-1}_4$
\\
\cline{2-8} \cline{10-10}
& \multicolumn{7}{|c |}{\myStrut{5} $ g\vpu{4}_c = g_3g_4\bar{g}_1$} 
\\
\cline{2-8} 
\end{tabular}

\vspace{0.15 in}

\begin{tabular}{l|c|c|c|c|c|c|c|c|c|c|c|c|}
\cline{2-6} \cline{8-10}
4. & \myStrut{4} $G_0$ & $G_{1}$ & $G_2$& $G_3$ & $G_4$ 
&\spn{0.1}{5} &
 \multicolumn{2}{c|}{$|\tstb{\Pi}|$, $|\bstb{\Pi}|$} & $|\cbstb{\Pi}|$
\\
\cline{2-6} \cline{8-10}
& \myStrut{5} $\bsmt 0 & 7 \\ 7 & 9 \esmt$ 
& $\bsmt 2 & 3 \\ 3 & 23 \esmt$ 
& $\bsmt 23 & 15 \\ 15 & 8 \esmt$
& $\bsmt 5 &0 \\ 0 & 5\esmt$
& $\bsmt 7 & 0 \\ 0 & 7\esmt$
&&
 \multicolumn{2}{c|}{$12$} & $192$
\\
\cline{2-6} \cline{8-10}
& \myStrut{4} $F_0$ & $F_{1}$ & $F_2$& $F_3$ & $F_4$ 
\\
\cline{2-6} \cline{8-13}
& \myStrut{5}
$\bsmt 0 & 17 \\17 & 15 \esmt$ &  
$\bsmt 11 & 9 \\ 21 & 2\esmt$ & 
$\bsmt 16 & 9 \\ 9 & 23 \esmt$  &
$\bsmt 5 & 0 \\ 0 & 5 \esmt$ &
$\bsmt 7 & 0 \\ 0 & 7 \esmt$
&\spn{0.1}{6}&
\multicolumn{6}{c|}{$G_0^{12} = I \quad G_{1}^2  = G_0^8 $}
\\
\cline{2-6}
\multicolumn{7}{c}{} & \multicolumn{6}{|c|}{\myStrut{4} $ G_{2}^2 = G_{0}^{10} \quad G^2_3 = I \quad G^2_4 = I$}
\\
\cline{8-13}
\end{tabular}

\vspace{0.1 in}

\begin{tabular}{l|c|c|c|c|c|c|c|c|c|c|}
\cline{2-11}
7. & \myStrut{4}  $\rv_{1}$ &  $\rv_{2}$ & $\rv_{3}$ & $\rv_{4}$ & $\bar{\rv}_{1}$ &$\qv_{1}$ &  $\qv_{2}$ & $\qv_{3}$ & $\qv_{4}$ & $\bar{\qv}_{1}$
\\
\cline{2-11}
& \myStrut{5} $\bsmt 2 \\1\esmt $ & $\bsmt 0 \\0 \esmt$ &
$\bsmt 2 \\ 1 \esmt $ & $\bsmt 0 \\ 0\esmt$ & $\bsmt 1 \\ 2 \esmt$ & $\bsmt 4 \\ 8\esmt$ & $\bsmt 0 \\ 0\esmt $ & $\bsmt 4\\ 8\esmt $& $\bsmt 0 \\ 0\esmt $ & $\bsmt 4\\ 4\esmt $
\\
\cline{2-11}
\end{tabular}

\vspace{0.2 in}

 \textbf{Orbits $13a,b$} 

\nopagebreak

\vspace{0.15 in}

\nopagebreak

\begin{tabular}{r|c |c|c|c|c|c|c|c|}
\cline{2-9}
1. & \myStrut{4} & $a$ & $r_1$ & $r_2$ & $t$ &  $b_1$ & $b_2$ &  $i$
\\
\cline{2-9}
&\myStrut{4} $g_a$ & $-a$ & $r_1$ & $r_2$ & $$t$$ &$b'_1$  & $b'_2$  &  $i$
\\
\cline{2-9}
&\myStrut{4} $g_s$ &  $a$ & $-r_1$ & $-r_2$ & $$t'$$ &$b''_1$  & $b''_2$  &  $i$ 
\\
\cline{2-9}
&\myStrut{4} $g_2$ &  $a$ & $r_1$ & $-r_2$ & $$t''$$ &$b_1$  & $b_2$  &  $i$ 
\\
\cline{2-9}
&\myStrut{4} $\bar{g}_1$ & $a$ & $r_1$ & $r_2$ & $$t$$ &$b_1$  & $b_2$  &  $-i$ 
\\
\cline{2-9}
& \multicolumn{8}{|c |}{\myStrut{5} $ g\vpu{2}_1 = g_s^2 \qquad g\vpu{4}_c = g^{12}_s\bar{g}\vpu{12}_1$} 
\\
\cline{2-9}
\end{tabular}

\vspace{0.15 in}

\begin{tabular}{r|c|c|c|}
\cline{2-2} \cline{4-4}
 2. & \myStrut{5} 
$t'=- \sin\frac{2\pi}{13}$
&\spn{0.1}{3}&
$g_s^{24} = g_2^4     =  e$
\\
& \myStrut{5} $t''  = \sin \frac{5\pi}{13}$
&&
$g\vpu{-1}_a g\vpu{-1}_s g^{-1}_a  =  g\vpu{-1}_s$
\\
& \myStrut{5} $b'_1  =\sqrt{250+40a+50r_1+8ar_1}$
&&
$g\vpu{-1}_a g\vpu{-1}_2 g^{-1}_a  =  g^{18}_s g^3_2$
\\
\cline{4-4}
& \myStrut{5}  $b''_1 = \sqrt{250-40a-50r_1+8a r_1}$
\\
& \myStrut{5}  $b'_2  = \sqrt{-205+40a-95r_1+16ar_1-(30-5a+12r_1-2a r_1)b'_1}$
\\
& \myStrut{5}  $b''_2  = i\sqrt{205 + 40 a - 95 r_1 - 16 a r_1 + (30 + 5 a - 12 r_1 - 2 a r_1) b''_1}$
\\
\cline{2-2}
\end{tabular}

\vspace{0.15 in}

\begin{tabular}{l|c|c|c|c|c|c|c|}
\cline{2-8}
4. & \myStrut{4} $G_0$ & $G_{as}$ & $G_{bs}$ & $G_{a1}$ & $G_{b1}$ & $G_{a2}$ & $G_{b2}$
\\
\cline{2-8}
& \myStrut{5} $\bsmt 0 & 12 \\ 1 & 12 \esmt$ 
& $\bsmt 9 & 11 \\ 7 & 4 \esmt$ 
& $\bsmt 6 & 9 \\ 2 & 7 \esmt$
& $\bsmt 11 & 2 \\ 11 & 0\esmt$
& $\bsmt 0 & 11 \\ 2 & 11\esmt$
& $\bsmt 6 & 12 \\ 1 & 5\esmt$
& $\bsmt 5 & 1 \\ 12 & 6\esmt$
\\
\cline{2-8}
& \myStrut{4} $F_0$ &  $F_{as}$ & $F_{bs}$ & $F_{a1}$ & $F_{b1}$  & $F_{a2}$ &  $F_{b2}$
\\
\cline{2-8}
& \myStrut{5}
$\bsmt 0 & 12 \\ 1 & 12 \esmt$ &  
$\bsmt 2 & 1 \\ 6 & 9\esmt$ & 
$\bsmt 10 & 2 \\ 11 & 6\esmt$  &
$\bsmt 0 & 6 \\ 2 & 11 \esmt$ &
$\bsmt 6 & 7 \\ 11 & 0 \esmt$ &
$\bsmt 1 & 8 \\ 12 & 6 \esmt$ &  
$\bsmt 9 & 5 \\ 1 & 5\esmt$ 
\\
\cline{2-8}
\end{tabular}

\vspace{0.15 in}

\begin{tabular}{l|c|c|c|c|c|}
\cline{2-4}\cline{6-6}
5. & \multicolumn{2}{c|}{$|\tstb{\Pi_a}|,|\tstb{\Pi_b}|,|\bstb{\Pi_a}|,|\bstb{\Pi_b}|$} & \myStrut{4} $|\cbstb{\Pi_a}|,|\cbstb{\Pi_b}|$
&\spn{1.9}{6}&
$G_0^{3} = I \quad G_{a1}^{12}  = G_{b1}^{12}  = I $
\\
\cline{2-4}
& \multicolumn{2}{c|}{$3$} & \myStrut{4} $144$
&&
$ G_{a2}^4 = G_{0}^{2}\quad G_{b2}^4 = G\vpu{2}_{0}$
\\
\cline{2-4}\cline{6-6}
\end{tabular}

\vspace{0.2 in }

 \textbf{Orbits $14a,b$} 

\nopagebreak

\vspace{0.15 in}

\nopagebreak

\begin{tabular}{r|c |c|c|c|c|c|c|c|c|c|}
\cline{2-11}
1. & \myStrut{4} & $a$ & $r_1$ & $r_2$ & $r_3$ & $r_4$ & $t$ &  $b_1$ & $b_2$ &  $i$
\\
\cline{2-11}
&\myStrut{4} $g_a$ & $-a$ & $r_1$ & $r_2$ &  $r_3$ & $r_4$ &$$t$$ &$b'_1$  & $b'_2$  &  $i$
\\
\cline{2-11}
&\myStrut{4} $g_s$ &   $a$ & $-r_1$ & $r_2$ &  $r_3$ & $r_4$ &$$t$$ &$b''_1$  & $b''_2$  &  $i$
\\
\cline{2-11}
&\myStrut{4} $g_2$ &  $a$ & $r_1$ & $r_2$ &  $r_3$ & $-r_4$ &$$t'$$ &$b_1$  & $b_2$  &  $i$ 
\\
\cline{2-11}
&\myStrut{4} $g_3$ &  $a$ & $r_1$ & $r_2$ &  $-r_3$ & $r_4$ &$$t$$ &$b_1$  & $b_2$  &  $i$ 
\\
\cline{2-11}
&\myStrut{4} $\bar{g}_1$ &$a$ & $r_1$ & $r_2$ &  $r_3$ & $r_4$ &$$t$$ &$b_1$  & $b_2$  &  $-i$
\\
\cline{2-11}
&\myStrut{4} $\bar{g}_2$ &$a$ & $r_1$ & $-r_2$ &  $r_3$ & $r_4$ &$$t$$ &$b_1$  & $b_2$  &  $i$
\\
\cline{2-11}
& \multicolumn{10}{|c |}{\myStrut{5} $ g\vpu{2}_1 = g_s^2 \qquad g\vpu{4}_c = g^{6}_s\bar{g}\vpu{12}_1$} 
\\
\cline{2-11}
\end{tabular}

\vspace{0.15 in}

\begin{tabular}{r|c|c|c|}
\cline{2-2} \cline{4-4}
 2. & \myStrut{5} 
$t'=\cos\frac{5\pi}{14}$
&\spn{0.1}{3} &
$g_s^{12} = g_2^6 = g_3^2     =  e$
\\
& \myStrut{5} $b'_1 =\sqrt{250+20a-50r_1-4ar_1}$
&&
$g\vpu{-1}_a g\vpu{-1}_s g^{-1}_a  =  g^{11}_s $
\\
& \myStrut{5}  $b''_1  = i\sqrt{-250+20a-50r_1+4ar_1}$
&&
$g\vpu{-1}_a g\vpu{-1}_2 g^{-1}_a  =  g^{10}_s g\vpu{10}_2$
\\
& \myStrut{5}  $b'_2 =\re\left(  \left(908 + 132a + 12i\sqrt{4962 + 234 a}\right)^{\frac{1}{3}}\right) $
&&
$ g\vpu{-1}_a g\vpu{-1}_3 g^{-1}_a = g\vpu{-1}_3$
\\
\cline{4-4}
& \myStrut{5}  $ b''_2  = \re\left( e^{\frac{2\pi i}{3}} \left(908 -132a + 12i\sqrt{4962 - 234 a}\right)^{\frac{1}{3}}\right)$
\\
\cline{2-2}
\end{tabular}

\vspace{0.15 in}

\begin{tabular}{l|c|c|c|c|c|c|}
\cline{2-7}
4. & \myStrut{4} $G_0$ & $G_{as}$ & $G_{bs}$ & $G_{a1}=G_{b1}$ & $G_{a2}=G_{b2}$ & $G_{a3}=G_{b3}$
\\
\cline{2-7}
& \myStrut{5} $\bsmt 0 & 13 \\ 15 & 13 \esmt$ 
& $\bsmt 7 & 13 \\ 15 & 20 \esmt$ 
& $\bsmt 1 & 20 \\ 8 & 21 \esmt$
& $\bsmt 27 & 21 \\ 7 & 20\esmt$
& $\bsmt 16 & 17 \\ 11 & 5\esmt$
& $\bsmt 25 & 25 \\ 3 & 22\esmt$
\\
\cline{2-7}
& \myStrut{4} $F_0$ &  $F_{as}$ & $F_{bs}$ & $F_{a1}=F_{b1}$  & $F_{a2}=F_{b2}$ & $F_{a3}=F_{b3}$
\\
\cline{2-7}
& \myStrut{5}
$\bsmt 0 & 13 \\ 15 & 13 \esmt$ &  
$\bsmt 20 & 15 \\ 13 & 7\esmt$ & 
$\bsmt 21 & 8 \\ 20 & 1\esmt$  &
$\bsmt 20 & 7 \\ 21 & 27\esmt$ &
$\bsmt 1 & 19 \\ 17 & 16\esmt$ &
$\bsmt 22 & 3 \\ 25 & 25 \esmt$
\\
\cline{2-7}
\end{tabular}

\vspace{0.15 in}

\begin{tabular}{l|c|c|c|c|c|}
\cline{2-4}\cline{6-6}
5. & \multicolumn{2}{c|}{$|\tstb{\Pi_a}|,|\tstb{\Pi_b}|,|\bstb{\Pi_a}|,|\bstb{\Pi_b}|$} & \myStrut{4}$|\cbstb{\Pi_a}|,|\cbstb{\Pi_b}|$
&\spn{0.1}{6} &
$G_0^{6} = I \quad G_{a1}^{6}  = G_{b1}^{6}  =I$
\\
\cline{2-4}
& \multicolumn{2}{c|}{$6$} & \myStrut{4} $432$
&&
$G_{a2}^6 =G_{b2}^6 =I \qquad  G_{a3}^2 =G_{b3}^2 =G_0^4$
\\
\cline{2-4}\cline{6-6}
\end{tabular}

\vspace{0.2 in}

\textbf{Orbit $15d$} 

\nopagebreak

\vspace{0.15 in}

\nopagebreak

\begin{tabular}{r|c|c|c|c|c|c|c|}
\cline{2-8}
1. & \myStrut{4} & $a$ & $r_1$ & $t$ &  $b_1$ & $b_2$ &  $i$
\\
\cline{2-8}
&\myStrut{4} $g_a$ & $-a$ & $r_1$  &$t'$ &$b'_1$  & $b_2$  &  $i$
\\
\cline{2-8}
&\myStrut{4} $g_1$ &  $a$ & $-r_1$  &$t''$ &$b''_1$  & $b'_2$  &  $i$
\\
\cline{2-8}
&\myStrut{4} $\bar{g}_1$ &  $a$ & $r_1$  &$t$ &$b_1$  & $b_2$  &  $-i$
\\
\cline{2-8}
& \multicolumn{7}{|c |}{\myStrut{5} $  g\vpu{12}_c = g^{12}_1\bar{g}\vpu{12}_1$} 
\\
\cline{2-8}
\end{tabular}

\vspace{0.1 in}

\begin{tabular}{r|c|c|c|}
\cline{2-2} \cline{4-4}
 2. & \myStrut{5} 
$t'=\cos\frac{11\pi}{15}$
&\spn{0.1}{3} &
$g_1^{24}     =  e$
\\
 & \myStrut{5} 
$t''=\cos\frac{7\pi}{15}$
&&
$g\vpu{-1}_a g\vpu{-1}_1 g^{-1}_a  =  g^{5}_1 $
\\
\cline{4-4}
& \myStrut{5} $b'_1  = \sqrt{30 + 28 a + 6 r_1 + 4 a r_1 + (60 + 24 a + 12 r_1 + 8 a r_1) t'}$
\\
& \myStrut{5}  $b''_1  = i\sqrt{-30 + 28 a + 6 r_1 - 4 a r_1 + (-60 + 24 a + 12 r_1 - 8 a r_1) t''}$
\\
& \myStrut{5}  $b'_2 = 2\re\Bigl(  e^{\frac{2\pi i}{3}}(10+30i)^{\frac{1}{3}}\Bigr)$
\\
\cline{2-2}
\end{tabular}

\vspace{0.15 in}

\begin{tabular}{r|c|c|c|c|c|c|c|c|}
\cline{2-5} \cline{7-9}
4. & \myStrut{4} $G_0$ & $G_{1}$ & $F_0$ & $F_1$ 
&\spn{3.3}{5}&
\multicolumn{2}{c|}{$|\tstb{\Pi}|$, $|\bstb{\Pi}|$} & $|\cbstb{\Pi}|$
\\
\cline{2-5} \cline{7-9}
& \myStrut{5}
$\bsmt 4 & 11 \\ 4 & 0 \esmt$ 
& $\bsmt 4 & 8 \\ 7 & 12 \esmt$
& $\bsmt 4 & 11 \\ 4 & 0 \esmt$ 
& $\bsmt 6 & 1 \\ 8 & 4\esmt$
&&
 \multicolumn{2}{c|}{$6$} & $144$
\\
\cline{2-5} \cline{7-9}
\end{tabular}

\vspace{0.1 in}

\begin{tabular}{r|c|c|c|}
\cline{2-2} \cline{4-4}
6. & \myStrut{5} $G_0^{6} = I \quad G_{1}^{24}  = I$ 
&\spn{3.72}{7}&
$\rv_{1} = \bar{\rv}_{1} = \boldsymbol{0} \quad \qv_{1}  =\bar{\qv}_{1} =\boldsymbol{0}$
\\
\cline{2-2} \cline{4-4}
\end{tabular}

\vspace{0.2 in}

 \textbf{Orbits $16a,b$} 

\nopagebreak

\vspace{0.15 in}

\nopagebreak

\begin{tabular}{l|c |c|c|c|c|c|c|c|c|c|}
\cline{2-11}
1. &\myStrut{4} & $a$ & $r_1$ & $r_2$ & $t$ &  $b_1$ & $b_2$ & $b_3$ &  $b_4$ & $i$
\\
\cline{2-11}
&\myStrut{4} $g_a$ & $-a$ & $r_1$  &$r_2$ &$t$   &  $b'_1$ & $b'_2$ & $b'_3$ & $b'_4$ & $i$
\\
\cline{2-11}
&\myStrut{4} $g_s$ &  $a$ & $-r_1$  &$r_2$ &$t$   &  $b_1$ & $b'_2$ & $b''_3$ & $b''_4$ & $i$
\\
\cline{2-11}
&\myStrut{4} $g_2$ &  $a$ & $r_1$  & $-r_2$ &$t'$   &  $b_1$ & $b_2$ & $b_3$ & $b_4$ & $i$
\\
\cline{2-11}
&\myStrut{4} $g_3$ &   $a$ & $r_1$  & $r_2$ &$t$   &  $-b_1$ & $b_2$ & $b_3$ & $b_4$ & $i$
\\
\cline{2-11}
&\myStrut{4} $\bar{g}_1$ &  $a$ & $r_1$  & $r_2$ &$t$   &  $b_1$ & $b_2$ & $b_3$ & $b_4$ & $-i$ 
\\
\cline{2-11}
& \multicolumn{10}{|c |}{\myStrut{5} $ g\vpu{2}_1 = g^2_s \qquad g\vpu{2}_c = g\vpu{2}_3\bar{g}\vpu{2}_1$} 
\\
\cline{2-11}
\end{tabular}

\vspace{0.15 in}

\begin{tabular}{l|l|}
\cline{2-2}
 2. & \myStrut{4} 
$t'=\cos\frac{3\pi}{16}$
\\
& \myStrut{4} $b'_1  = \sqrt{a+11}$
\\
& \myStrut{4}  $b'_2  = \sqrt{-a r_1+195}$
\\
& \myStrut{4}  $b'_3  = \sqrt{195+3ar_1+(a+4r_1)b'_2}$
\\
& \myStrut{4}  $b''_3 = \sqrt{195+3ar_1-(a+4r_1)b'_2}$
\\
& \myStrut{4}  $b'_4  = \sqrt{29744-2704r_1 -\bigl(936-936r_1+16 ar_1+(65-5a-3r_1+3ar_1)b'_2\bigr)b'_3}$
\\
& \myStrut{4}  $ b''_4  = \sqrt{29744+2704r_1-(936+936r_1+16ar_1+(65+5a+3r_1+3ar_1)b'_2)b''_3}$
\\
\cline{2-2}
\end{tabular}

\vspace{0.15 in}

\begin{tabular}{l|c|}
\cline{2-2}
3. & \myStrut{4} $ g_s^{16} = g_2^8 = g_3^2     =  e \quad g\vpu{-1}_a g\vpu{-1}_s g^{-1}_a  =  g^{15}_s \quad g\vpu{-1}_a g\vpu{-1}_2 g^{-1}_a  =  g^{8}_s g\vpu{10}_2 \quad g\vpu{-1}_a g\vpu{-1}_3 g^{-1}_a = g\vpu{-1}_3$
\\
\cline{2-2}
\end{tabular}

\vspace{0.15 in}

\begin{tabular}{l|c|c|c|c|c|c|c|c|}
\cline{2-9}
4. & \myStrut{4} $G_0$ & $G_{as}$ & $G_{bs}$ & $G_{a1}$ & $G_{b1}$ & $G_{a2}$ & $G_{b2}$ & $G_{a3}=G_{b3}$
\\
\cline{2-9}
& \myStrut{5} 
 $\bsmt 0 & 15 \\ 17 & 15 \esmt$ 
& $\bsmt 17 & 21 \\ 6 & 15 \esmt$ 
& $\bsmt 13 & 27 \\ 8 & 19 \esmt$ 
& $\bsmt 11 & 15\\ 17 & 26 \esmt$ 
& $\bsmt 26 & 17\\ 15 & 11 \esmt$ 
&$\bsmt 21 & 29 \\ 3 & 18 \esmt$ 
& $\bsmt 19 & 27 \\ 5 & 14 \esmt$ 
& $\bsmt 31 & 0 \\ 0 & 31 \esmt$ 
\\
\cline{2-9}
& \myStrut{4} $F_0$ &  $F_{as}$ & $F_{bs}$ & $F_{a1}$ & $F_{b1}$  & $F_{a2}$ &$F_{b2}$ & $F_{a3}=F_{b3}$
\\
\cline{2-9}
& \myStrut{5}
$\bsmt 0 & 15 \\ 17 & 15 \esmt$ &  
$\bsmt 15 & 11 \\ 26 & 17\esmt$ &  
$\bsmt 19 & 5 \\ 24 & 13\esmt$ & 
$\bsmt 26 & 17 \\ 15 & 11\esmt$ & 
$\bsmt 11 & 15 \\ 17 & 26\esmt$ &
$\bsmt 26 & 31 \\ 29 & 21 \esmt$ &  
$\bsmt 6 & 9 \\ 27 & 19\esmt$ & 
$\bsmt 31 & 0 \\ 0 & 31\esmt$
\\
\cline{2-9}
\end{tabular}

\vspace{0.15 in}

\begin{tabular}{l|c|c|c|c|c|}
\cline{2-4}\cline{6-6}
5. & \multicolumn{2}{c|}{$|\tstb{\Pi_a}|,|\tstb{\Pi_b}|, |\bstb{\Pi_a}|,|\bstb{\Pi_b}|$} &\myStrut{4}  $|\cbstb{\Pi_a}|,|\cbstb{\Pi_b}|$
&\spn{0.1}{6}& 
$G_0^{6} = I \quad G_{a1}^{8}  = G\vpu{8}_0 \quad  G_{b1}^{8}  =G^5_0$
\\
\cline{2-4}
& \multicolumn{2}{c|}{$6$} & \myStrut{4}  $768$
&&
$G_{a2}^8 =G_{b2}^8 =G\vpu{8}_0 \quad  G_{a3}^2 =G_{b3}^2 =I$
\\
\cline{2-4}\cline{6-6}
\end{tabular}

\vspace{0.2 in}

\textbf{Orbit $19e$} 

\nopagebreak

\vspace{0.15 in}

\nopagebreak

\begin{tabular}{r|c|c|c|c|c|c|c|c|c|c|c|c|c|}
\cline{2-7} \cline{9-14}
1. & \myStrut{4} & $a$ & $r_1$ & $t$ &  $b_1$ &  $i$
&\spn{0.1}{2} & \multicolumn{6}{c|}{$t'=\cos \frac{3\pi}{19} \quad b'_1  = \sqrt{2a-1}$}
\\
\cline{2-7} \cline{9-14}
& \myStrut{4} $g_a$ & $-a$ & $r_1$  &$t$ &$b'_1$    &  $i$
\\
\cline{2-7} \cline{9-14}
& \myStrut{4} $g_1$ &  $a$ & $r_1$  &$t'$ &$-b_1$    &  $i$
& \spn{0.1}{3} & \multicolumn{6}{c|}{$ g_1^{18}     =  e \quad g\vpu{-1}_a g\vpu{-1}_1 g^{-1}_a  =  g\vpu{-1}_1$}
\\
\cline{2-7} \cline{9-14}
& \myStrut{4} $\bar{g}_1$ &  $a$ & $r_1$  &$t$ &$b_1$    &  $-i$
\\
\cline{2-7} \cline{9-12}
& \myStrut{4} $\bar{g}_2$ &   $a$ & $-r_1$  &$t$ &$b_1$    &  $-i$
&\spn{0.1}{4} &
$G_0$ & $G_{1}$ & $F_0$ & $F_1$
\\
\cline{2-7} \cline{9-12}
& \multicolumn{6}{|c |}{\myStrut{5} $g\vpu{9}_c = g^9_1\bar{g}\vpu{9}_1$} 
&&
$\bsmt 16 & 7 \\ 12 & 4 \esmt$ &
$\bsmt 18 & 9 \\ 10 & 8 \esmt$ &
$\bsmt 3 & 12 \\ 7 & 15 \esmt$ &  
$\bsmt 10 & 3 \\ 9 & 18\esmt$ 
\\
\cline{2-7} \cline{9-12}
\end{tabular}

\vspace{0.1 in}

\begin{tabular}{r|c|c|c|c|c|}
\cline{2-4} \cline{6-6}
5. & \multicolumn{2}{c|}{$|\tstb{\Pi}|, |\bstb{\Pi}|$} & \myStrut{4} $|\cbstb{\Pi}|$
& \spn{01.7}{6} &
 $G_0^{18} = I \qquad G_{1}^{18}  = I$
\\
\cline{2-4} \cline{6-6}
&  \multicolumn{2}{c|}{$18$} & \myStrut{4} $324$
\\
\cline{2-4}
\end{tabular}

\vspace{0.2 in}

\textbf{Orbit $24c$} 

\nopagebreak

\vspace{0.15 in}

\nopagebreak

\begin{tabular}{r|c|c|c|c|c|c|c|c|c|}
\cline{2-10}
1. &\myStrut{4} & $a$ & $r_1$ & $r_2$ & $t$ &  $b_1$ & $b_2$ & $b_3$ & $i$
\\
\cline{2-10}
&\myStrut{4} $g_a$ & $-a$ & $r_1$  &$r_2$ &$t$    &  $b'_1$ &  $b_2$ &  $b_3$ &  $i$
\\
\cline{2-10}
&\myStrut{4} $g_1$ &  $a$ & $r_1$  &$-r_2$ &$t'$    &  $b_1$ &  $b'_2$ &  $b'_3$ &  $i$
\\
\cline{2-10}
&\myStrut{4} $g_2$ &  $a$ & $-r_1$  &$r_2$ &$t''$    &  $b_1$ &  $b''_2$ &  $b_3$ &  $i$
\\
\cline{2-10}
&\myStrut{4} $g_3$ &   $a$ & $r_1$  &$r_2$ &$t$    &  $-b_1$ &  $b_2$ &  $b_3$ &  $i$
\\
\cline{2-10}
&\myStrut{4} $\bar{g}_1$ &  $a$ & $r_1$  &$r_2$ &$t$    &  $b_1$ &  $b_2$ &  $b_3$ &  $-i$
\\
\cline{2-10}
& \multicolumn{9}{|c |}{\myStrut{5} $g_c = g_3\bar{g}_1$} 
\\
\cline{2-10}
\end{tabular}

\vspace{4 ex}

\begin{tabular}{r|c|c|c|}
\cline{2-2} \cline{4-4}
 2. & \myStrut{4} 
$t'=  \cos \frac{7\pi }{24} \quad t'' = \cos \frac{11 \pi }{24} \quad b'_1  =\sqrt{a+3}$
& \spn{0.7}{3} & 
$g_1^{12} = g_2^4 = g_3^2     =  e$
\\
 & \myStrut{5.5} 
$b'_2 = \sqrt{(4+r_1)(3-r_2)} \quad b''_2 = \sqrt{(4-r_1)(3+r_2)}$
&&
$g\vpu{-1}_a g\vpu{-1}_1 g^{-1}_a  =  g^5_1 g\vpu{3}_3$
\\
& $b'_3 = \frac{4}{3}\re\left(e^{\frac{2\pi i}{3}} (27+9 i a r_2)^{\frac{1}{3}} \right)$
&&
$g\vpu{-1}_a g\vpu{-1}_2 g^{-1}_a = g_1^6 g\vpu{6}_2$
\\
\cline{2-2}
\multicolumn{3}{c|}{} & \myStrut{4} $g\vpu{-1}_a g\vpu{-1}_3 g^{-1}_a =g\vpu{-1}_3$
\\
\cline{4-4}
\end{tabular}

\begin{tabular}{l|c|c|c|c|c|c|c|c|c|c|c|}
\cline{2-5} 
4. & \myStrut{4} $G_0$ & $G_{1}$ & $G_2$ & $G_3$ 
\\
\cline{2-5} \cline{7-9}
& \myStrut{5}
$\bsmt 24 & 5 \\ 43 & 29 \esmt$ 
& $\bsmt 2 & 13 \\ 35 & 15 \esmt$ 
& $\bsmt 47 & 12 \\ 36 & 11 \esmt$
& $\bsmt 47 & 0 \\ 0 & 47 \esmt$ 
& \spn{2}{5}  &
\multicolumn{2}{c|}{$|\tstb{\Pi}|$, $|\bstb{\Pi}|$} & $|\cbstb{\Pi}|$
\\
\cline{2-5} \cline{7-9}
& \myStrut{4} $F_0$ & $F_1$ &$F_2$ & $F_3$
&&
\multicolumn{2}{c|}{ $12$} & $1152$
\\
\cline{2-5} \cline{7-9}
& \myStrut{5} 
$\bsmt 24 & 5 \\ 43 & 29\esmt$ 
&$\bsmt 39 & 43 \\ 13 & 2\esmt$  
&$\bsmt 47 & 36 \\ 12 & 47\esmt$   
&$\bsmt 47 & 0 \\ 0 & 47\esmt$ 
\\
\cline{2-5}
\end{tabular}

\begin{tabular}{l|c|c|c|c|c|c|c|c|c|c|c|c|c|c|c|c|c}
 \cline{11-11}
\multicolumn{10}{r|}{6.} &
\myStrut{4} $G_0^{12} = I \quad G_{1}^{12}  = G_0^{6} $
\\
\cline{2-9}
7. & \myStrut{4}  $\rv_{1}$ &  $\rv_{2}$ & $\rv_{3}$ & $\bar{\rv}_{1}$ &$\qv_{1}$ &  $\qv_{2}$ & $\qv_{3}$ &  $\bar{\qv}_{1}$
& \multicolumn{1}{c|}{\te{0.47}{$\phantom{x}$}} & 
\myStrut{4} $ G_2^4  = I \quad G_3^2 = I$
\\
\cline{2-9}\cline{11-11}
& \myStrut{5} $\bsmt 1 \\2\esmt $ & $\bsmt 2 \\1 \esmt$ &
$\bsmt 2 \\ 1 \esmt $ & $\bsmt 1 \\ 2\esmt$ & $\bsmt 8\\ 8 \esmt$ & $\bsmt 8 \\ 16\esmt$ & $\bsmt 8 \\ 16\esmt $ & $\bsmt 8\\ 8\esmt $
\\
\cline{2-9}
\end{tabular}

\vspace{0.2 in }

\textbf{Orbit $28c$} 

\nopagebreak

\vspace{0.15 in}

\nopagebreak

\begin{tabular}{l|c|c|c|c|c|c|c|c|c|}
\cline{2-10}
1. & \myStrut{4} & $a$ & $r_1$ & $r_2$ & $t$ &  $b_1$ & $b_2$ & $b_3$ & $i$
\\
\cline{2-10}
&\myStrut{4} $g_a$ & $-a$ & $r_1$  &$r_2$ &$t$    &  $b'_1$ &  $b'_2$ &  $b_3$ &  $i$
\\
\cline{2-10}
&\myStrut{4} $g_1$ &  $a$ & $-r_1$  &$r_2$ &$t'$    &  $b_1$ &  $b_2$ &  $b_3$ &  $i$
\\
\cline{2-10}
&\myStrut{4} $g_2$ &  $a$ & $r_1$  &$r_2$ &$t$    &  $b_1$ &  $-b_2$ &  $b'_3$ &  $i$
\\
\cline{2-10}
&\myStrut{4} $g_3$ &  $a$ & $r_1$  &$-r_2$ &$t''$    &  $b_1$ &  $b_2$ &  $b_3$ &  $i$
\\
\cline{2-10}
&\myStrut{4} $g_4$ &  $a$ & $r_1$  &$r_2$ &$t$    &  $-b_1$ &  $-b_2$ &  $b_3$ &  $i$
\\
\cline{2-10}
&\myStrut{4} $\bar{g}_1$ &  $a$ & $r_1$  &$r_2$ &$t$    &  $b_1$ &  $b_2$ &  $b_3$ &  $-i$
\\
\cline{2-10}
& \multicolumn{9}{|c |}{\myStrut{5} $g_c = g_4\bar{g}_1$} 
\\
\cline{2-10}
\end{tabular}

\vspace{0.1 in}

\begin{tabular}{r|c|c|c|}
\cline{2-2}
 2. & \myStrut{4} 
$t'=\cos \frac{3\pi}{28} \quad t''=\cos \frac{15\pi}{28} \quad b'_1  = \sqrt{a-1}$
\\
& $ b'_2 = \sqrt{a-5} \quad b'_3 = 2\re\biggl( e^{\frac{2\pi i}{3}}\left( 189+21i a\sqrt{3} \right)^{\frac{1}{3}} \biggr)$
\\
\cline{2-2}
\end{tabular}

\vspace{0.1 in}

\begin{tabular}{l|c|}
\cline{2-2}
3. & \myStrut{4} $ g_1^{6} = g_2^6 = g_3^2   =g_4^2  =  e$
\\
&\myStrut{4} $g\vpu{-1}_a g\vpu{-1}_1 g^{-1}_a  =  g\vpu{-1}_1 \qquad g\vpu{-1}_a g\vpu{-1}_2 g^{-1}_a = g_2^5  \qquad  g\vpu{-1}_a g\vpu{-1}_3 g^{-1}_a =g_2^3 g\vpu{3}_3 g\vpu{3}_4 \qquad g\vpu{-1}_a g\vpu{-1}_4 g^{-1}_a =g\vpu{-1}_4$
\\
\cline{2-2}
\end{tabular}%
 }
\end{document}